\documentclass[lettersize,journal]{IEEEtran}
\usepackage{algorithmic}
\usepackage{algorithm}
\usepackage{array}
\usepackage[caption=false,font=small,labelfont=sf,textfont=sf]{subfig}
\usepackage{textcomp}
\usepackage{stfloats}
\usepackage{url}
\usepackage{verbatim}
\usepackage{graphicx}
\usepackage{cite}
\usepackage{widetext}
\usepackage{makecell}
\usepackage{color}
\usepackage{multirow}
\usepackage{booktabs}   
\usepackage{tabularx}   
\usepackage{enumitem}   
\usepackage{ragged2e}   
\usepackage[table,xcdraw]{xcolor}
\definecolor{lightblue}{rgb}{0.9,1, 1}
\usepackage{amsmath} 
\usepackage{amsfonts}
\usepackage{amssymb} 
\usepackage{amsthm}  

\newtheorem{proposition}{Proposition}

\usepackage{algorithm}

\begin{document}
\title{Compositional Zero-Shot Recognition based on Tangent Space Disentanglement for Composite Modulation Signals}
\author{Yurui Zhao, Xiang Wang, Zhitao Huang, Baoguo Li
\thanks{The authors are with College of Electronic Science and Technology, National University of Defense Technology, Changsha 410073, China.}
\thanks{Corresponding author: Xiang Wang (E-mail: xwang@nudt.edu.cn)}
\thanks{Research is supported by the National Natural Science Foundation of China, Grant No.62271494.}
\thanks{Manuscript received April 20, 2024.}}

\markboth{Journal of \LaTeX\ Class Files,~Vol.~14, No.~8, August~2021}%
{Shell \MakeLowercase{\textit{et al.}}: A Sample Article Using IEEEtran.cls for IEEE Journals}
\maketitle

\begin{abstract} 
Automatic composite modulation recognition (ACMR) is critical for integrated sensing and communication (ISAC) systems, while conventional approaches face significant challenges due to the semantic coupling between inner-layer and outer-layer modulations in composite modulation (CM), degraded performance under joint hardware and channel imperfections, and limited capability to handle unknown modulation schemes. 
To this end, we design a disentangled semantic space and propose zero-shot learning framework. 
Within this framework, a logarithmic projection first linearizes the multiplicative coupling between modulation layers and a learnable geometric transformation is used for layer-wise semantic features.
We instantiate the framework as the Tangent Space Disentanglement Network (TSDN). 
TSDN integrates logarithmic mapping, a spatial transformer network for learning the geometric transformation, and a multi-objective loss function that balances discrimination with cross-domain generalization.
Comprehensive experiments demonstrate that TSDN achieves over 93\% zero-shot recognition accuracy, outperforms unified-semantic and multi-task baselines by significant margins, and maintains robust performance under combined channel fading and hardware imperfections down to 4 dB SNR. 
\end{abstract}


\section{Introduction}

\begin{figure*}[!t]
    \centering
    \includegraphics[width=0.8\linewidth]{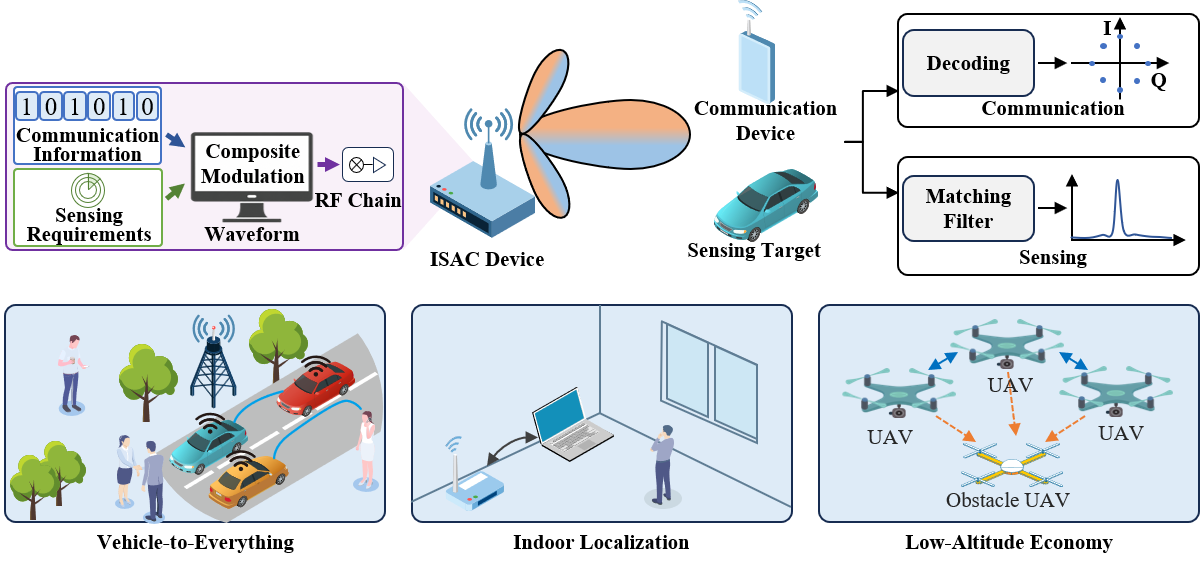}
    \caption{Schematic of the ISAC system showing the role of composite modulation (CM) in enabling sensing and communication, with applications in vehicle-to-everything, indoor localization, and low-altitude economy.}
    \label{fig:scene}
\end{figure*}

\IEEEPARstart{I}{ntegrated} Sensing and Communication (ISAC) unifies sensing and communication on a shared platform, improving spectrum and infrastructure utilization \cite{luong2021radio, liu2022survey, zheng2019radar, zhou2022integrated}. 
It enables autonomous driving, WiFi sensing, vehicle-to-everything \cite{lei2025unified}, high-precision indoor localization, and is foundational for 6G unmanned aerial vehicles (UAVs) and the low-altitude economy.
\color{black}
The primary challenge in ISAC lies in designing a versatile waveform capable of reconciling the inherently conflicting requirements of sensing and communication.
Composite modulation (CM) has emerged as a promising solution to this challenge \cite{chiriyath2019novel, zhang2021overview, huang2025design}.

Unlike conventional approaches, CM achieves an intrinsic fusion of an inner communication layer and an outer sensing layer into an inseparable waveform, offering high degrees of freedom to optimize the sensing-communication trade-off, as presented in Fig.~\ref{fig:scene}.
For instance, 16QAM-2FSK prioritizes high data rates, while BPSK-LFM favors accurate sensing.
By adaptively selecting modulation schemes according to channel conditions and task priorities, CM empowers cognitive ISAC systems to enhance resource efficiency and robustness.
It is important to distinguish CM from compound modulation, which refers to the linear summation of independent signal sources, a topic on which substantial advances have been reported \cite{pan2020automatic, si2021intra}. 
In contrast, CM performs waveform-level fusion where components are coupled at the generation stage. 
As a result, the constituent modulations do not exist as independent temporal entities and require specialized feature representations distinct from those developed for compound signals.

However, the performance of cognitive ISAC systems depends critically on automatic composite modulation recognition (ACMR). 
 Due to the adaptive nature of transmission strategies, the modulation parameters of intercepted waveforms change dynamically with tasks and channel conditions, making ACMR a core determinant of overall system performance. 
On the communication side, accurate identification is a prerequisite for correct demodulation. Any misidentification inevitably causes severe bit-error-rate (BER) degradation and synchronization failure, ultimately compromising communication reliability.
On the sensing side, this identification enables the receiver to dynamically reconfigure its matched filters or super-resolution algorithms according to the identified radar component. This capability directly enhances range-Doppler resolution and reduces the false-alarm rate, thereby preserving sensing fidelity.
Beyond cooperative ISAC systems, ACMR is equally indispensable in non-cooperative scenarios such as spectrum monitoring, electronic intelligence (ELINT), and electronic warfare, where the modulation parameters of intercepted waveforms are entirely unknown a priori. 
In such contexts, ACMR constitutes the indispensable first step in converting raw observations into actionable signal intelligence. Consequently, robust ACMR is not merely a classification task but a foundational capability that underpins closed-loop resource management and interference-aware processing across the full operational range of advanced ISAC networks.

As a specialized category within the broader framework of automatic modulation recognition (AMR), ACMR inherits the fundamental principles of signal classification while introducing unique challenges associated with multi-layered modulation.
\color{black}
Early studies in AMR primarily adopted maximum likelihood estimation (MLE), framing modulation classification as a multiple-hypothesis testing problem \cite{huan1995likelihood}. As limitations of MLE became evident, research shifted toward feature-based methods, typically involving preprocessing, feature extraction, and classification \cite{bkassiny2012survey}.
The rapid advancement of deep learning (DL) has significantly propelled the development of the AMR field.
O’Shea first demonstrated CNN-based classification using raw IQ data \cite{o2016convolutional}, while subsequent works examined network depth, kernel size, and architectural variations such as residual, inception, and CLDNN models \cite{west2017deep, hermawan2020cnn, liu2020modulation, lei2024understanding}. 
To capture temporal dependencies, researchers introduced recurrent models, including long short-term memory (LSTM) and denoising autoencoders, which achieved strong performance under noisy conditions \cite{rajendran2018deep, ke2021real}. 
Hybrid models have further combined CNNs and LSTMs, enabling effective spatial-temporal feature learning \cite{zhang2020automatic, wang2021multidimensional, sainath2015convolutional}.

In contrast, research on ACMR remains relatively limited, with existing efforts primarily focusing on hand-crafted feature approaches or DL-based approaches. 
For hand-crafted feature approaches, phase-locked loop-based (PLL-based) demodulation with empirical phase cumulative distribution function (CDF) features \cite{lijun2021modulation} and cyclic-paw-point (CPP) features derived from cyclic spectra \cite{yan2024automatic, yan2024efficient} have been proposed, typically classified using support vector machine (SVM) or random forest. More recent work leverages DL to improve robustness, such as lightweight neural networks with CPP features \cite{yan2024automatic2}, ResNet applied to frequency-domain graphs \cite{wang2021modulation}, and multi-dimensional time-frequency representations \cite{wang2025blind}. Transfer learning has also been adopted to address few-shot recognition scenarios \cite{li2021composite}.

Despite these advances, ACMR research faces two fundamental challenges. 
(1) Existing methods perform reasonably well on low-order composite modulations such as BPSK-FM, but struggle to scale to high-order modulations such as 128QAM-8FSK. 
(2) Most approaches are designed for predefined modulation combinations and lack generalization to unseen CM schemes. 
This limitation restricts their applicability in practical cognitive ISAC systems that require adaptability and robustness.

\color{black}
To address these challenges, this paper proposes a novel zero-shot recognition framework for CM signals. 
The core of the framework is a disentangled semantic space constructed through tangent space disentanglement, where the multiplicative coupling between modulation layers is first linearized via logarithmic projection and then decomposed into layer-wise semantic features by a learnable geometric transformation. 
We instantiate this framework as the tangent space disentanglement network (TSDN). 
The network integrates logarithmic mapping, a spatial transformer network for learning the geometric transformation, and a multi-objective loss that balances discrimination with cross-domain generalization.
The main contributions of this work are as follows.

\begin{itemize}

\item We design a disentangled semantic space by mapping CM signal waveforms into a tangent space. This mapping converts the inherent multiplicative coupling between modulation layers into an additive form. The linearized representation allows a principled separation of the composite modulation through linear operations, providing a systematic basis for feature disentanglement.

\item We propose a zero-shot CM recognition framework built upon a disentangled semantic space. 
To disentangle inner-layer and outer-layer modulations, the framework assigns a dedicated learnable transformation module to each layer and employs a multi-task learning mechanism to extract discriminative semantic features from the tangent space. 
As a result, the framework generalizes well to previously unseen CM schemes and offers enhanced scalability for practical ISAC systems.

\item We instantiate the framework as the TSDN. TSDN integrates logarithmic mapping to project signals into the tangent space, a spatial transformer network to learn the geometric transformation for layer-wise decoupling, and a multi-objective loss function that jointly optimizes modulation discrimination and cross-domain generalization. The architecture adopts lightweight design principles, resolving the critical trade-off between model complexity and recognition performance in practical communication systems.

\end{itemize}

The remainder of this paper is organized as follows. Section II presents the CM signal model and formally defines the zero-shot ACMR task. Section III introduces the proposed disentangled semantic space and describes the overall zero-shot recognition framework built upon it. Section IV details the TSDN architecture. 
Section V reports experimental results and analysis that validate the feasibility and advantages of the method. Finally, Section VI concludes the paper.

\color{black}

\section{Mathematical Model}
\label{sec:system_model}
\subsection{Composite Modulation Primer}

\begin{figure}
    \centering
    \includegraphics[width=\linewidth]{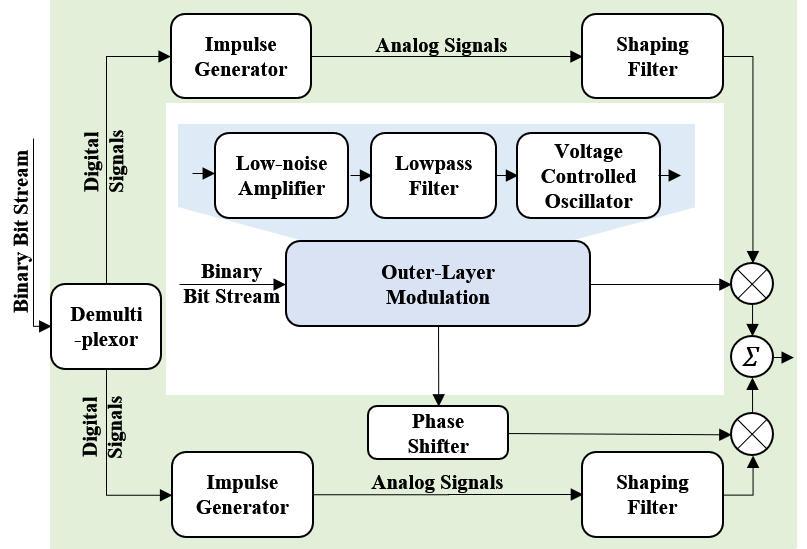}
    \caption{Structure of an ISAC emitter with CM modulator contains two-layer modulations, containing the inner-layer modulation (green) and the outer-layer modulation (blue).}
    \label{fig:sig_gen}
\end{figure}

\begin{figure*}[!t]
\centering
\subfloat[BPSK]{\includegraphics[width=0.32\linewidth]{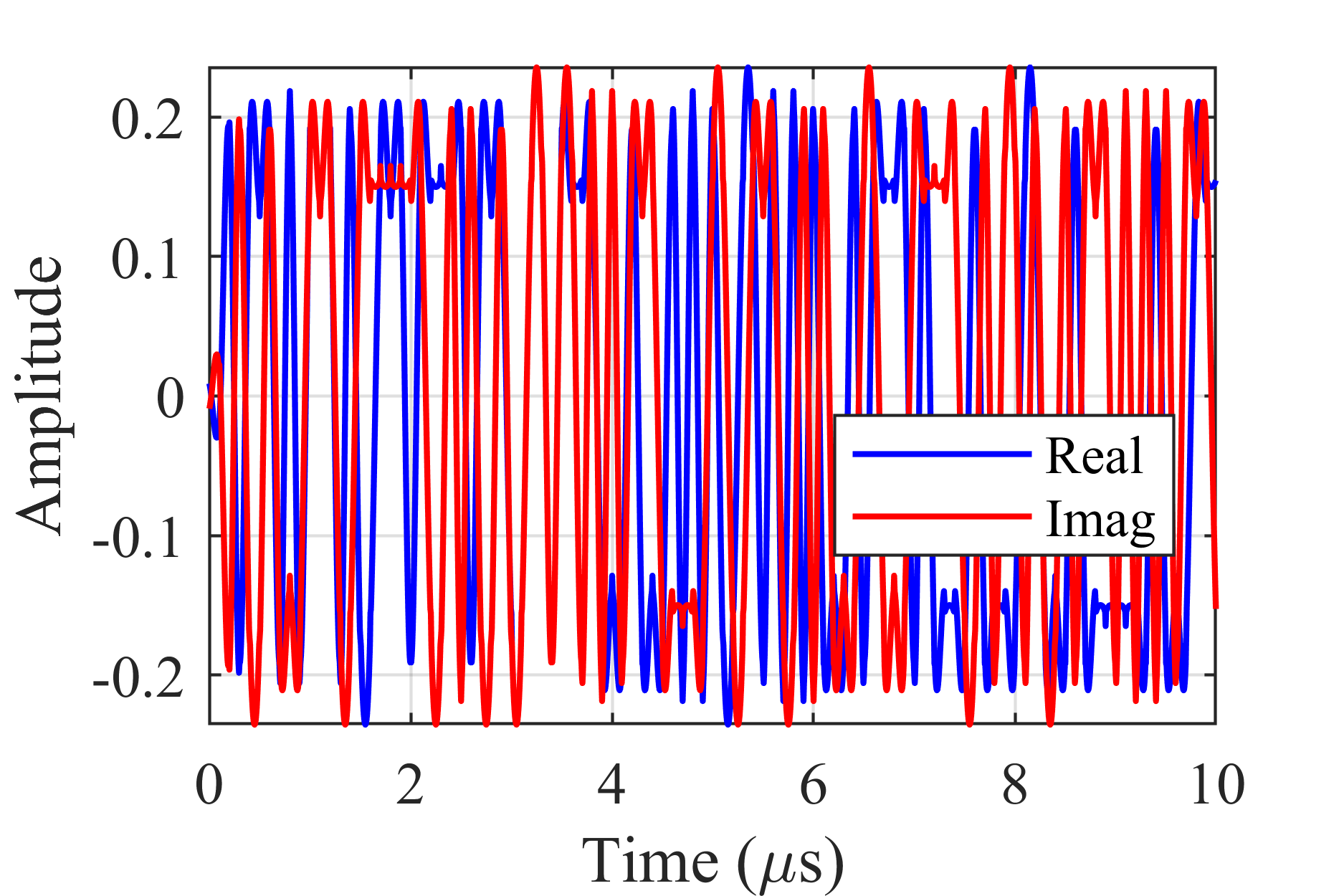}}
\subfloat[LFM]{\includegraphics[width=0.32\linewidth]{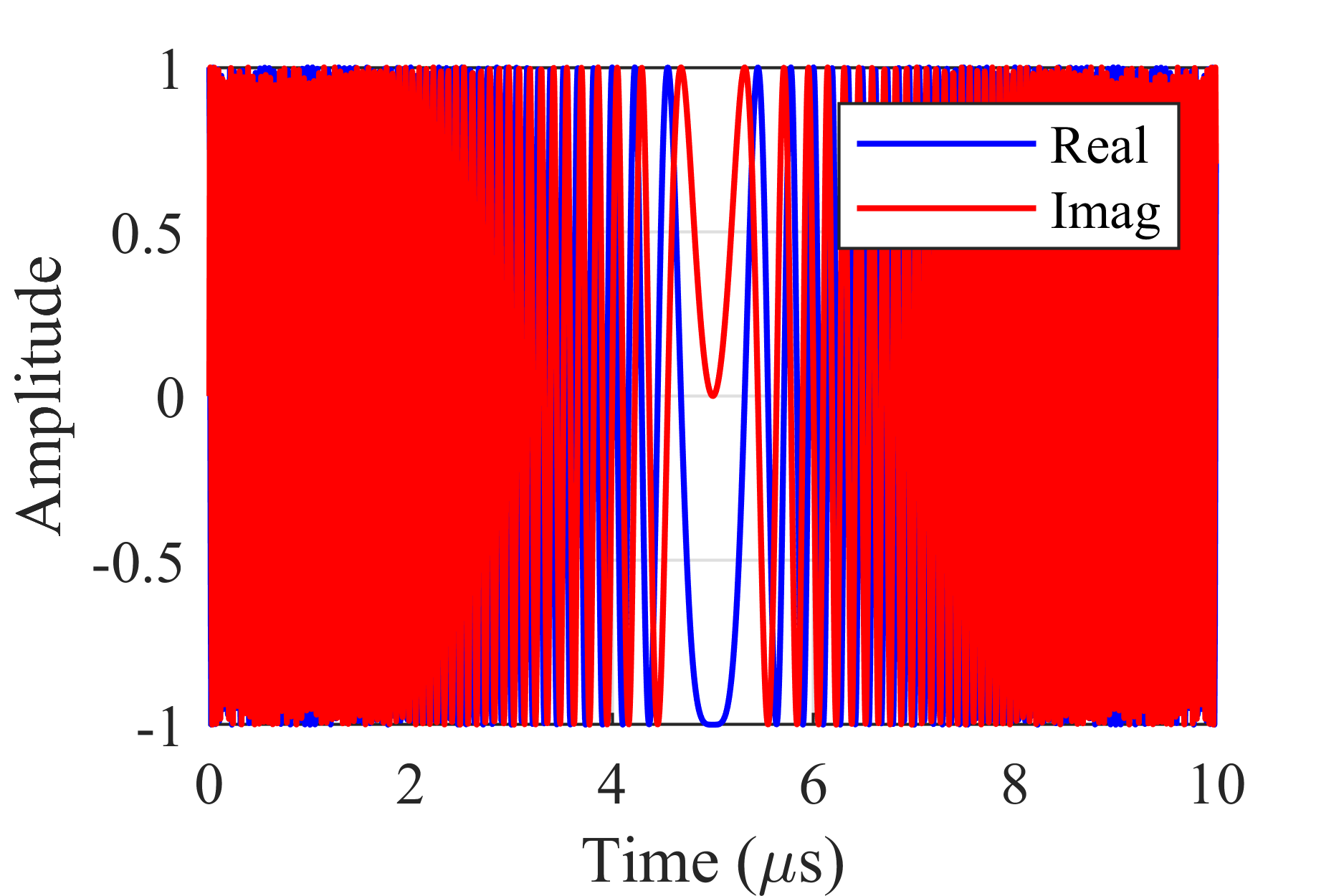}}
\subfloat[BPSK-LFM]{\includegraphics[width=0.32\linewidth]{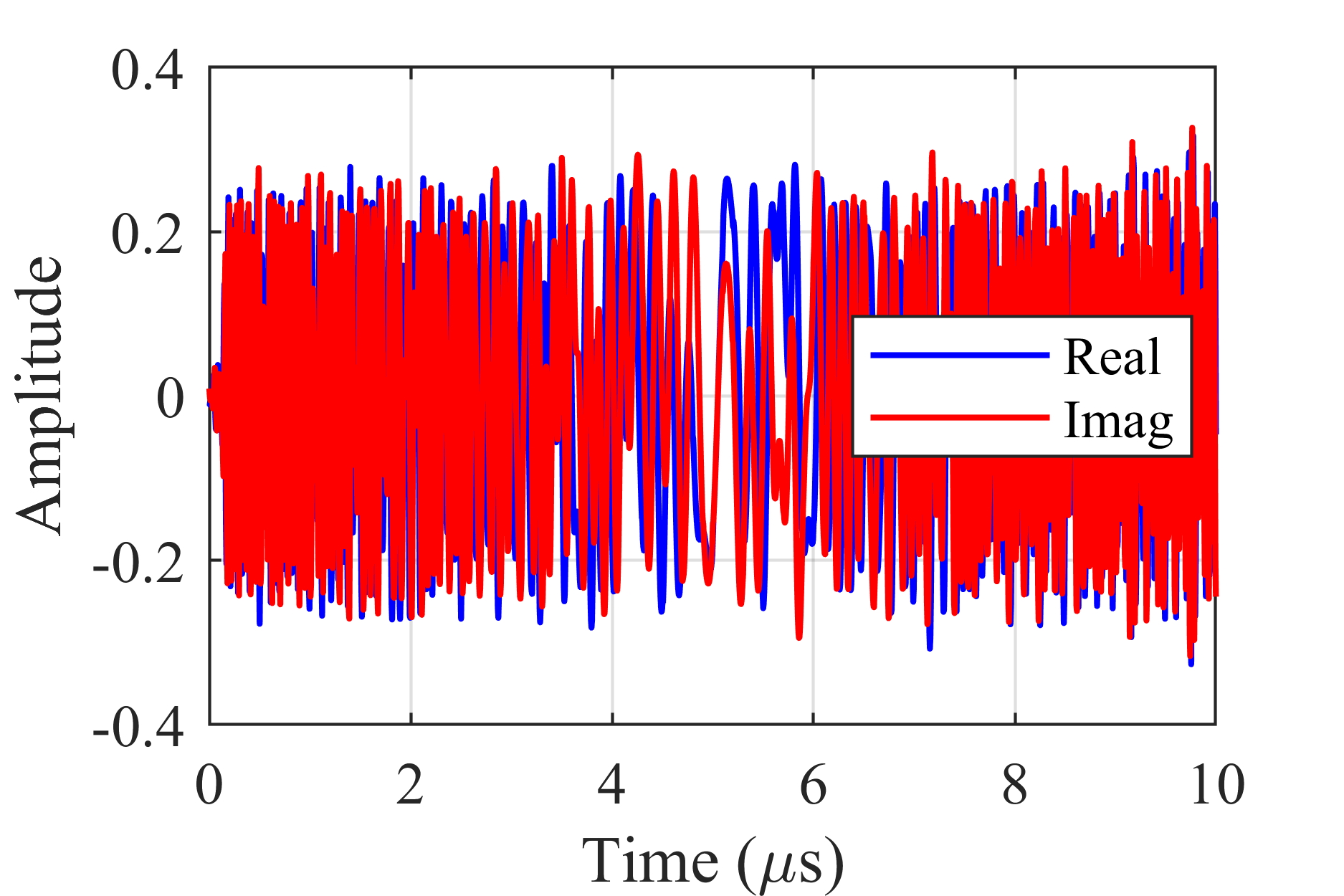}}
\\
\subfloat[BPSK]{\includegraphics[width=0.32\linewidth]{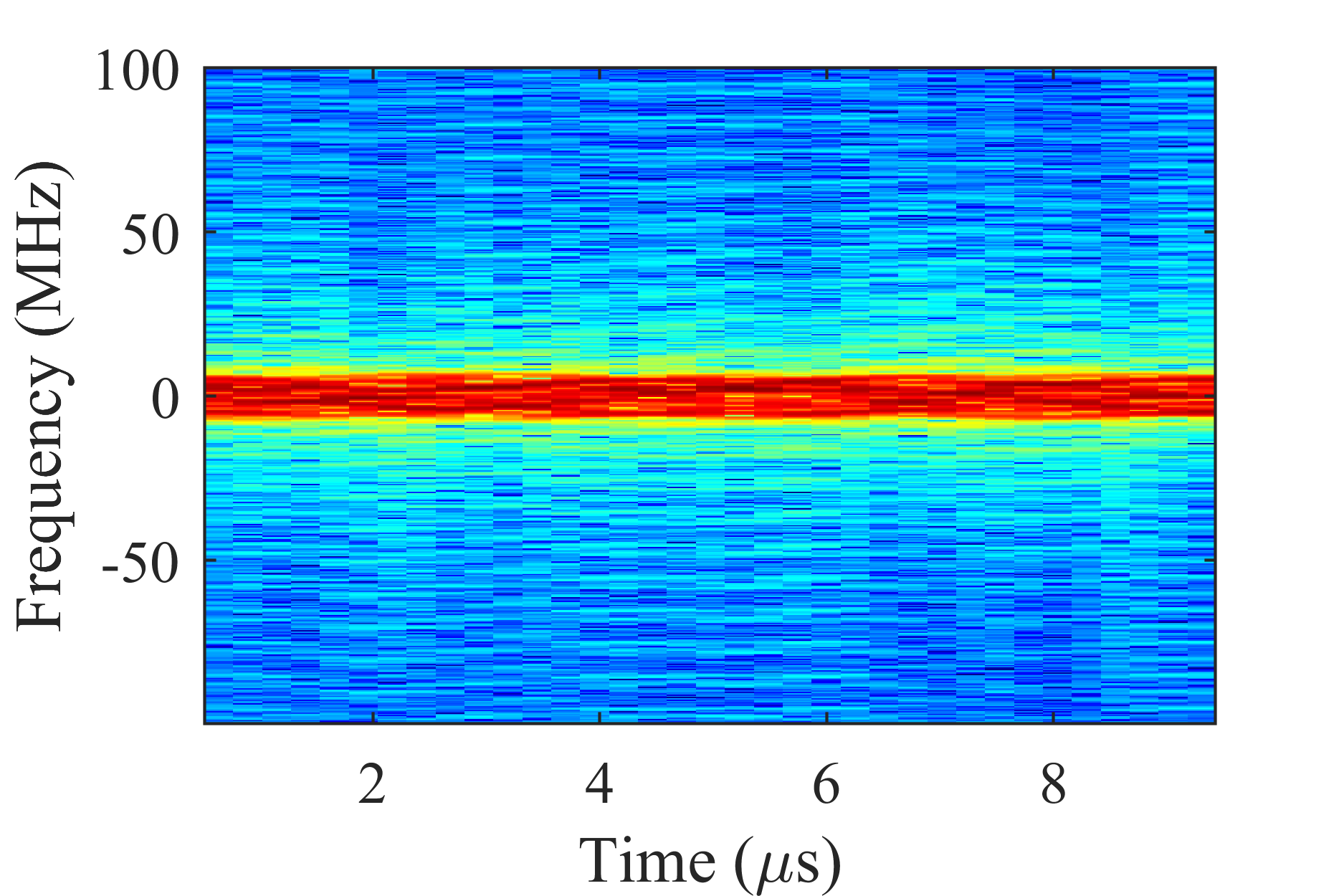}}
\subfloat[LFM]{\includegraphics[width=0.32\linewidth]{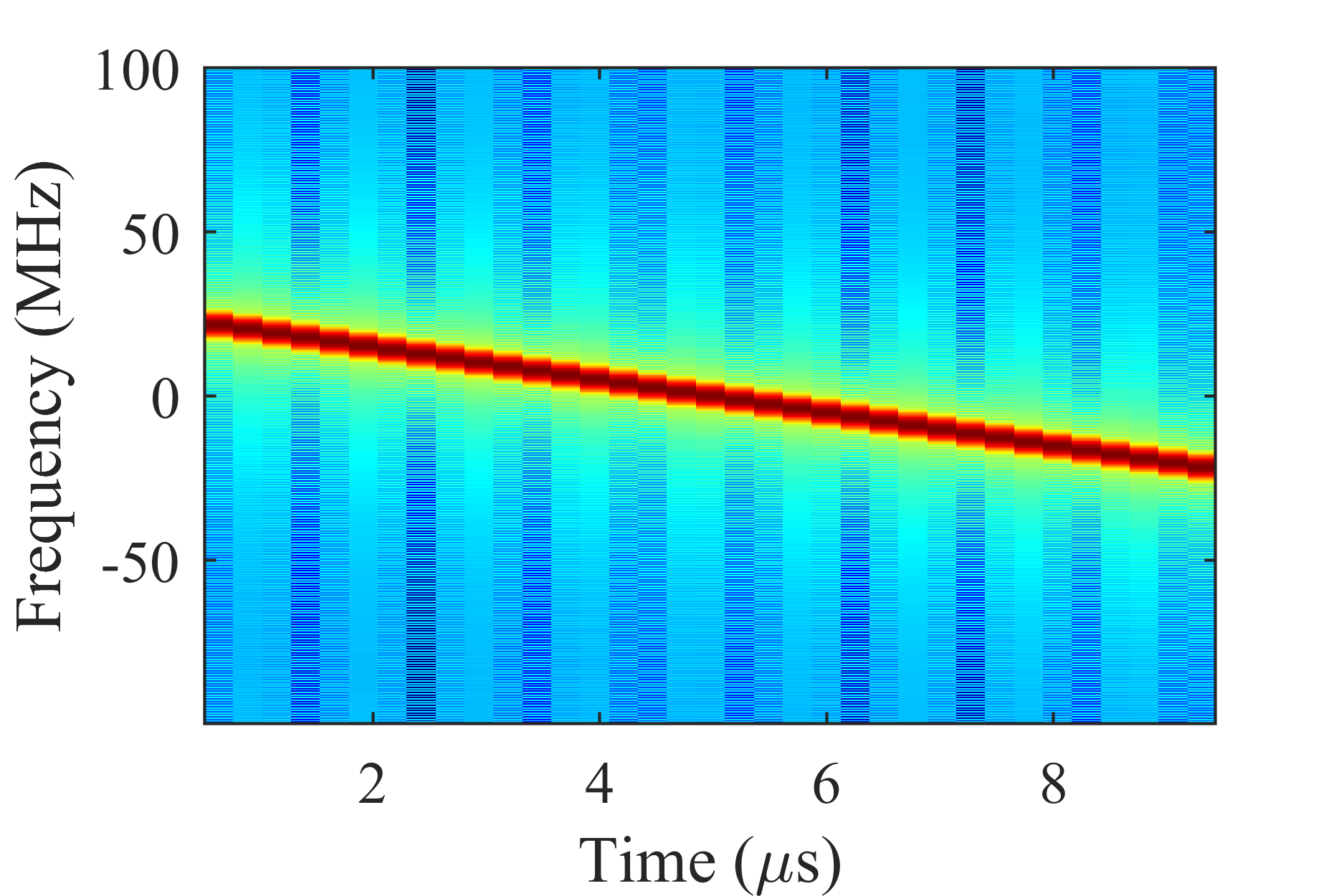}}
\subfloat[BPSK-LFM]{\includegraphics[width=0.32\linewidth]{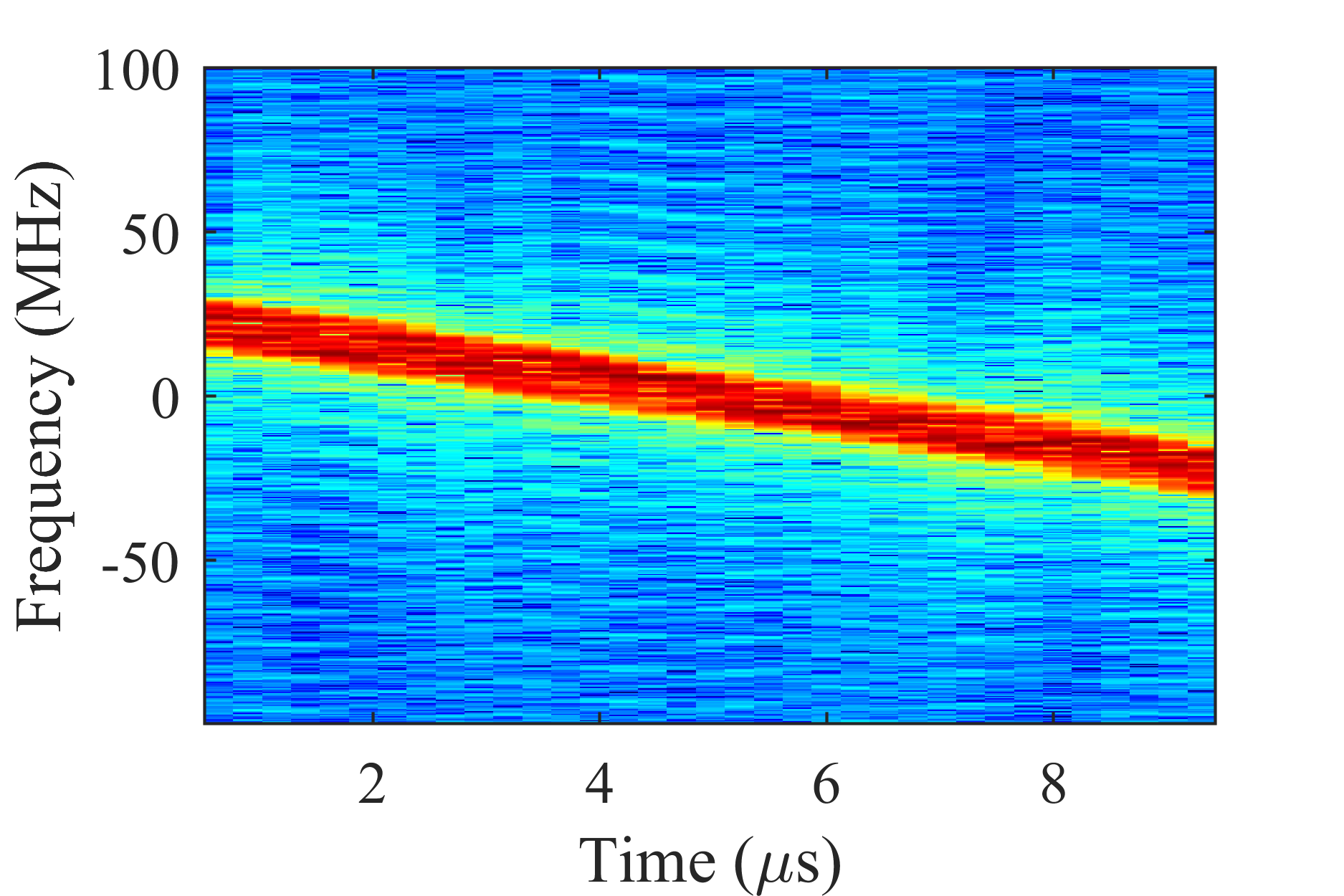}}

\caption{\textcolor{black}{Time-domain waveforms (top row) and time-frequency spectrograms (bottom row) of CM signals illustrate how the inner-layer symbols ride on the outer-layer chirp structure. (a),(d) Pure inner-layer modulation (BPSK). (b),(d) Pure outer-layer modulation (LFM). (c),(f) Composite modulation (BPSK-LFM).}}
\label{fig:visual_signal}
\end{figure*}

For concreteness, we consider a two-layer CM scheme similar to that in \cite{consultative2009radio}, where the signal is formed by an inner and an outer modulation layer. 
The inner-layer modulation is primarily used for communication, whereas the outer-layer modulation is designed for sensing. 
\textcolor{black}{
Fig.~\ref{fig:sig_gen} illustrates the signal generation flow, and Fig.~\ref{fig:visual_signal} compares the time-domain waveforms and time-frequency spectrograms of pure inner-layer, pure outer-layer, and CM signals.
}


For the inner layer, CM systems commonly adopt digital amplitude-phase modulation, expressed as
\begin{equation}
\label{eq11}
\begin{split}
{s_{\text{\rm in}}}(t) &= A \sum\limits_{n = 0}^{{N_{\text{\rm in}}}-1} h(t - n{T_{\text{\rm in}}}) \cdot a_n \cdot \exp \left( j{\phi_n}\right) \\
&= (a_{I,n} + j a_{Q,n}) \cdot h(t - n{T_{\text{\rm in}}})
\end{split},
\end{equation}
where $A$ is the amplitude scaling factor, $T_{\rm in}$ the symbol duration, $N_{\rm in}$ the number of symbols, $a_n$ and $\phi_n$ the amplitude and phase for the $n$-th symbol determined by modulation order $M$ ($\phi_n \in \{0, 2\pi/M, \dots, 2\pi\}$, $a_n$ from a discrete constellation), so $a_{I,n}, a_{Q,n} \in [\pm A, \pm 3A, \dots, \pm (2\sqrt{M}-1)A]$. 
$h_{\text{\rm in}}(t)$ is a shaping filter,
\begin{equation}
\label{eq6}
h_{\text{\rm in}}(t) = \frac{{\sin \left( {\pi t/{T_h}} \right)}}{{\pi t/{T_h}}} \cdot \frac{{\cos \left( {\alpha \pi t/{T_h}} \right)}}{{1 - {{\left( {2\alpha t/{T_h}} \right)}^2}}},
\end{equation}
where $\alpha$ and $T_h$ are the roll-off factor and cut-off time, respectively.
 
\color{black}
The outer-layer modulation employs frequency-modulated waveforms, whose instantaneous phase varies according to a modulation function $\varphi(t)$ that encodes either digital information or analog sensing parameters, written as
\begin{equation}
\label{eq:out_general}
s_{\rm out}(t) = g(t) \cdot \exp\left( j2\pi f_c t + j\varphi(t) \right),
\end{equation}
where $f_c$ is the carrier frequency.
$g(t)$ denotes the envelope function, taking the form of ${\rm rect}(t/T_p)$ for pulsed waveforms and unity otherwise.
$\varphi(t)$ is the phase modulation function that fully characterizes the specific scheme adopted.
For instance, in M-ary frequency shift keying (MFSK), data symbols $b_n$ are mapped onto discrete frequency offsets, yielding
\begin{equation}
\label{eq:mfsk_phi}
\varphi_{\rm MFSK}(t) = 2\pi f_{sep} \sum_{n=1}^{N_{\rm out}} b_n\, h(t - n T_{\rm out}),
\end{equation}
where $f_{sep}$ is the frequency separation between adjacent tones, $T_{\rm out}$ is the symbol duration, and $N_{\rm out}$ is the number of outer-layer symbols.
In linear frequency modulation (LFM), also known as chirp, the phase varies quadratically in time to produce a continuous sweep across bandwidth $B$ over pulse duration $T_p$, expressed as
\begin{equation}
\label{eq:lfm_phi}
\varphi_{\rm LFM}(t) = \pi k t^2,
\end{equation}
with chirp rate $k = B/T_p$ and $g_{\rm LFM}(t) = {\rm rect}(t/T_p)$.
The resulting time--bandwidth product $T_p \cdot B$ enables high range resolution via pulse compression without demanding excessive peak power, making LFM particularly attractive for radar and joint sensing systems.
Minimum shift keying (MSK) and Gaussian-filtered MSK (GMSK) extend this idea into the digital domain by imposing a continuous-phase constraint. The phase function integrates the symbol-weighted pulse train rather than switching abruptly between frequencies, i.e.,
\begin{equation}
\label{eq:msk_phi}
\varphi_{\rm MSK}(t) = \frac{\pi}{2T_{\rm out}} \int_0^t \sum_{n} b_n\, h(\tau - nT_{\rm out})\, d\tau,
\end{equation}
and GMSK further smooths $\varphi_{\rm MSK}(t)$ through a Gaussian pre-modulation filter to suppress spectral sidelobes.

\color{black}

The CM signal is synthesized by multiplying the inner-layer complex baseband signal with the outer-layer signal, and thus admits a compact unified form, expressed as
\begin{equation}
\label{eq:cm_unified}
s_{\rm CM}(t) = s_{\rm in}(t) \cdot s_{\rm out}(t).
\end{equation}
Assuming $N$ is the number of sampling points, the CM signal can be represented as $\textbf{s}_{\rm CM}=[s_{\rm CM}(1),s_{\rm CM}(2),\dots,s_{\rm CM}(N)]$, $\textbf{s}_{\text{\rm in}}=[s_{\text{\rm in}}(1),s_{\text{\rm in}}(2),\dots,s_{\text{\rm in}}(N)]$, and $\textbf{s}_{\text{\rm out}}=[s_{\text{\rm out}}(1),s_{\text{\rm out}}(2),\dots,s_{\text{\rm out}}(N)]$.
According to Eq.~(\ref{eq:cm_unified}), it can be inferred that the CM signals can be expressed as the product of the inner-layer modulation and the outer-layer modulation, given by
\begin{equation}
\textbf{s}_{\rm CM} = \textbf{s}_{\rm in} \odot \textbf{s}_{\rm out},
\end{equation}
where $\odot$ denotes the element-wise product operation.

\color{black}

\subsection{Imperfection Factors}
In practice, the received signal deviates from the ideal CM waveform due to transmitter hardware non-idealities, multipath fading, and additive noise. We therefore adopt a three-stage observation model, expressed as
\begin{equation}
\label{eq:imperfection_unified}
\left\{ {\begin{array}{*{20}{l}}
{{s_{{\rm{rx}}}}(n) = \underbrace {[h*{s_{{\rm{UIM}}}}](n)}_{({\rm{I}})\;{\rm{fading}}} + \underbrace {w(n)}_{({\rm{II}})\;{\rm{noise}}}}\\
{{s_{{\rm{UIM}}}}(n) = \underbrace {{s_{{\rm{CM}}}}(n) \cdot \Delta s(n)}_{({\rm{III}})\;{\rm{hardware}}\;{\rm{UIM}}}}
\end{array}} \right. ,
\end{equation}
where $*$ denotes convolution, $h(n)$ is the channel response, $w(n)\sim\mathcal{CN}(0,\sigma_w^2)$ is additive white Gaussian noise (AWGN), $s_{\rm CM}(n)$ is the ideal CM signal, and $\Delta s(n)$ aggregates all transmitter-side unintentional modulations (UIM).

\subsubsection{Hardware Imperfection}

Although an exact closed-form expression for aggregate hardware non-idealities is generally intractable, the dominant effects can be decomposed into unintentional modulations of amplitude, frequency, and phase \cite{zhao2022concentrate}.
Consequently, the hardware-impaired signal admits the compact factorization, written as
\begin{equation}
\label{eqn:uim}
s_{\rm UIM}(n) = s_{\rm CM}(n) \cdot \underbrace{\bigl(1 + \Delta A_n\bigr)\, e^{j\Delta\omega_n n}\, e^{j\Delta\varphi_n}}_{\displaystyle \Delta s(n)},
\end{equation}
where $\Delta A_n$, $\Delta\omega_n$, and $\Delta\varphi_n$ denote the amplitude, frequency, and phase UIM components, respectively.
Thus all transmitter imperfections collapse into a multiplicative distortion $\Delta s(n)$ acting pointwise on $s_{\rm CM}(n)$.

Three primary sources contribute to $\Delta s(n)$.
First, IQ imbalance arising from imperfect matching between I and Q branches introduces amplitude mismatch $\alpha$ and phase error $\phi$, 
\begin{equation}
\label{eq:iq}
\begin{split}
s_{\rm IQ}(n) = & (1 + \alpha)\, x_{\rm real}(n)\cos(2\pi f_c n + \phi) \\
& - j\, x_{\rm imag}(n)\sin(2\pi f_c n)
\end{split},
\end{equation}
affecting $\Delta A_n$ and $\Delta\varphi_n$ in \eqref{eqn:uim}.
Second, high-frequency oscillators (HFO) exhibit random phase fluctuations, i.e., oscillator phase noise, which manifest as an undesired multiplicative term on the signal. Under the widely-used sinusoidal model 
\cite{kim2002prediction}, Bessel expansion generates sidebands at $\pm f_{\rm os}$ around the carrier. To first order, the oscillator-distorted signal is
\begin{equation}
\label{eq:pn_approx}
s_{\rm os}(n) \approx x(n) + j\tfrac{\lambda}{2}\bigl[ x(n)e^{j2\pi f_{\rm os}n} - x(n)e^{-j2\pi f_{\rm os}n} \bigr],
\end{equation}
where $\lambda \ll 1$ is the modulation index and $f_{\rm os}$ the offset frequency, contributing to $\Delta\omega_n$ and $\Delta\varphi_n$ in \eqref{eqn:uim}.
Third, power amplifier (PA) nonlinearity operating near saturation primarily manifests as amplitude-to-amplitude distortion, i.e., AM-AM, and can be modeled via Taylor series \cite{liu2008specific}:
\begin{equation}
\label{eq:pa}
s_{\rm PA}(n) = \sum_{i=1}^{\infty} g_i [x(n)]^i,
\end{equation}
where higher-order terms ($i\geq 2$) affect $\Delta A_n$.

\subsubsection{Channel Effect}

The hardware-impaired signal propagates through a frequency-selective channel, yielding the tapped-delay-line model with $L$ paths, 
\begin{equation}
\label{eq:channel}
{s}_{\rm MP}(n) = \sum_{l=1}^{L} a_l\, s_{\rm UIM}(n - n_l).
\end{equation}
where $a_l$ and $n_l$ are the path gain and the delay of the $l$-th path, respectively, with $l = 1, 2, \dots, L$.

Substituting the hardware imperfection and the channel effect  into \eqref{eq:imperfection_unified} yields the complete observation model
\begin{equation}
\label{eq:full_model}
s_{\rm rx}(n) = \sum_{l=1}^{L} a_l \bigl[ s_{\rm CM}(n-n_l) \cdot \Delta s(n-n_l) \bigr] + w(n),
\end{equation}
which serves as input to identification algorithms.

\subsection{Zero-shot ACMR}
\label{subsec:zsl_formulation}

ACMR aims to identify the modulation scheme from a received signal, expressed as  
\begin{equation}
\label{eq:acmr}
{\mathcal F}: \mathbf{s}_{\rm rx} \longrightarrow y_{\rm CM},
\end{equation}
where $y_{\rm CM}$ denotes the CM label which decomposes naturally into an inner-layer and an outer-layer type, 
\begin{equation}
\label{eq:cm_label}
y_{\rm CM} = (y_{\rm in},\, y_{\rm out}) \in {\mathcal Y}_{\rm in} \times {\mathcal Y}_{\rm out},
\end{equation}
where ${\mathcal Y}_{\rm in}$ and ${\mathcal Y}_{\rm out}$ are the sets of candidate inner-layer and outer-layer modulation types. The total number of possible CM classes grows as $|{\mathcal Y}_{\rm in}| \cdot |{\mathcal Y}_{\rm out}|$, which quickly becomes prohibitive for exhaustive training.
This gives rise to a zero-shot ACMR problem.

Formally, let ${\mathcal D}_{\rm KN} = \{(\mathbf{s}_{\rm rx}, y) \mid y \in {\mathcal Y}_{\rm KN}\}$ denote the training set and ${\mathcal D}_{\rm UN} = \{(\mathbf{s}_{\rm rx}, y) \mid y \in {\mathcal Y}_{\rm UN}\}$ the test set, where ${\mathcal Y}_{\rm KN} \cap {\mathcal Y}_{\rm UN} = \emptyset$ and both are subsets of ${\mathcal Y}_{\rm in} \times {\mathcal Y}_{\rm out}$.
The zero-shot ACMR task then seeks a mapping
\begin{equation}
\label{eq:zs_acmr}
{\mathcal F}_{\rm ZS}: \mathbf{s}_{\rm rx} \longrightarrow (y_{\rm in}, y_{\rm out}), \quad (y_{\rm in}, y_{\rm out}) \in {\mathcal Y}_{\rm UN},
\end{equation}
that leverages only the semantic relationships between known and unknown classes. 

This problem is fundamentally distinct from conventional combinatorial generalization. In tasks such as multi-label classification or compositional vision, the constituent attributes can typically be observed independently. For example,  color and shape are recognized separately and then combined into a prediction. In CM, however, the inner and outer layers are not independently observable in the waveform domain. As revealed in Fig. 3,  the two layers interact through pointwise multiplication, producing a single entangled waveform from from which neither $y_{\rm in}$ nor $y_{\rm out}$ can be directly read off. 
Consequently, zero-shot ACMR is a waveform-level zero-shot learning problem where the raw observation fuses both sources into an inseparable mixture. 

\color{black}

\section{Zero-shot Learning Framework based on Disentangled Semantic Space}
\label{sec:framework}

\begin{figure}[!t]
    \centering
    \subfloat[]{%
        \includegraphics[width=\linewidth]{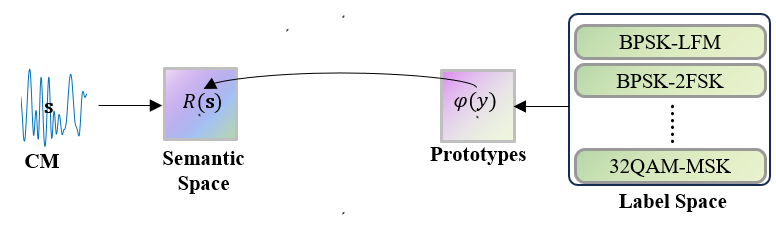}%
    }
    \\
    \subfloat[]{%
        \includegraphics[width=\linewidth]{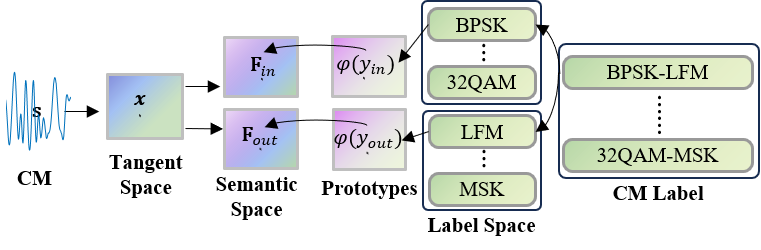}%
    }
    \caption{Comparison of zero-shot learning frameworks. (a) Traditional unified semantic space: each CM type as an independent class. (b) Proposed disentangled semantic space: each CM type as a disentanglement multi-task. }
    \label{fig:framework}
\end{figure}

A straightforward baseline, illustrated in Fig.~\ref{fig:framework}(a), treats each CM type as an atomic class and performs standard zero-shot classification in a global semantic space. 
Let the received signal be embedded via a feature extractor $\mathcal{R}(\cdot)$ and matched against prototypes of the target CM labels, written as
\begin{equation}
\label{eq:zs_baseline}
\hat{y}_{\rm CM} = \mathop {\arg \max }\limits_{{y_{{\rm{CM}}}} \in {{\mathcal Y}_{{\rm{UN}}}}}  \bigl\langle \mathcal{R}(\mathbf{s}_{\rm rx}),\, \boldsymbol{\varphi}(y_{\rm CM}) \bigr\rangle,
\end{equation}
where $\boldsymbol{\varphi}(y)$ denotes the prototype embedding of label $y_{\rm CM}$. $\bigl\langle \cdot, \cdot \bigr\rangle$ denotes the inner product, which measures the similarity between the received signal feature and a candidate label prototype.
This formulation ignores the compositional structure $y_{\rm CM} = (y_{\rm in}, y_{\rm out})$ and therefore offers limited extrapolation to unseen combinations of inner-layer and outer-layer modulations.

To overcome this limitation, we propose a disentangled semantic space in which the inner-layer and outer-layer modulations are represented, matched, and classified independently.
The core insight is that the multiplicative coupling $s_{\rm CM}(t) = s_{\rm in}(t) \cdot s_{\rm out}(t)$, although nonlinear in the waveform domain, becomes additive after a logarithmic projection onto the tangent space of the signal manifold. Once linearized, the two modulation components can be separated by a learnable geometric transformation, after which each is independently matched against its own layer-wise prototypes.

This design yields a three-stage framework, illustrated in Fig.~\ref{fig:framework}(b):
\textit{(i) Tangent-space mapping}: project the received signal onto a tangent space where the multiplicative coupling becomes additive (Section~\ref{subsec:log_linear}).
\textit{(ii) Learnable geometric disentanglement}: separate the additively combined components via an input-dependent affine transformation whose parameters are predicted from the observed signal (Section~\ref{subsec:disentangle}).
\textit{(iii) Compositional semantic matching}: classify each disentangled component independently against layer-wise semantic prototypes and compose the final CM label (Section~\ref{subsec:zsl}).

\subsection{Tangent-Space Mapping}
\label{subsec:log_linear}

The main obstacle to component separation is the nonlinear multiplicative coupling $s_{\rm CM}(t) = s_{\rm in}(t) \cdot s_{\rm out}(t)$.
Rather than addressing this coupling directly, we exploit the logarithmic homomorphism $\log : (\mathbb{R}^+, \times) \to (\mathbb{R}, +)$ to transform products into sums.

To accommodate the additive noise and channel effects in the full received model \eqref{eq:full_model}, we first isolate the additive perturbations into a composite distortion factor,
\begin{equation}
\label{eq:eta_def}
\eta(n) \triangleq \frac{\displaystyle\sum_{l=1}^{L} a_l \bigl[ s_{\rm CM}(n-n_l) \cdot \Delta s(n-n_l) \bigr] + w(n)}{s_{\rm CM}(n) \cdot \Delta s(n)},
\end{equation}
which yields the compact factorization $s_{\rm rx}(n) = s_{\rm CM}(n) \cdot \Delta s(n) \cdot \eta(n)$.
Applying the element-wise logarithmic map $\mathcal{T}(\cdot) = \log(\cdot)$ produces the linearized scalar decomposition:
\begin{equation}
\label{eq:log_decomp}
x(n) = x_{\rm CM}(n) + x_{\Delta}(n) + \tilde{n}(n),
\end{equation}
where $x(n) = \mathcal{T}(s_{\rm rx}(n))$, $x_{\rm CM}(n) = \mathcal{T}(s_{\rm CM}(n))$, $x_{\Delta}(n) = \mathcal{T}(\Delta s(n))$, and $\tilde{n}(n) = \mathcal{T}(\eta(n))$.
Collecting $N$ samples into column vectors gives the additive representation,
\begin{equation}
\label{eq:log_vec}
\mathbf{x} = \mathbf{x}_{\rm CM} + \mathbf{x}_{\Delta} + \tilde{\mathbf{n}}.
\end{equation}

This linearization admits a precise geometric interpretation.
The space of strictly positive-valued waveforms forms a Lie group under element-wise multiplication, with the all-ones vector as the identity element.
The operator $\mathcal{T}(\cdot)$ corresponds to the inverse exponential map at this identity \cite{absil2008optimization}, projecting signals from the curved multiplicative manifold onto its flat tangent space.
While the projection does not automatically resolve the blind source separation problem, it globally converts the nonlinear group operation into standard vector addition, providing a rigorous geometric foundation for the subsequent parametric separation.

\subsection{Learnable Geometric Disentanglement}
\label{subsec:disentangle}

In the tangent space, the signal components are additively superimposed as in \eqref{eq:log_vec}.
The remaining objective is to isolate the CM component $\mathbf{x}_{\rm CM}$ from the nuisance terms $\mathbf{x}_{\Delta}$ and $\tilde{\mathbf{n}}$.
Although this recovery is inherently ill-posed without explicit prior constraints, we show that the hypothesis class of input-dependent affine transformations possesses sufficient expressive capacity to encapsulate the exact solution.

\begin{proposition}[Affine Disentanglement Sufficiency]
\label{prop:affine_disentangle}
For any composite tangent-space representation $\mathbf{x} = \mathbf{x}_{\rm CM} + \mathbf{x}_{\Delta} + \tilde{\mathbf{n}}$, there exists an input-dependent affine transformation that perfectly isolates $\mathbf{x}_{\rm CM}$:
\begin{equation}
\label{eq:affine}
\begin{bmatrix} \mathbf{x} & 1 \end{bmatrix} \boldsymbol{\Theta}(\mathbf{x}) = c_\theta\, \mathbf{x}_{\rm CM},
\end{equation}
where $\boldsymbol{\Theta}(\mathbf{x}) \in \mathbb{C}^{2 \times N}$ is the sample-specific parameter matrix and $c_\theta > 0$ is a scalar scaling constant.
\end{proposition}

\begin{proof}
Partition $\boldsymbol{\Theta}(\mathbf{x}) = \begin{bmatrix} \mathbf{A} \\ \mathbf{b} \end{bmatrix}$ with $\mathbf{A} \in \mathbb{C}^{1\times N}$ and $\mathbf{b} \in \mathbb{C}^{1\times N}$, so that \eqref{eq:affine} reduces to $\mathbf{x}\mathbf{A} + \mathbf{b}$.
Setting $\mathbf{A}^* = c_\theta \mathbf{I}$ and $\mathbf{b}^* = -c_\theta (\mathbf{x}_{\Delta} + \tilde{\mathbf{n}})$ and substituting $\mathbf{x} = \mathbf{x}_{\rm CM} + \mathbf{x}_{\Delta} + \tilde{\mathbf{n}}$ yields:
\begin{equation}
\mathbf{x}\mathbf{A}^* + \mathbf{b}^* = c_\theta (\mathbf{x}_{\rm CM} + \mathbf{x}_{\Delta} + \tilde{\mathbf{n}}) - c_\theta (\mathbf{x}_{\Delta} + \tilde{\mathbf{n}}) = c_\theta\, \mathbf{x}_{\rm CM},
\end{equation}
which establishes the assertion.
\end{proof}

Proposition~\ref{prop:affine_disentangle} is an existence result: it guarantees that the affine model class is rich enough to contain the ideal disentangling operator, but does not provide a closed-form solution, since the oracle parameters depend on the unobservable components $\mathbf{x}_{\Delta}$ and $\tilde{\mathbf{n}}$.
We therefore cast the parameter generation as a learnable task, i.e.,
\begin{equation}
\label{eq:learnable_theta}
\boldsymbol{\Theta}(\mathbf{x}) = T_\Phi(\mathbf{x}),
\end{equation}
where $T_\Phi$ is a parametric mapping optimized end-to-end from data.
Rather than hand-crafting separation filters, the network learns to predict input-dependent affine parameters that approximate the ideal disentangling transform, leveraging the implicit regularization of the neural architecture.
The specific architectural realization of $T_\Phi$ is deferred to Section~\ref{sec:method}.

\subsection{Compositional Semantic Matching}
\label{subsec:zsl}

The tangent-space mapping and geometric disentanglement established above produce layer-wise feature representations from the entangled waveform.
We now construct the semantic space against which these features are matched, completing the framework.
Given the tangent-space representation $\mathbf{x}$, the feature extractor $\mathcal{R}_\Phi$ applies the learned disentangling transform and produces two layer-wise feature vectors, expressed as
\begin{equation}
\label{eq:feature_extraction}
\bigl[\mathbf{F}_{\rm in} \;\; \mathbf{F}_{\rm out}\bigr] = \mathcal{R}_\Phi(\mathbf{x}),
\end{equation}
where $\mathbf{F}_{\rm in}, \mathbf{F}_{\rm out} \in \mathbb{R}^{d}$ encode the characteristics of the inner and outer modulation components, respectively.

Let ${\mathcal Y}_{\rm in}$ and ${\mathcal Y}_{\rm out}$ denote the candidate inner-layer and outer-layer modulation types, with the full CM label space ${\mathcal Y}_{\rm CM} \subseteq {\mathcal Y}_{\rm in} \times {\mathcal Y}_{\rm out}$ partitioned into seen classes ${\mathcal Y}_{\rm CM}^{\rm s}$ and unseen classes ${\mathcal Y}_{\rm CM}^{\rm u} = {\mathcal Y}_{\rm CM} \setminus {\mathcal Y}_{\rm CM}^{\rm s}$.
Unlike the unified semantic space \eqref{eq:zs_baseline} that assigns a single prototype to each composite label, we construct a disentangled semantic space comprising two independent learnable prototype sets. Each inner-layer type $y_{\rm in} \in {\mathcal Y}_{\rm in}$ is associated with a trainable prototype $\boldsymbol{\psi}_{\rm in}(y_{\rm in}) \in \mathbb{R}^{d}$, while each outer-layer type $y_{\rm out} \in {\mathcal Y}_{\rm out}$ is $\boldsymbol{\psi}_{\rm out}(y_{\rm out}) \in \mathbb{R}^{d}$.
Classification is performed independently for each layer via nearest-prototype matching, expressed as
\begin{equation}
\label{eq:zsl_classify}
\left\{ \begin{array}{l}
\hat{y}_{\rm in} = \mathop {\arg \max }\limits_{{y_{{\rm{in}}}} \in {{\mathcal Y}_{{\rm{in}}}}}  \bigl\langle \mathbf{F}_{\rm in},\, {\psi}_{\rm in}(y_{\rm in}) \bigr\rangle \\[4pt]
\hat{y}_{\rm out} = \mathop {\arg \max }\limits_{{y_{{\rm{out}}}} \in {{\mathcal Y}_{{\rm{out}}}}}  \bigl\langle \mathbf{F}_{\rm out},\, {\psi}_{\rm out}(y_{\rm out}) \bigr\rangle
\end{array} \right. .
\end{equation}
The final prediction is the composite label $\hat{y}_{\rm CM} = (\hat{y}_{\rm in}, \hat{y}_{\rm out})$.

This disentangled semantic space provides two essential properties for zero-shot generalization.
\emph{(1)} Compositionality: since each layer's prediction depends solely on its own prototype set, any combination $(y_{\rm in}, y_{\rm out})$ whose individual constituents have appeared in the training set can be recognized at test time without dedicated training samples.
\emph{(2)} Scalability: incorporating a new modulation type into either layer requires only registering its prototype, with no retraining over the expanded label space.
Both properties stem directly from the factored structure of the semantic space and the preceding disentanglement, which together convert the waveform-level zero-shot problem identified in Section~\ref{subsec:zsl_formulation} into two standard single-layer recognition tasks.

\color{black}

\section{Method for Zero-shot ACMR}
\label{sec:method}

\subsection{Overview}  
We propose the practical network TSDN with a symmetric dual-branch architecture (Fig.~\ref{fig:network}), containing one branch for inner-layer modulation recognition and the other one for outer-layer modulation. 
In each branch, the input $\mathbf{s}$ is logarithmically mapped to tangent space, then processed by an STN to adaptively decompose linear components.
A one-dimensional convolutional neural network (1D-CNN) backbone extracts features, followed by a classification layer for modulation prediction. This dual-branch design disentangles the two layers, enhancing zero-shot recognition accuracy and robustness. 

\begin{figure*}[!t]
    \centering  \includegraphics[width=0.95\linewidth]{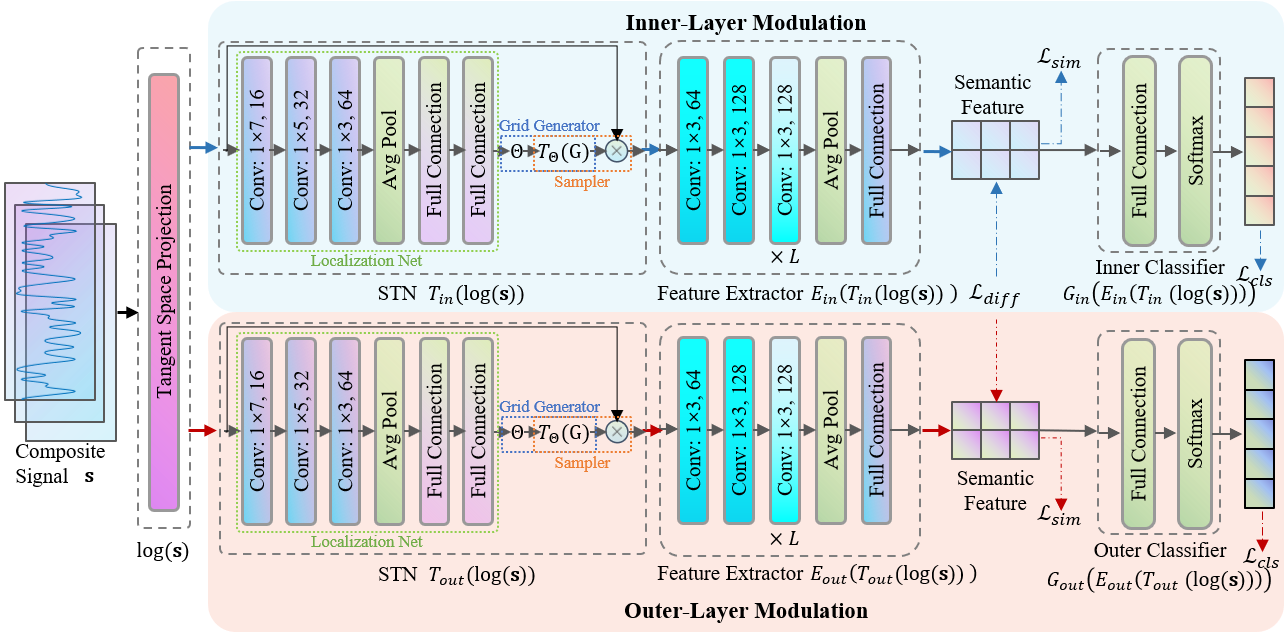}  
    \caption{Architecture of the proposed TSDN for zero-shot automatic composite modulation recognition based on tangent space disentanglement.}  
    \label{fig:network}  
\end{figure*}

\subsection{Spatial Transformer Network} 
\color{black}
In TSDN, the STN is employed to learn the linear alignment parameters for geometric disentanglement ~\cite{NIPS2015_33ceb07b}. 
Unlike the idealized construction in Proposition~1, the STN learns the transformation matrices \(\mathbf{\Theta}_{\mathrm{in}}\) and \(\mathbf{\Theta}_{\mathrm{out}}\) directly from the input data via end-to-end training, using only class labels as supervision. Formally,
\begin{equation}
\left\{ \begin{array}{l}
\mathbf{x}_{\mathrm{in}} = T_{\mathrm{in}}(\mathbf{x}) = \left[ \mathbf{x} \quad \mathbf{1} \right] \mathbf{\Theta}_{\mathrm{in}} \\[4pt]
\mathbf{x}_{\mathrm{out}} = T_{\mathrm{out}}(\mathbf{x}) = \left[ \mathbf{x} \quad \mathbf{1} \right] \mathbf{\Theta}_{\mathrm{out}}
\end{array} \right. .
\end{equation}
The STN module, adapted to temporal signal processing, offers low computational overhead and seamless integration with convolutional backbones. By actively aligning the input signals prior to feature extraction, it facilitates more discriminative learning for both inner-layer and outer-layer modulations.
\color{black}

\subsection{Feature Extractor}  
The two sequences generated by the STN transformation correspond to the inner-layer and outer-layer modulated signals, respectively. 
They are directly input into a standalone 1D-CNN. As the backbone feature extraction module, the 1D-CNN effectively captures local temporal dependencies with low computational overhead. Furthermore, its hierarchical structure allows for the extraction of semantically meaningful features across multiple temporal scales.

The proposed backbone employs a symmetric dual-branch architecture to independently process the decomposed input signals $\mathbf{x}_{\rm{\rm in}}$ and $\mathbf{x}_{\rm{\rm out}}$ in parallel. Each branch consists of identical feature extraction modules comprising cascaded 1D convolutional layers followed by linear projection layers. Specifically, the input signal $\mathbf{x}_{\rm in} \in \mathbb{R}^{N \times D}$ is fed into the first branch, while $\mathbf{x}_{\rm out} \in \mathbb{R}^{N \times D}$ is processed through the second branch, where $N$ represents the temporal sequence length and $D=2$ denotes the input feature dimension.

The feature extraction process for each branch can be formulated as
\begin{equation}
\left\{ {\begin{array}{*{20}{c}}
{{{\bf{F}}_{{\rm{in}}}} = {E_{\theta_{\rm in}} }({{\bf{x}}_{{\rm{in}}}}) = {\rm{Linear}}({\rm{GAP}}({\rm{Conv1D}}({{\bf{x}}_{{\rm{in}}}})))}\\
{{{\bf{F}}_{{\rm{out}}}} = {E_{\theta_{\rm out}} }({{\bf{x}}_{{\rm{out}}}}) = {\rm{Linear}}({\rm{GAP}}({\rm{Conv1D}}({{\bf{x}}_{{\rm{out}}}})))}
\end{array}} \right. ,
\end{equation}
where $f_{\mathbf{\theta_{\rm in}}}(\cdot)$ and $f_{\mathbf{\theta_{\rm out}}}(\cdot)$ represent feature extractors with parameters $\mathbf{\theta}_{\rm in}$ and $\mathbf{\theta}_{\rm out}$. 
$\text{Conv1D}(\cdot)$ denotes the cascaded convolutional operations. 
$\text{GAP}(\cdot)$ is the global average pooling operation, and $\text{Linear}(\cdot)$ represents the linear transformation. This process produces high-level feature embeddings $\mathbf{F}_{\rm in} \in \mathbb{R}^{d} $, $\mathbf{F}_{\rm out} \in \mathbb{R}^{d}$, where $d$ is the output feature dimension.

Subsequently, individual classification heads map these feature representations to their corresponding label spaces
\begin{equation}
\left\{ {\begin{array}{*{20}{c}}
{ {\bf{z}}_{\rm in} = {\rm{softmax}}({{\bf{W}}_{{\rm{in}}}}{{\bf{F}}_{{\rm{in}}}} + {{\bf{b}}_{{\rm{in}}}})}\\
{{\bf{z}}_{\rm out} = {\rm{softmax}}({{\bf{W}}_{{\rm{out}}}}{{\bf{F}}_{{\rm{out}}}} + {{\bf{b}}_{{\rm{out}}}})}
\end{array}} \right.,
\end{equation}
where $\textbf{W}_{\rm in} \in \mathbb{R}^{C_{\rm in} \times d}$, $\textbf{W}_{\rm out} \in \mathbb{R}^{C_{\rm out} \times d}$ are the projection matrices, $\mathbf{b}_{\rm in} \in \mathbb{R}^{C_{\rm in}}$, $\mathbf{b}_{\rm out} \in \mathbb{R}^{C_{\rm out}}$ are the bias vectors, and $C_{\rm in}$ and $C_{\rm out}$ are the number of inner-layer and outer-layer modulation classes. 
${\bf z}_{\rm in}$ and ${\bf z}_{\rm out}$ are the probability vectors for the inner-layer and outer-layer modulation schemes.
${{\hat y}_{{\rm{in}}}} = \arg \max (\bf{z}_{\rm in})$ and ${{\hat y}_{{\rm{out}}}} = \arg \max (\bf{z}_{\rm out})$ are the predicted labels for the inner-layer and outer-layer modulation types.
This symmetric design ensures consistent feature extraction across both signal components while maintaining computational efficiency and enabling independent optimization of each branch during training.

\subsection{Loss Function}  
To optimize TSDN, we design a composite loss function that integrates three components, jointly ensuring robust classification, compact feature embedding, and effective subspace disentanglement:  
\begin{equation}  
    \mathcal{L} = \lambda_1 \mathcal{L}_1 + \lambda_2 \mathcal{L}_2 + \lambda_3 \mathcal{L}_3,  
\end{equation}  
where $\mathcal{L}_1$, $\mathcal{L}_2$, and $\mathcal{L}_3$ denote cross-entropy loss, center loss, and difference loss, respectively. 
The weights $\lambda_1$, $\lambda_2$, and $\lambda_3$ are hyperparameters that balance the contributions of their corresponding loss terms.

The first component, denoted as $\mathcal{L}_1$, is the cross-entropy loss. This function measures the discrepancy between the predicted probability distributions of the two branches, which correspond to the inner-layer and outer-layer modulation types, and their respective ground-truth labels \cite{mao2023cross}. 
Since we leverage a multi-label strategy for the inner-layer and outer-layer modulation, we have   
\begin{equation}  
    \mathcal{L}_1 = -\sum_{i=1}^{|{\mathcal Y}_{\rm in}|} {\bf z}^i_{\rm in} \log (\hat{{\bf z}}^i_{\rm in})-\sum_{i=1}^{|{\mathcal Y}_{\rm out}|} {\bf z}^i_{\rm out} \log (\hat{{\bf z}}^i_{\rm out}),  
\end{equation}  
where $|{\mathcal Y}_{\rm in}|$ and $|{\mathcal Y}_{\rm out}|$ are respectively the total number of inner-layer and outer-layer modulation classes. This term enforces discriminative capability at the classification level.  

\color{black}
The second loss term \(\mathcal{L}_{2}\) is the center loss in order to jointly enforce intra-class compactness in both the inner-layer and outer-layer feature spaces.
The center loss for each branch is defined as the mean squared Euclidean distance between features and their corresponding centers.
Let $\mathbf{F}_{\text{in}}^{(i)}\in\mathbb{R}^{d}$ and $\mathbf{F}_{\text{out}}^{(i)}\in\mathbb{R}^{d}$ denote the inner and outer features of the $i$-th sample in a mini-batch of size $m$, with corresponding labels $y_{\text{in}}^{(i)}$ and $y_{\text{out}}^{(i)}$, and the center loss ${\mathcal L}_2$ can be expressed as
\begin{equation}
\mathcal{L}_2 = \frac{1}{2}\sum_{i=1}^{m} \big\| \mathbf{F}_{\text{in}}^{(i)} - \mathbf{c}^{y_{\text{in}}^{(i)}}_{\text{in}} \big\|_2^2+\frac{1}{2}\sum_{i=1}^{m} \big\| \mathbf{F}_{\text{out}}^{(i)} - \mathbf{c}^{y_{\text{out}}^{(i)}}_{\text{out}} \big\|_2^2,
\end{equation}
where $\mathbf{c}_{\rm in} \in \mathbf{C}_{\text{in}}\in\mathbb{R}^{|\mathcal{Y}_{\text{in}}|\times d}$ and $\mathbf{c}_{\rm out} \in \mathbf{C}_{\text{out}}\in\mathbb{R}^{|\mathcal{Y}_{\text{out}}|\times d}$ are learnable class centers.

Following~\cite{wen2016discriminative}, the class centers are not updated via back-propagation but with an explicit mini-batch moving average.
For each branch, the centers are independently updated as
\begin{equation}
\mathbf{c}_j^{\text{in}} \leftarrow \mathbf{c}_j^{\text{in}} - {lr}_{2} \cdot \frac{\sum_{i=1}^{m} \delta(y_{\text{in}}^{(i)}=j)\big(\mathbf{c}_j^{\text{in}} - \mathbf{F}_{\text{in}}^{(i)}\big)}{1 + \sum_{i=1}^{m} \delta(y_{\text{in}}^{(i)}=j)},
\end{equation}
and analogously for $\mathbf{c}_j^{\text{out}}$, where $\delta(\cdot)$ is the indicator function and ${lr}_{2}$ is the center learning rate.
During training, the center vectors are detached from the computational graph, so that gradients from $\mathcal{L}_{{2}}$ guide only the feature extractors.

\color{black}
The third component $\mathcal{L}_3$ promotes feature disentanglement across branches by penalizing their cross-correlation. Drawing inspiration from domain separation networks \cite{NIPS2016_45fbc6d3}, this loss term enforces orthogonality constraints between the feature subspaces of the inner-layer and outer-layer representations.  
\begin{equation}  
    \mathcal{L}_3 = \left\| \mathbf{F}_{\rm in}^{\mathrm{T}} \mathbf{F}_{\rm out} \right\|_{\rm F}^2,  
\end{equation}  
where $\mathbf{F}_{\rm in}$ and $\mathbf{F}_{\rm out}$ are the feature matrices from the two branches, and $\|\cdot\|_{\rm F}$ denotes the Frobenius norm. 
By minimizing $\mathcal{L}_3$, the network explicitly separates the two subspaces, thereby improving feature disentanglement and overall recognition performance.

\section{Simulation Results and Analysis}
\label{sec:experiment}

\color{black}
\subsection{Design of Experiments}  

To thoroughly validate the effectiveness of the proposed TSDN, five targeted experiments were designed. Experiment 1 focuses on visualizing the learned feature representations, thereby highlighting the strengths of the tangent space disentanglement framework. Experiments 2 and 3 present ablation studies examining both the overall framework and individual components. Experiment 4 analyzes the model’s robustness in zero-shot settings through performance evaluation with different unknown targets, specifically addressing discrimination tasks under high-order modulation schemes and performance across held-out modulation types. Finally, Experiment 5 evaluates the algorithm under realistic impairments, including AWGN, multipath effects, and hardware imperfections.

\subsection{Dataset}

Although several open-source AMR datasets have been released recently~\cite{zhang2026cognitive, tekbiyik2019hisarmod, o2018over, huang2024multi}, they are designed for either single modulation recognition or compound modulation with predefined categories, and do not support the controllable evaluation of compositional zero-shot generalization.
We therefore construct a customized dataset that faithfully follows the CM signal model in Section~\ref{sec:system_model} and allows systematic holdout of any inner--outer combination.

All signals are generated at baseband with a sampling rate of 200~MHz.
For inner-layer communication modulations, the symbol rate is 10~Msps; for outer-layer sensing modulations, 2FSK, 4FSK, 8FSK, and MSK use 2~Msps with 5~MHz frequency spacing, while LFM occupies 50~MHz bandwidth.
Each CM type contains 5{,}000 independent samples of 512 points each.

\begin{table*}[!t]
\centering
\caption{CM dataset configuration. Each cell corresponds to one composite modulation type formed by the inner-layer (row) and outer-layer (column) scheme, yielding a total of $|\mathcal{Y}_{\rm in}| \times |\mathcal{Y}_{\rm out}| = 25$ CM classes and $|\mathcal{Y}_{\rm SM}| = 10$ SM classes.}
\label{tab:dataset}
\small
\setlength{\tabcolsep}{6pt}
\renewcommand{\arraystretch}{1.2}
\begin{tabular}{c|cccccc}
\toprule
\textbf{Inner \textbackslash Outer} & \textbf{LFM} & \textbf{2FSK} & \textbf{4FSK} & \textbf{8FSK} & \textbf{MSK}  & \textbf{NONE}\\
\midrule
BPSK  & BPSK-LFM  & BPSK-2FSK  & BPSK-4FSK  & BPSK-8FSK  & BPSK-MSK & BPSK\\
QPSK  & QPSK-LFM  & QPSK-2FSK  & QPSK-4FSK  & QPSK-8FSK  & QPSK-MSK & QPSK\\
8PSK  & 8PSK-LFM  & 8PSK-2FSK  & 8PSK-4FSK  & 8PSK-8FSK  & 8PSK-MSK & 8PSK\\
16QAM & 16QAM-LFM & 16QAM-2FSK & 16QAM-4FSK & 16QAM-8FSK & 16QAM-MSK & 16QAM\\
32QAM & 32QAM-LFM & 32QAM-2FSK & 32QAM-4FSK & 32QAM-8FSK & 32QAM-MSK & 32QAM\\
NONE & LFM & 2FSK & 4FSK & 8FSK & MSK & \textbackslash \\
\bottomrule
\end{tabular}
\end{table*}

The inner-layer set \({\mathcal Y}_{\rm in} =\)\{BPSK, QPSK, 8PSK, 16QAM, 32QAM, NONE\} and the outer-layer set \({\mathcal Y}_{\rm out} =\) \{LFM, 2FSK, 4FSK, 8FSK, MSK, NONE\} jointly define $|{\mathcal Y}_{\rm CM}| = 25$ CM types where ${\mathcal Y}_{\rm CM} = {\mathcal Y}_{\rm in} \times {\mathcal Y}_{\rm out}$, together with $|{\mathcal Y}_{\rm SM}| = 10$ single modulation (SM) types, as shown in Table~\ref{tab:dataset}.

\subsection{Training Protocol}

The models were trained with a custom three-stage multistage strategy that progressively enhances classification performance, feature compactness, and branch complementarity. 
In Stage 1, training used the sum of two cross-entropy losses ${\mathcal L}_1$ to establish basic classification capabilities for the two tasks. Upon triggering early stopping monitored on validation loss with a patience of 20, the process automatically advanced to Stage 2, which additionally incorporated center losses to promote intra-class compactness ${\mathcal L}_2$. 
In Stage 3, a difference loss ${\mathcal L}_3$ was further added to reduce feature redundancy between the two branches while preserving discriminative power.

All loss components were combined into a single total loss for unified backward propagation across stages. The loss weights were fixed at $\lambda_1 = 1$ for the first center loss, $\lambda_2 = 0.002$ for the second center loss, and $\lambda_d = 0.001$ for the difference loss. 
Training ran for a maximum of 200 epochs with an initial learning rate of $10^{-3}$ and a cosine annealing scheduler. 
The dataset was randomly partitioned into training, validation, and test sets in an 8:1:1 ratio. Experiments employed a batch size of 128 and a hidden feature dimension of 64. The best model, selected according to validation performance, was saved via checkpoints and loaded for final evaluation to ensure stable and reproducible test results. All implementations were performed in PyTorch 2.5.0 under Python 3.10.0 on Windows 10, using an AMD 16-core CPU at 4.0 GHz and an NVIDIA RTX 4090 GPU.

\subsection{Experiment 1: Rationality of the Proposed Tangent Space Disentanglement}

\begin{figure*}[!t]
\centering
\subfloat[CNN1]{\includegraphics[width=0.24\linewidth]{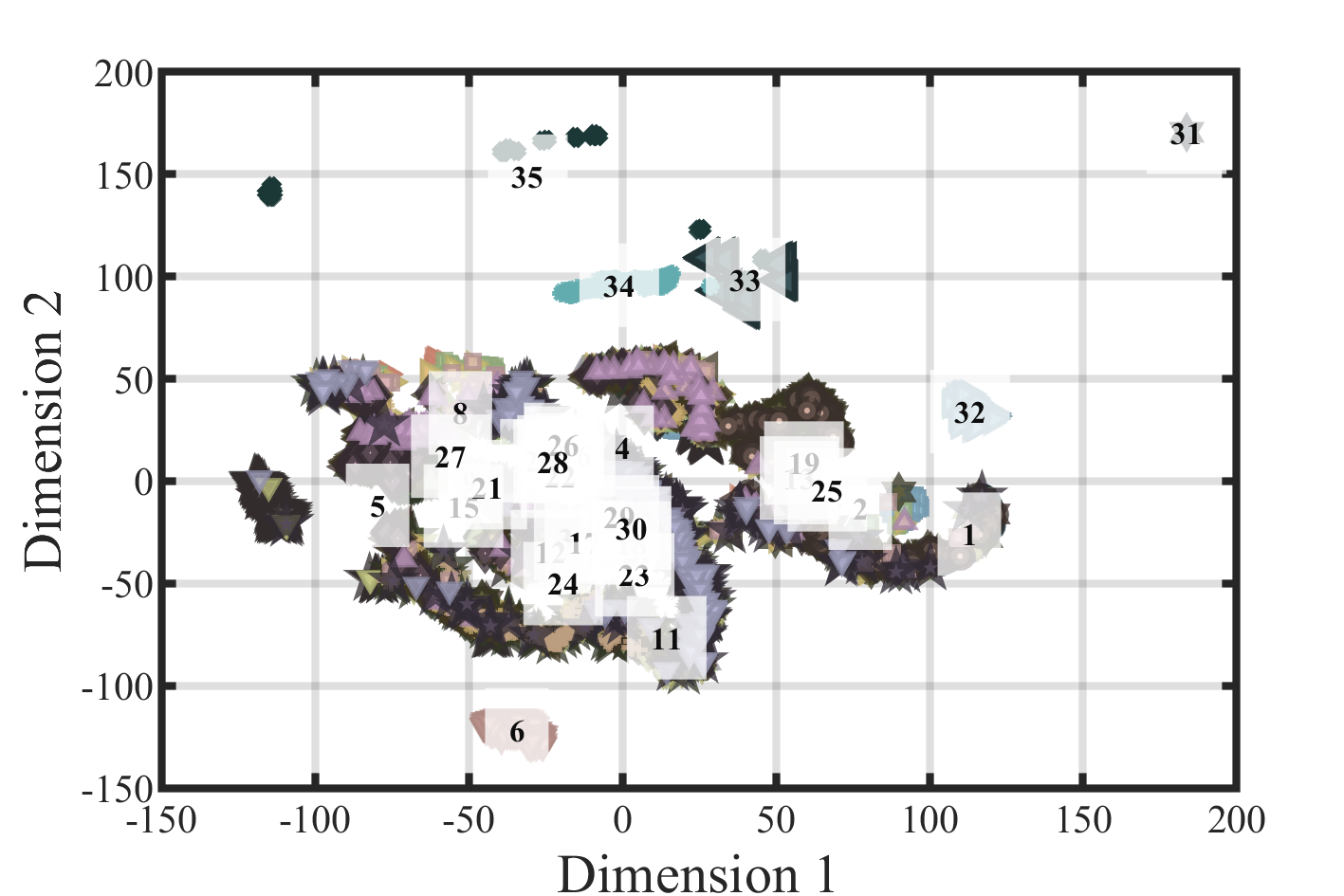}}
\subfloat[CNN2]{\includegraphics[width=0.24\linewidth]{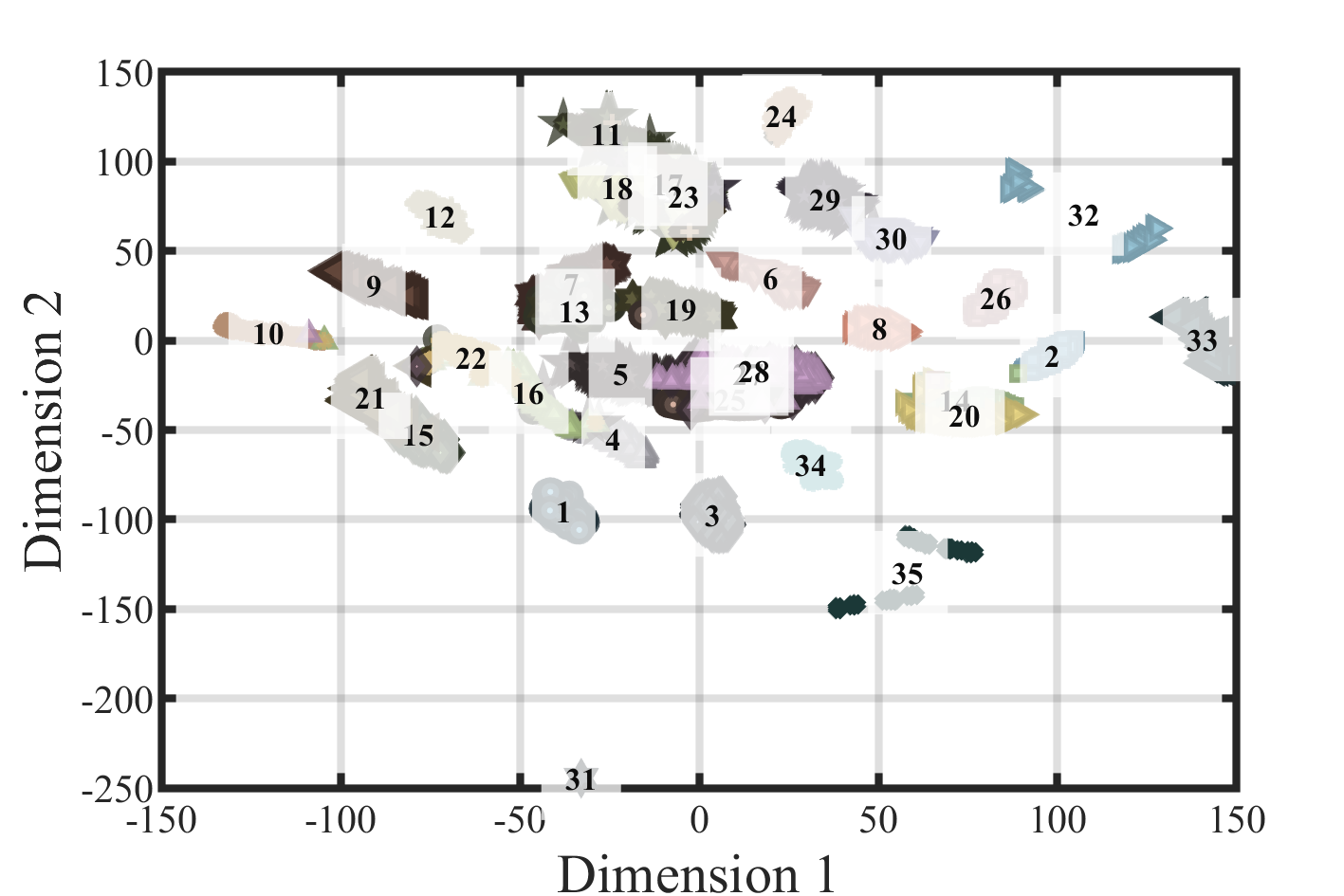}}
\subfloat[Resnet]{\includegraphics[width=0.24\linewidth]{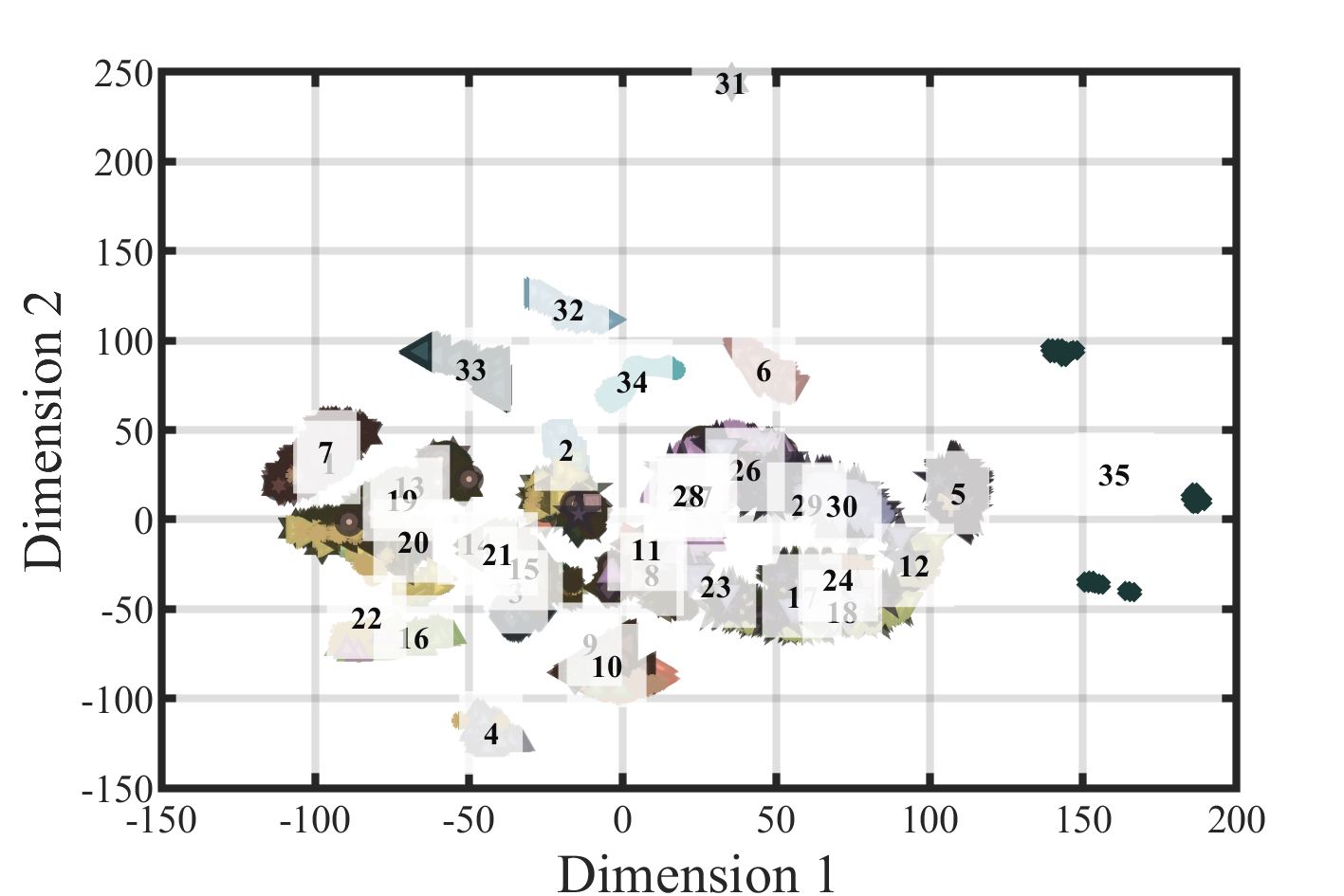}}
\subfloat[DenseNet]{\includegraphics[width=0.24\linewidth]{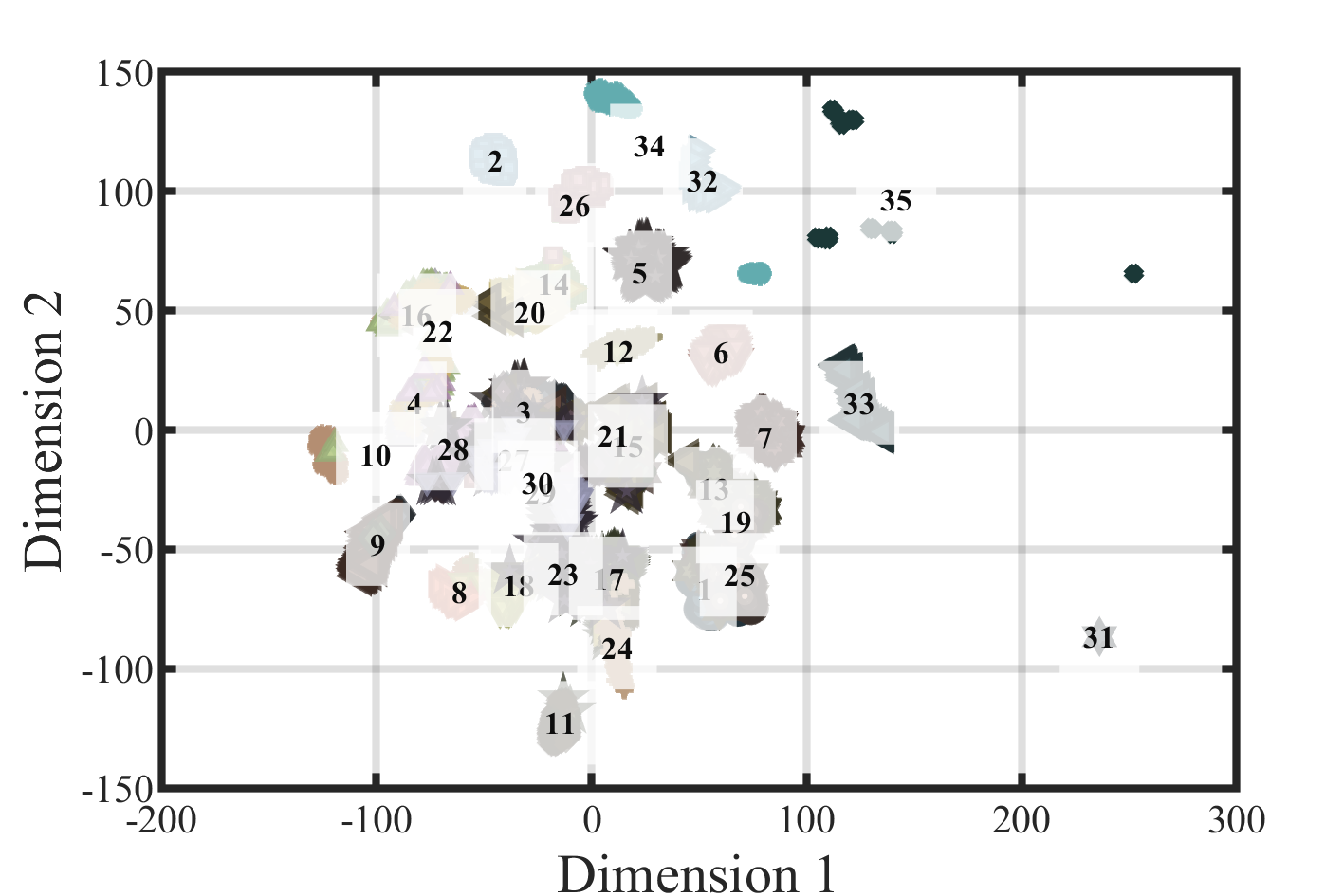}}
\\
\subfloat[GRU]{\includegraphics[width=0.24\linewidth]{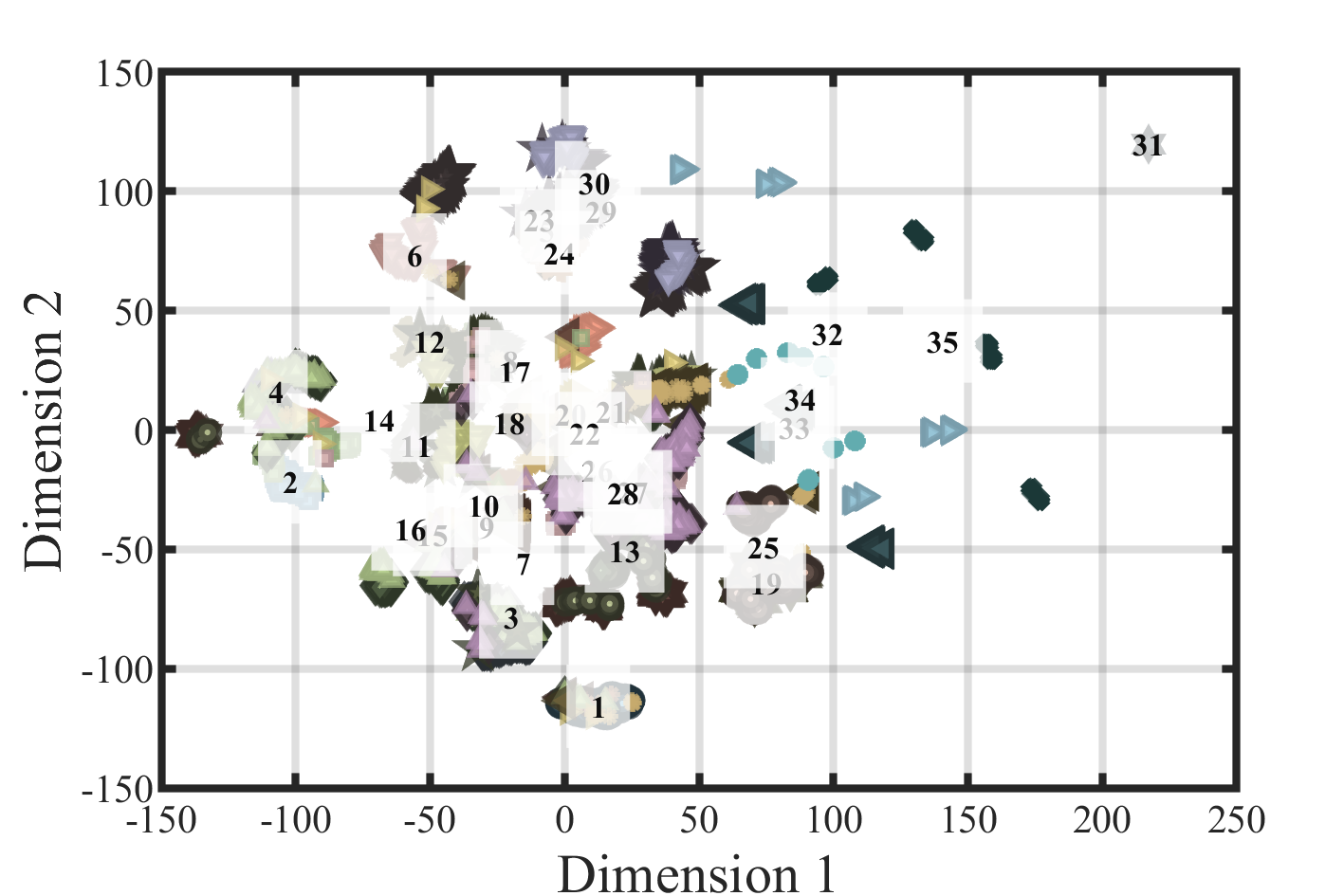}}
\subfloat[BiGRU]{\includegraphics[width=0.24\linewidth]{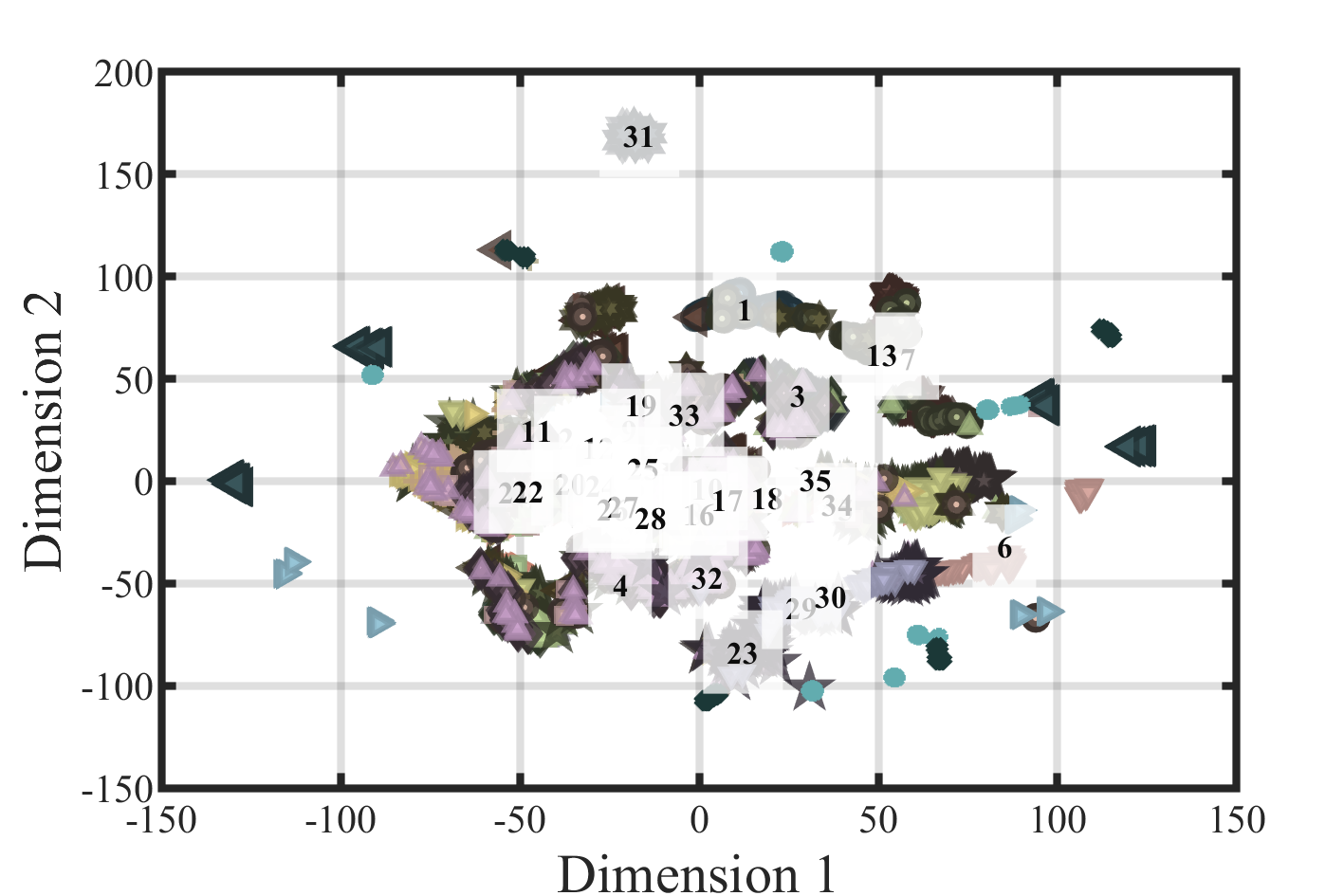}}
\subfloat[ICAMC]{\includegraphics[width=0.24\linewidth]{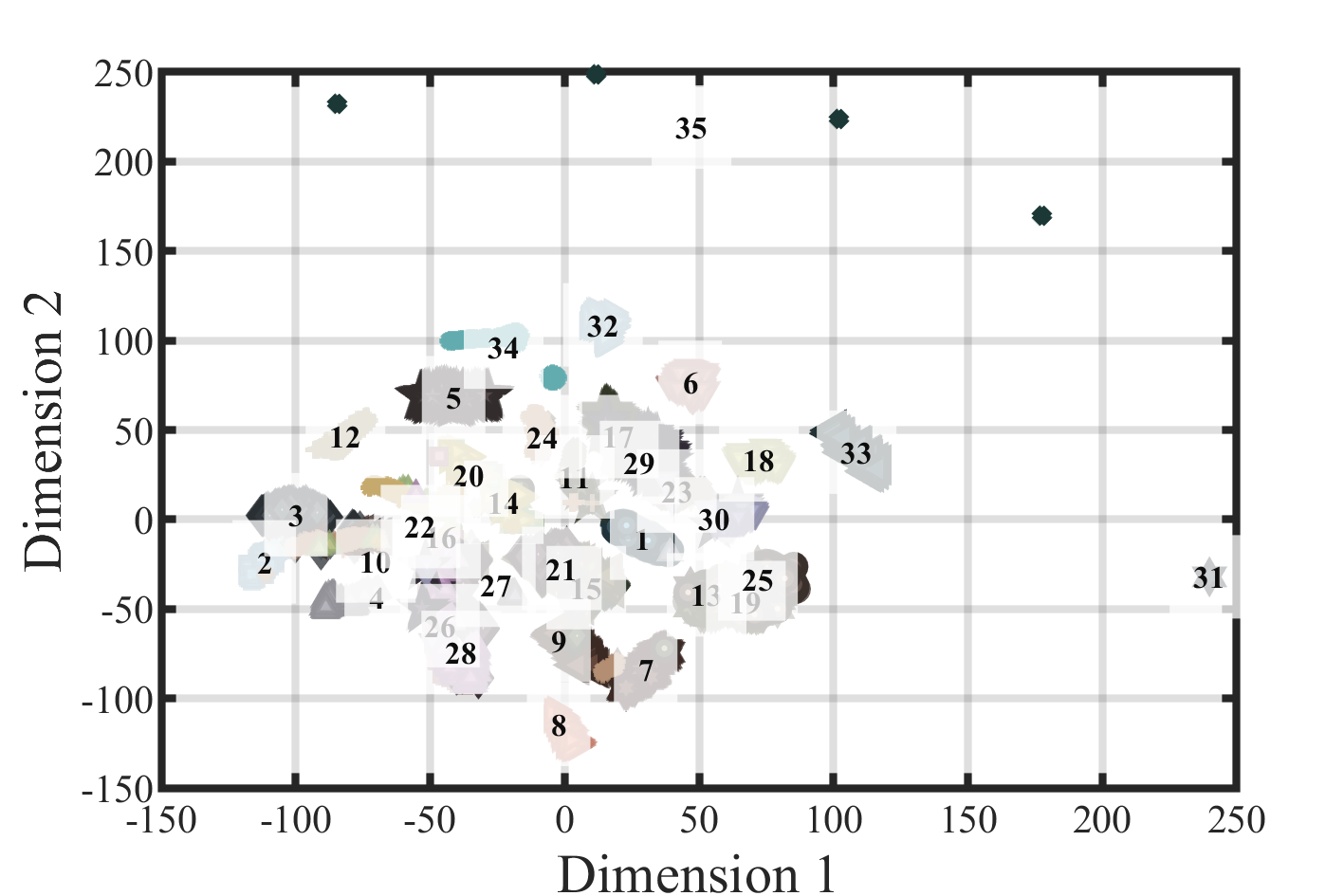}}
\subfloat[Proposed]{\includegraphics[width=0.24\linewidth]{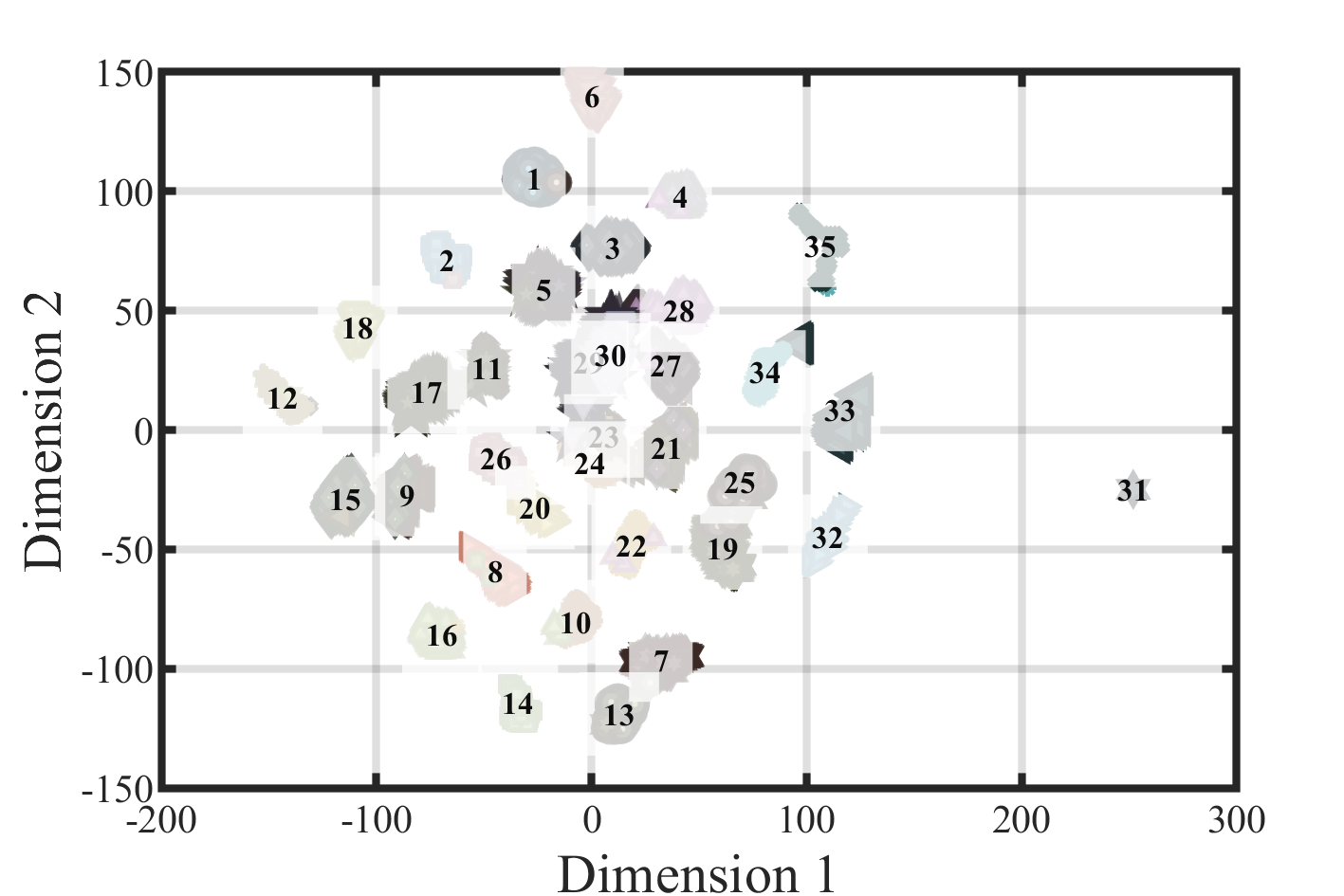}}
\\
\includegraphics[width=\linewidth]{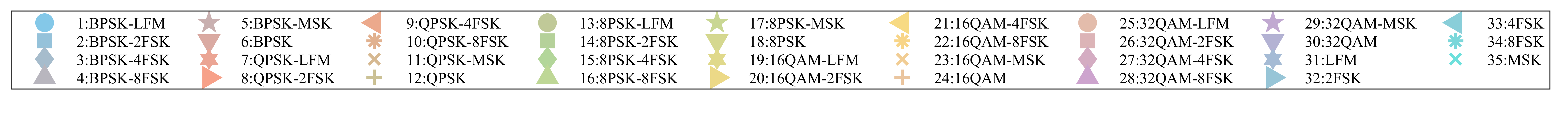}
\caption{Visualization of semantic features for known signals. The plot includes 25 composite modulation signals and 10 single modulation signals. }
\label{fig:visual}
\end{figure*}

This experiment evaluates the quality of the learned semantic representations produced by the proposed tangent space disentanglement framework.
Since the effectiveness of the compositional zero-shot pipeline in Section~\ref{sec:framework} hinges on whether disentangled features form well-structured, layer-wise separable clusters, we visualize the semantic space under a fully supervised setting to isolate feature quality from zero-shot generalization effects.

We adopt a closed-set classification scenario covering all 25 CM and 10 SM types defined in Table~\ref{tab:dataset}.
No additive noise is applied in this experiment to focus on the intrinsic separability of the learned representations.
The proposed method is compared against several representative deep learning baselines, containing
CNN1~\cite{o2016convolutional}, CNN2~\cite{tekbiyik2020robust}, MCNet~\cite{huynh2020mcnet}, ICAMC~\cite{hermawan2020cnn}, ResNet~\cite{liu2017deep}, DenseNet~\cite{liu2017deep},
GRU~\cite{hong2017automatic} and BiGRU~\cite{chang2022hierarchical}.
All models are trained under identical conditions with the same optimizer, learning rate schedule, and number of epochs to ensure a fair comparison.

To provide an intuitive assessment of the learned representations, we project the high-dimensional feature vectors onto a two-dimensional plane using t-distributed stochastic neighbor embedding (t-SNE)~\cite{van2008visualizing}.
Fig.~\ref{fig:visual} presents the resulting embeddings for the proposed method and three representative baselines.
Several qualitative observations emerge from the visualization.
First, the proposed method produces visibly tighter intra-class clusters: samples belonging to the same type are densely concentrated, indicating that the tangent space projection and disentanglement module extract consistent, modulation-specific features.
Second, the inter-class separation is substantially wider: distinct CM types form well-separated regions in the embedding space, with minimal overlap between clusters.
Third, a notable structural pattern is observed: CM types sharing the same inner-layer modulation (e.g., BPSK-LFM, BPSK-2FSK, BPSK-4FSK, BPSK-8FSK, BPSK-MSK) tend to form sub-clusters within a broader region, reflecting the semantic space follows the physics rule.

\subsection{Experiment 2: Ablation Study on the Proposed Framework}
\label{subsec:exp_ablation}

This experiment systematically validates three design choices underlying the proposed method: the disentangled semantic space versus traditional semantic space, the dual-branch disentanglement strategy versus alternative multi-task architectures, and the lightweight convolutional backbone versus deeper network architectures.
We adopt the one-CM-held-out protocol with BPSK-LFM as the unseen class, i.e., ${\mathcal Y}_{\rm CM}^{\rm u} = \{(\text{BPSK}, \text{LFM})\}$, and report results under a noise-free setting to isolate the contribution of each design component.

\subsubsection{Disentangled Semantic Space vs. Unified Semantic Space}
\label{subsubsec:zsl_framework}

\begin{figure}
    \centering
    \includegraphics[width=\linewidth]{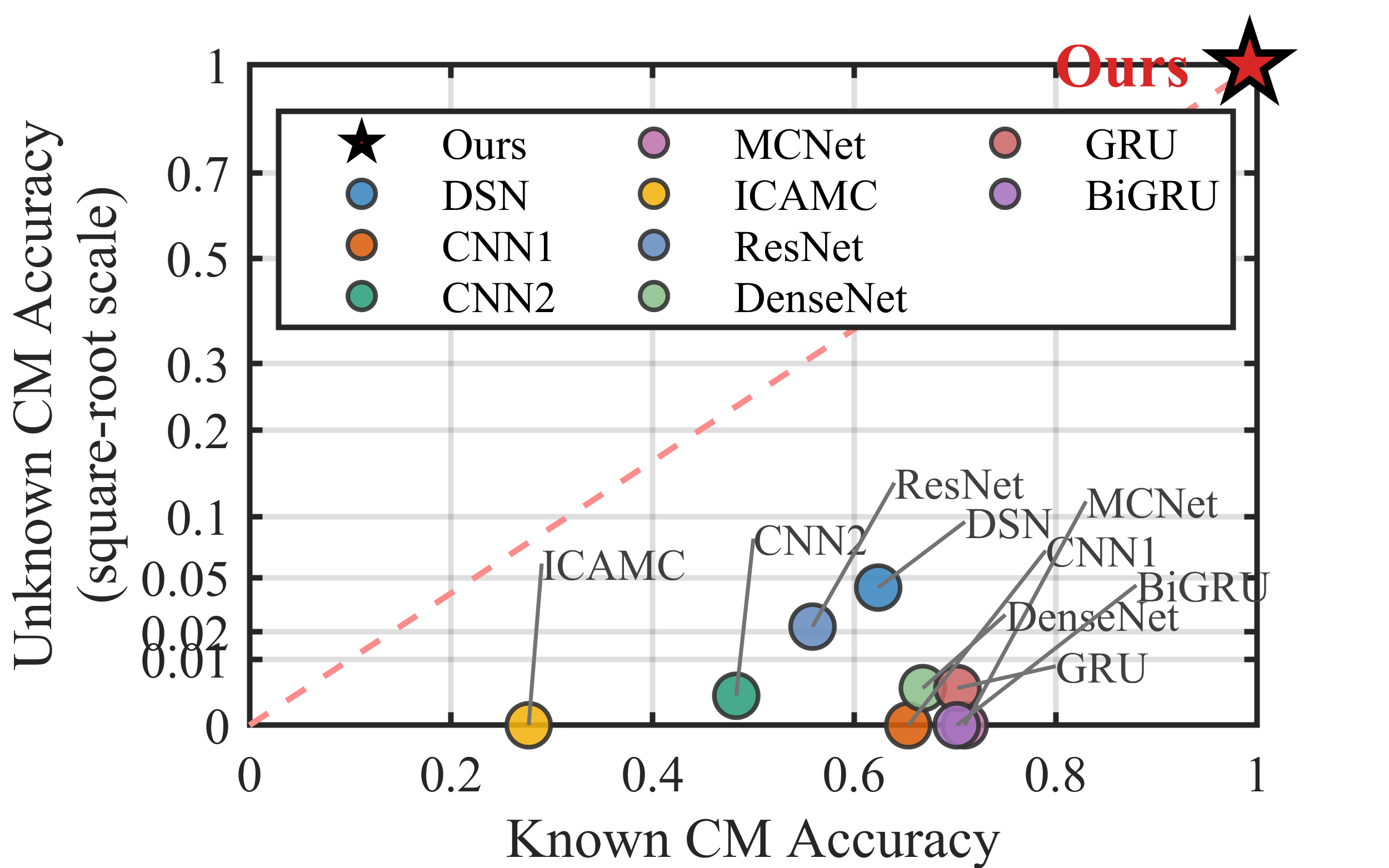}
    \caption{Scatter plot comparing model performance on known and unknown targets. The x-axis represents performance on known targets, while the y-axis represents performance on unknown targets.}
    \label{fig:compare_semantic}
\end{figure}

A central claim of this work is that the conventional ZSL paradigm, which maps all classes into a single unified semantic space \cite{Zhao2025ZeroShotAM, sr2cnn}, is ill-suited for CM recognition due to the waveform-level entanglement between modulation layers.
To evaluate this claim, we compare the proposed disentangled semantic space against a representative unified-prototype baseline built on a frozen large language model (LLM).
Specifically, the baseline employs OPT-6.7B~\cite{Zhao2025ZeroShotAM} as a text encoder to map natural-language descriptions of each CM type into semantic prototype vectors.

Fig.~\ref{fig:compare_semantic} presents the results.
The unified-prototype baseline exhibits a sharp degradation in known-class accuracy compared to its supervised counterpart, and nearly fails to recognize the unseen BPSK-LFM combination.
We attribute this phenomenon to a semantic gap. 
The LLM-generated prototypes encode high-level textual semantics that do not correspond to the low-level signal features distinguishing one CM type from another.
In contrast, the proposed disentangled semantic space constructs layer-wise prototypes directly in the learned feature domain, inheriting the geometric structure revealed in Experiment 1.
This eliminates the alignment gap and enables both accurate known-class recognition and effective zero-shot transfer.

\subsubsection{Dual-Branch Disentanglement vs. Alternative Multi-Task Architectures}
\label{subsubsec:disentangle_strategy}

\begin{table*}[t]
\centering
\caption{Experimental results of different methods on Known and Unknown datasets. Accuracy and F1-score (F1) are reported for each label. The best results are highlighted in \textbf{bold}.}
\label{tab:performance_comparison}
\small 
\resizebox{0.95\textwidth}{!}{
\begin{tabular}{ccccccccc} 
\toprule
\multirow{2}{*}{{Dataset}} & \multirow{2}{*}{{Method}}  & \multicolumn{2}{c}{{\thead{Inner-Layer}}} & \multicolumn{2}{c}{{\thead{Outer-Layer}}} & \multicolumn{2}{c}{{\thead{CM}}} \\ 
 \cmidrule(lr){3-4}  \cmidrule(lr){5-6} \cmidrule(lr){7-8}
  &  & Accuracy & F1-score & Accuracy & F1-score & Accuracy & F1-score \\ 
\midrule
\multirow{3}{*}{Known} & \thead{Factorized \\ Classifier} &  0.9913 $\pm$ 0.0056 & 0.9916 $\pm$ 0.0054 & 0.9965 $\pm$ 0.0029 & 0.9963 $\pm$ 0.0030 & 0.9878 $\pm$ 0.0082 & 0.9873 $\pm$ 0.0085 \\
 & \thead{Dual-head \\ Classifier} & 0.9909 $\pm$ 0.0037 & 0.9912 $\pm$ 0.0035 & 0.9974 $\pm$ 0.0017 & 0.9973 $\pm$ 0.0018 & 0.9885 $\pm$ 0.0037 & 0.9880 $\pm$ 0.0037 \\
 & \textbf{Ours} & 0.9933$\pm$0.0026 & 0.9935 $\pm$ 0.0025 &  0.9991 $\pm$ 0.0004 & 0.9991 $\pm$ 0.0004 & 0.9924 $\pm$ 0.0024 & 0.9922 $\pm$ 0.0024 \\ 
\midrule
\multirow{3}{*}{Unknown} 
& \thead{Factorized \\ Classifier }
& 0.1507 $\pm$ 0.6473 & 0.2077 $\pm$ 0.8916 & 1.0000 $\pm$ 0.0000 & 1.0000 $\pm$ 0.0000 & 0.1507 $\pm$ 0.6473 & 0.2077 $\pm$ 0.8916 \\
 & \thead{Dual-head \\ Classifier}  
 & 0.0149 $\pm$ 0.0494 & 0.0289 $\pm$ 0.0950 & 1.0000 $\pm$ 0.0000 & 1.0000 $\pm$ 0.0000 & 0.0149 $\pm$ 0.0494 &  0.0289 $\pm$ 0.0950 \\
 & \textbf{Ours} & \textbf{1.0000$\pm$1.0000} & \textbf{1.0000$\pm$1.0000} & \textbf{1.0000$\pm$1.0000} & \textbf{1.0000$\pm$1.0000} & \textbf{1.0000$\pm$1.0000} & \textbf{1.0000$\pm$1.0000} \\ 
\bottomrule
\end{tabular}
}
\end{table*}

Having established the advantage of the disentangled semantic space, we now investigate how the disentanglement should be implemented.
Since the composite label $(y_{\rm in}, y_{\rm out})$ involves two classification tasks, this can naturally be formulated as a multi-task learning problem.
We compare the proposed dual-branch architecture against two mainstream alternatives:
(1) Factorized classifier~\cite{yang2017deep}: a single shared backbone feeds into one unified classifier head that outputs a joint logit matrix, implicitly modeling inter-task dependencies through weight factorization.
(2) Dual-head classifier~\cite{kendall2018multi}: a single shared backbone feeds into two independent classifier heads, each producing logits for one modulation layer.
The proposed method differs from both by employing two independent feature branches, each dedicated to one modulation layer, followed by layer-specific classifiers.
This design is motivated by the physical observation that the inner-layer (amplitude-phase modulation) and outer-layer (frequency modulation) occupy fundamentally different signal characteristics, as established in Section~\ref{sec:system_model}.

Table~\ref{tab:performance_comparison} reports the accuracy and macro F1-score for both known and unseen modulation types, evaluated at the inner-layer, outer-layer, and composite CM levels.
Both the factorized and dual-head classifiers achieve competitive performance on known classes, confirming that the multi-task formulation is well-posed.
However, they exhibit substantial degradation on the unseen BPSK-LFM type, particularly at the outer-layer level.
We attribute this to the shared backbone: when a single feature extractor must simultaneously represent both modulation layers, the learned features form a compromise that captures the dominant patterns of known combinations but fails to generalize to the specific inner--outer pairing absent from training.
The proposed dual-branch architecture avoids this bottleneck by learning specialized representations for each layer, yielding superior zero-shot performance while maintaining known-class accuracy.

\subsubsection{Backbone Architecture Comparison}
\label{subsubsec:backbone}

\begin{figure}[!t]
\centering
\subfloat[]{\includegraphics[width=\columnwidth]{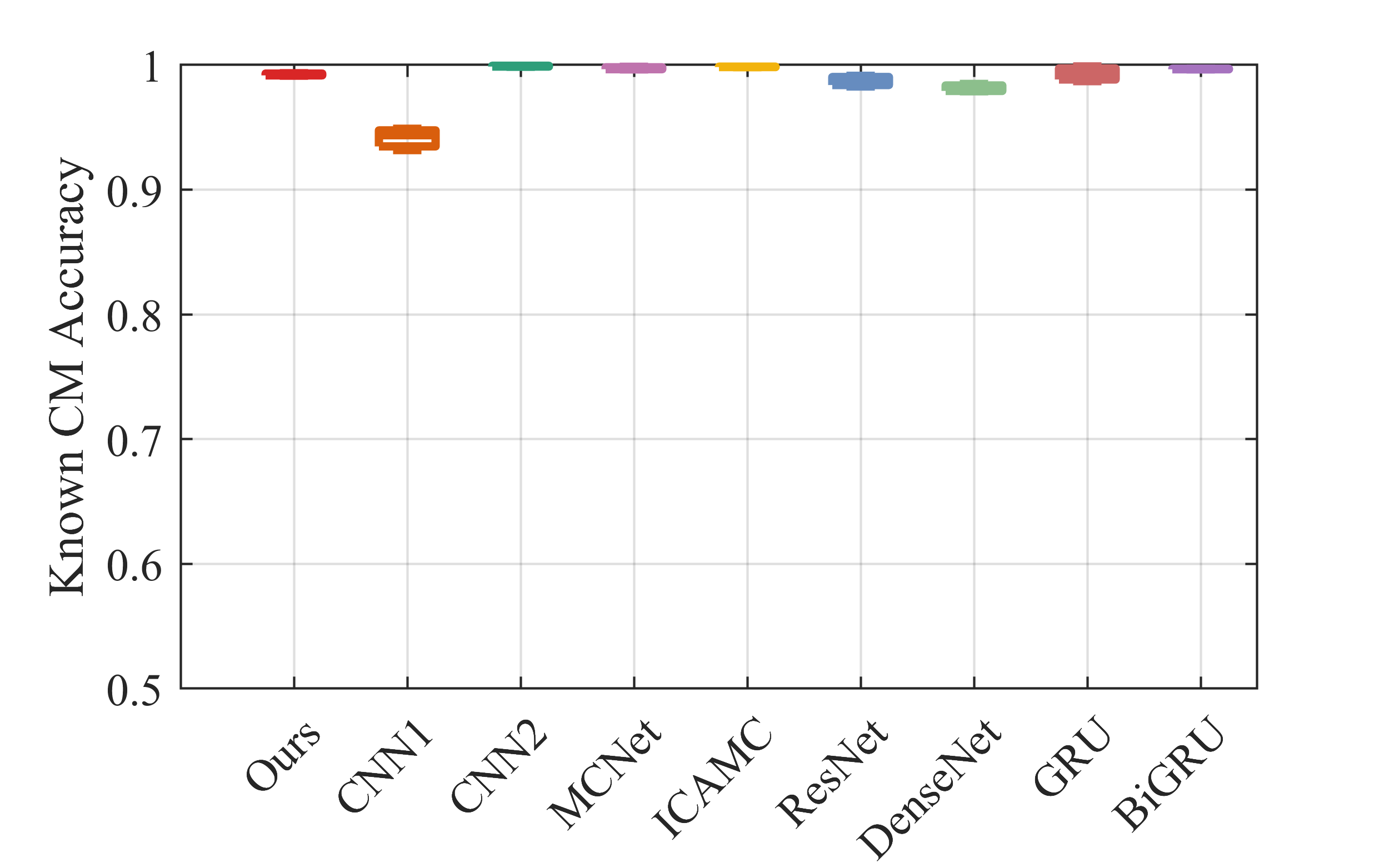}}
\\
\subfloat[]{\includegraphics[width=\columnwidth]{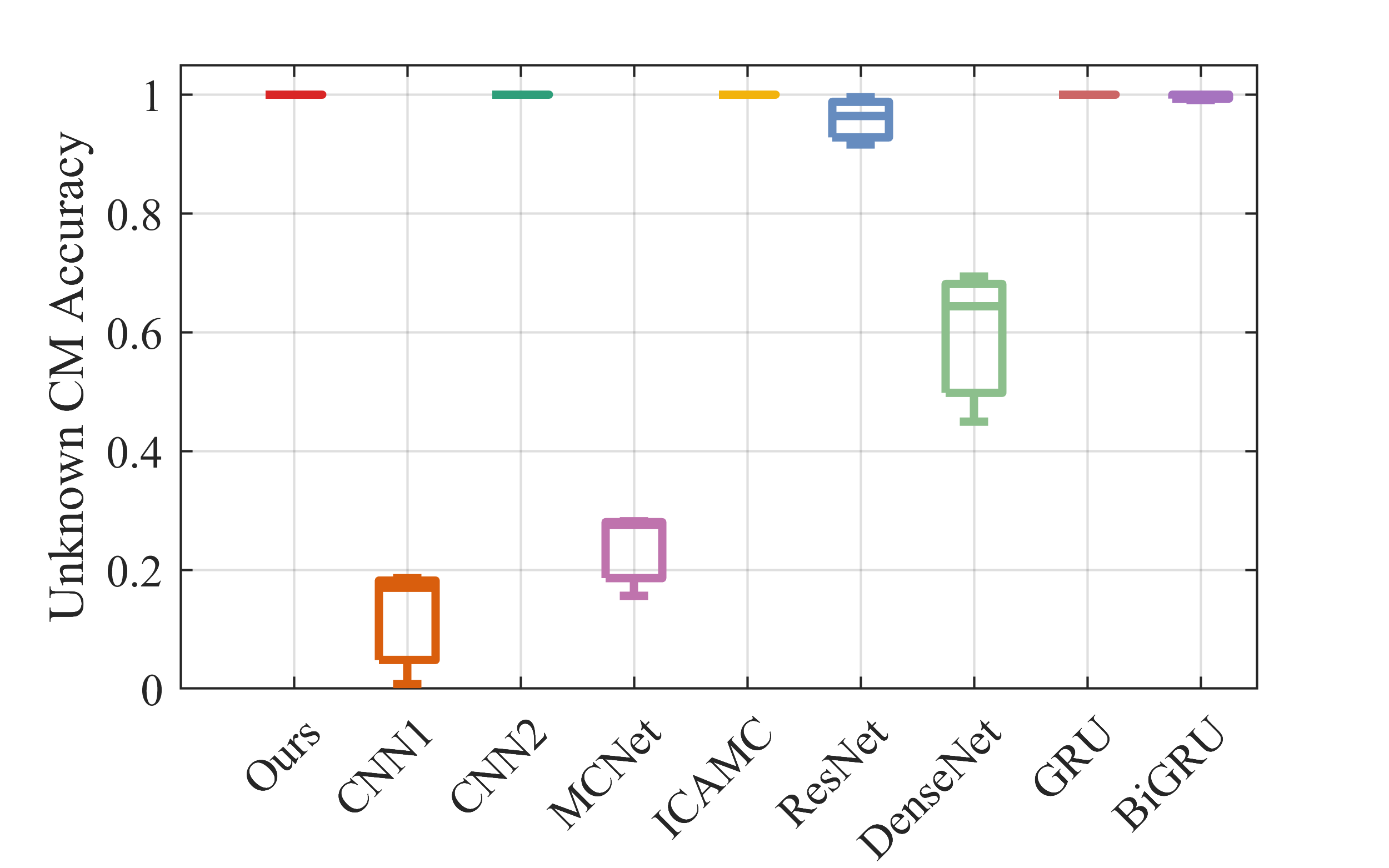}}
\\
\caption{Comparison of different neural network architectures. (a) Performance for known accuracy. (b) Performance for unknown accuracy.}
\label{fig:known_recog}
\end{figure}

\begin{table}[!t]
\centering
\caption{Comparison of model complexity. Lower values indicate lighter models.}
\label{tab:complexity}
\begin{tabular}{lcc}
\toprule
Model & Params & FLOPs \\
\midrule
{Ours}           & {167K  }          & {77.4M} \\
CNN1           & 465K            & 12.9M \\
CNN2           & 1.03M           & 163.5M \\
MCNet          & 1.01M           & 22.5M \\
ICAMC          & 1.51M           & 118.8M \\
ResNet         & 1.64M           & 1.38G \\
DenseNet       & 2.01M           & 1.76G \\
GRU            & 399K            & 155.6M \\
BiGRU          & 301K            & 105.3M \\
\bottomrule
\end{tabular}
\end{table}

The final ablation dimension concerns the choice of neural network backbone.
The proposed method employs a lightweight convolutional architecture optimized for the tangent-space representation.
To justify this choice, we compare against several widely-used backbones in the AMR literature, as presented in Fig.~\ref{fig:known_recog}.
Table~\ref{tab:complexity} reports the zero-shot recognition accuracy alongside two complexity metrics, i.e. parameter count (Params) and floating-point operations (FLOPs).
While several deeper architectures achieve comparable known-class accuracy, their zero-shot performance is markedly lower, suggesting that over-parameterized backbones tend to overfit to the specific training combinations at the expense of generalization.

\subsection{Experiment 3: Ablation Study of TSDN Component}

The proposed framework comprises three architectural modules, namely the backbone network, the logarithmic mapping, and the spatial transformer network STN, as well as three loss components $\mathcal{L}_1$, $\mathcal{L}_2$, and $\mathcal{L}_3$.
This experiment evaluates the contribution of each component through a systematic ablation study.
We adopt the one-CM-held-out protocol with BPSK-LFM as the unseen class at SNR = 8 dB.
All results are averaged over three random seeds; both accuracy and macro F1-score are reported for known and unknown classes.
The complete ablation results are presented in Table~\ref{tab:ablation}.

\begin{table*}[!t]
    \centering
    \caption{Ablation study of TSDN component. BPSK-LFM is held out as the unseen class at SNR = 8 dB. All metrics are reported as {mean $\pm$ std} across three random seeds with a 95\% confidence interval.}
    \label{tab:ablation}
    
    \begin{tabular}{lcccccccccc}
        \toprule
        \multirow{2}{*}{Variant} & \multicolumn{3}{c}{Modules} & \multicolumn{3}{c}{Loss} & \multicolumn{2}{c}{Known} & \multicolumn{2}{c}{Unknown} \\
        \cmidrule(lr){2-4} \cmidrule(lr){5-7}  \cmidrule(lr){8-9} \cmidrule(lr){10-11}
        & Backbone & Log & STN & ${\mathcal L}_1$ & ${\mathcal L}_2$& ${\mathcal L}_3$ &  Accuracy & Macro-F1   & Accuracy & F1-score  \\
        \midrule
        w/o Log  &  \checkmark & -- &  \checkmark 
        &  \checkmark &  \checkmark &  \checkmark 
        & ${0.9562\pm0.0518}$  & ${0.9543 \pm 0.0530}$
        & $0.8351\pm0.3135$ & $0.9066\pm0.1933$ \\
        w/o STN  &  \checkmark  &  \checkmark & -- 
        &  \checkmark &  \checkmark &  \checkmark 
        & $0.9293\pm0.0062$ & $0.9266\pm0.0067$ 
        & $0.9184\pm0.0582$ & $0.9574\pm0.0318$\\
        w/o Log, STN &  \checkmark & -- & --
        &  \checkmark &  \checkmark &  \checkmark
        & $0.9318\pm0.0174$ & $0.9291\pm0.0184$
        & $0.8236\pm0.0993$ & $0.9029\pm0.0605$\\
        w/o ${\mathcal L}_3$ &  \checkmark &  \checkmark &  \checkmark
        &  \checkmark & \checkmark &  -- 
        & $0.9472\pm0.0435$ & $0.9451\pm0.0447$ 
        & $0.7790\pm0.3859$ & $0.8703\pm0.2336$\\
        w/o ${\mathcal L}_2$, ${\mathcal L}_3$ &  \checkmark &  \checkmark &  \checkmark
        &  \checkmark & -- &  -- 
        & $0.9468\pm0.0447$ & $0.9448\pm0.0460$ 
        & $0.7964\pm0.3314$ & $0.8827\pm0.1992$ \\
        \textbf{Ours} &  \checkmark &  \checkmark &  \checkmark
        &  \checkmark & \checkmark &  \checkmark
        & $0.9354\pm0.0035$ & $0.9335 \pm 0.0038$ 
        & $\bf{0.9352 \pm0.0310} $ & $\bf{0.9665 \pm 0.0166}$\\
        \bottomrule
    \end{tabular}
\end{table*}

Among the architectural modules, the logarithmic mapping is the most critical: removing it drops the unknown accuracy from 93.52\% to 83.51\% with a standard deviation of 31.35\%, validating the theoretical role of tangent-space projection in Section~\ref{sec:framework}.
Removing the STN causes a moderate drop at 91.84\%, and jointly removing both w/o Log and STN yields accuracy comparable to removing Log alone at 82.36\%, confirming that Log is the dominant factor while the STN provides complementary disentanglement refinement as analyzed in Section~\ref{sec:framework}.
Among the loss components, ${\mathcal L}_3$ is crucial: its removal produces the largest single-component degradation, while further removing ${\mathcal L}_2$ yields no significant additional drop, indicating that ${\mathcal L}_3$ carries the primary regularization effect on the semantic space.
Across all ablations, the full model achieves the highest unknown accuracy of 93.52\% with the lowest variance at 3.10\%, confirming that every component contributes synergistically to compositional generalization.

\subsection{Experiment 4: Generalization across Modulation Type}

\subsubsection{Discrimination under High-Order Modulation Schemes}
\label{subsubsec:exp_highorder}

\begin{figure}[!t]
    \centering
    \subfloat[]{\includegraphics[width=0.9\columnwidth]{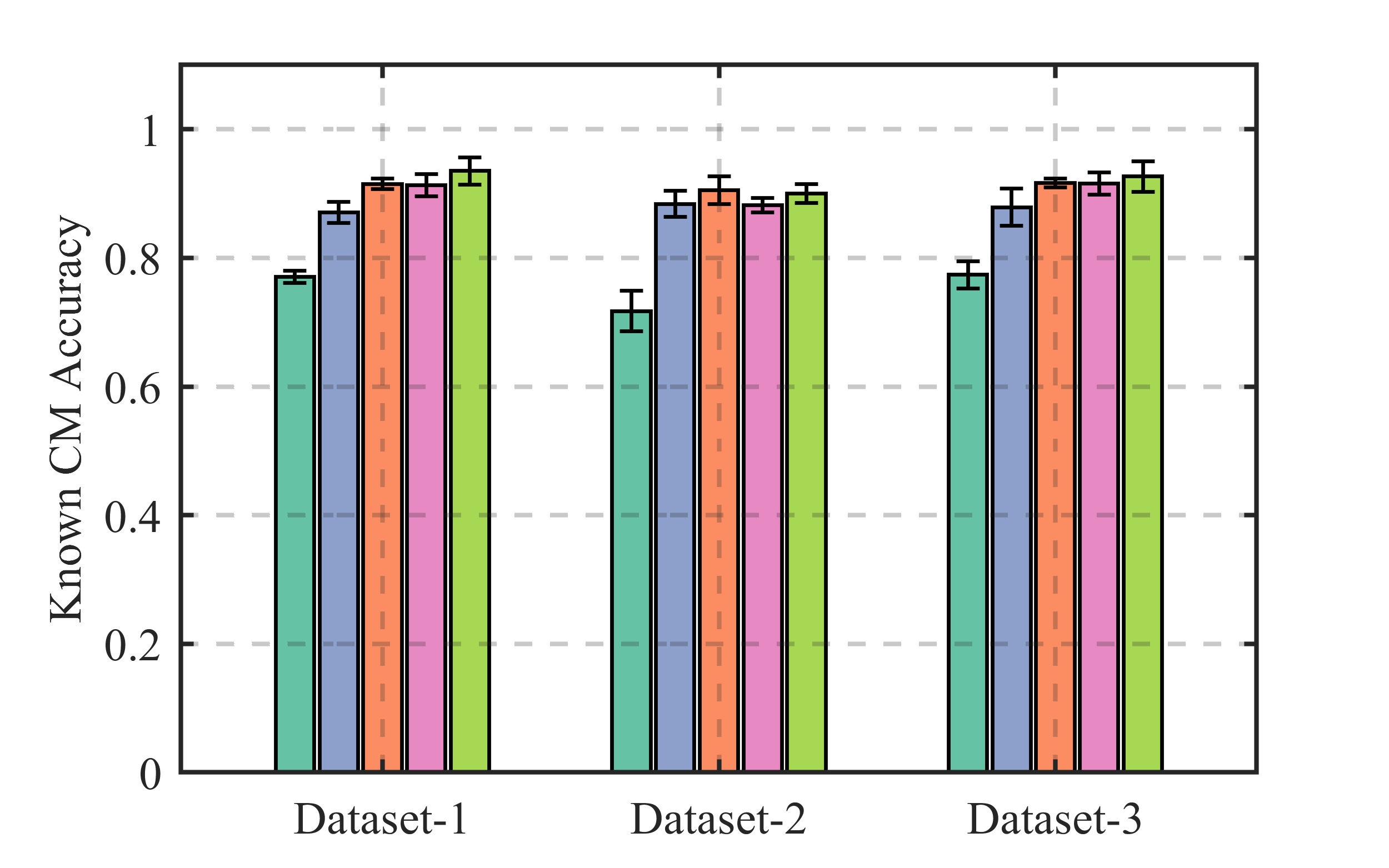}}
    \\
    \subfloat[]{\includegraphics[width=0.9\columnwidth]{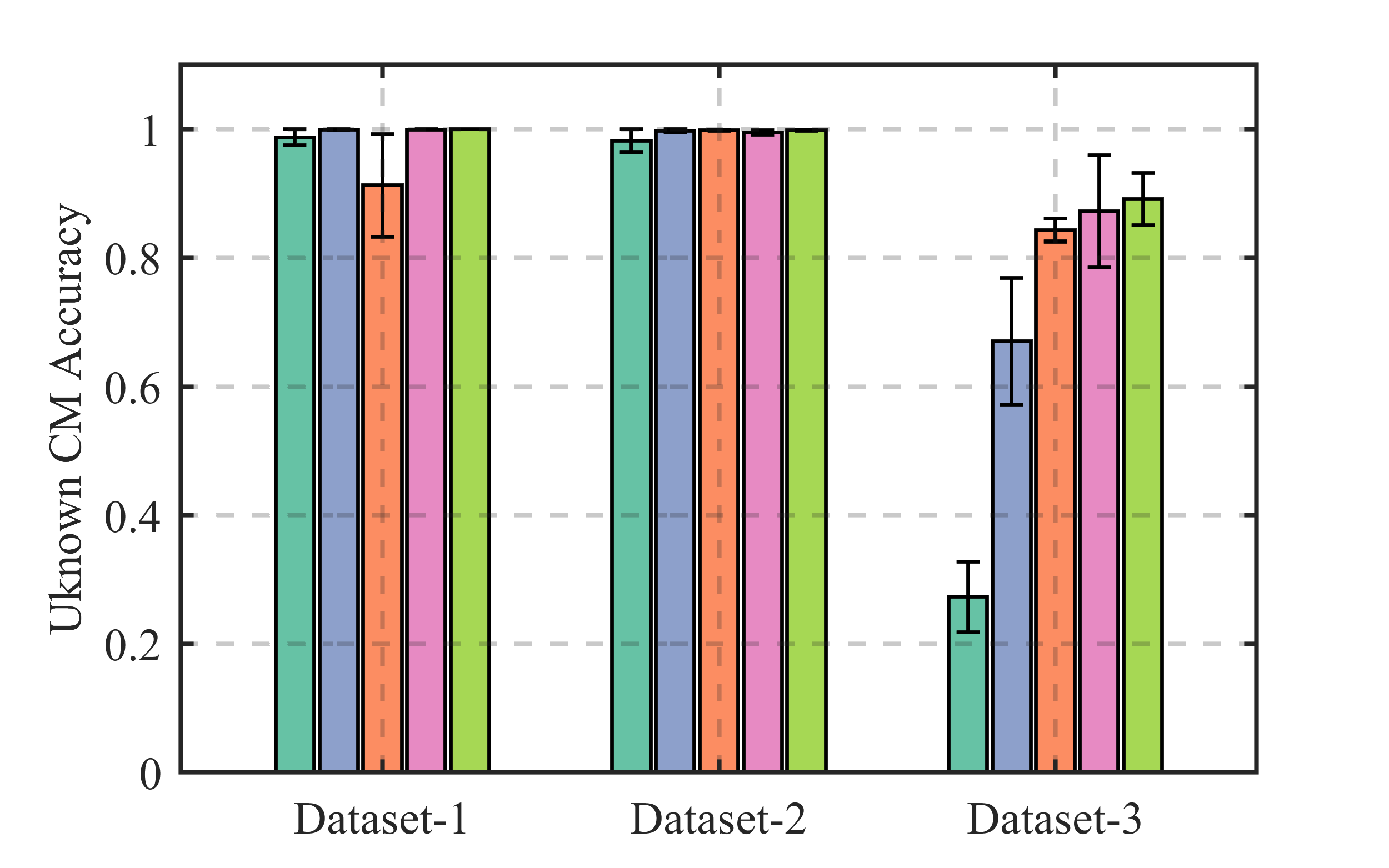}}
    \\
    \includegraphics[width=0.9\linewidth]{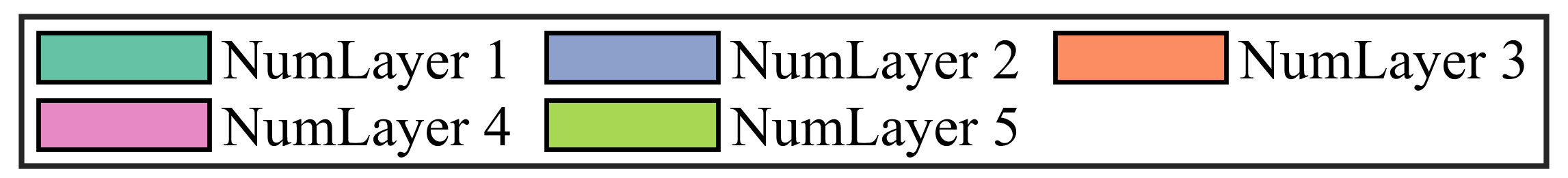}
    \caption{Model performance across high-order modulations with varying unknown items. Dataset 1: BPSK-LFM, Dataset 2: 8PSK-2FSK, and Dataset 3: 128QAM-8FSK.}
    \label{fig:high-order}
\end{figure}

The preceding experiments were conducted on a moderate-scale dataset ($5 \times 5 = 25$ CM types).
This experiment extends the evaluation to a more challenging scenario by introducing high-order QAM schemes and Gaussian-filtered frequency shift keying (GFSK), which produce extremely dense constellation diagrams and spectrally overlapping features that are difficult to distinguish from conventional FSK variants.

\begin{table*}[!t]
\centering
\caption{Extended CM dataset with high-order modulation schemes. 
Inner-layer set \({\mathcal Y}_{\rm in} =\) \{BPSK, QPSK, 8PSK, 16QAM, 32QAM, 64QAM, 128QAM\}, outer-layer set \({\mathcal Y}_{\rm out}\) = \{LFM, 2FSK, 4FSK, 8FSK, MSK, GFSK\}, yielding $7 \times 6 = 42$ CM types.}
\label{tab:dataset_highorder}
\small
\setlength{\tabcolsep}{4.5pt}
\renewcommand{\arraystretch}{1.15}
\begin{tabular}{c|ccccccc}
\toprule
\textbf{Inner \textbackslash Outer} & \textbf{LFM} & \textbf{2FSK} & \textbf{4FSK} & \textbf{8FSK} & \textbf{MSK} & \textbf{GFSK} & \textbf{NONE}\\
\midrule
BPSK   & {\color{black} Dataset 1}     & BPSK-2FSK  & BPSK-4FSK  & BPSK-8FSK  & BPSK-MSK  & BPSK-GFSK & BPSK\\
QPSK   & QPSK-LFM  & QPSK-2FSK  & QPSK-4FSK  & QPSK-8FSK  & QPSK-MSK  & QPSK-GFSK & QPSK\\
8PSK   & 8PSK-LFM  & {\color{black} Dataset 2}    & 8PSK-4FSK  & 8PSK-8FSK  & 8PSK-MSK  & 8PSK-GFSK & 8PSK\\
16QAM  & 16QAM-LFM & 16QAM-2FSK & 16QAM-4FSK & 16QAM-8FSK & 16QAM-MSK & 16QAM-GFSK  & 16QAM\\
32QAM  & 32QAM-LFM & 32QAM-2FSK & 32QAM-4FSK & 32QAM-8FSK & 32QAM-MSK & 32QAM-GFSK  & 32QAM\\
64QAM  & 64QAM-LFM & 64QAM-2FSK & 64QAM-4FSK & 64QAM-8FSK & 64QAM-MSK & 64QAM-GFSK  & 64QAM\\
128QAM & 128QAM-LFM& 128QAM-2FSK& 128QAM-4FSK& {\color{black} Dataset 3}  & 128QAM-MSK& 128QAM-GFSK & 128QAM\\
NONE & {LFM} & {2FSK} & {4FSK} & {8FSK} & {MSK} & {GFSK} &  \textbackslash \\
\bottomrule
\end{tabular}
\vspace{2pt}

{\small Dataset~1 (low complexity), Dataset~2 (medium complexity), Dataset~3 (high complexity).}
\end{table*}

Table~\ref{tab:dataset_highorder} presents the expanded design space, where the inner-layer set grows from 5 to 7 types by adding 64QAM and 128QAM, and the outer-layer set grows from 5 to 6 types by adding GFSK, yielding 42 CM types in total.
Three evaluation datasets are constructed by holding out a progressively more challenging CM type as unseen:
\begin{itemize}
    \item \textbf{Dataset 1} ($\mathcal{D}_1$): BPSK-LFM held out. Both constituent modulations are low-order and well-separated in the feature space, representing the easiest zero-shot scenario.
    \item \textbf{Dataset 2} ($\mathcal{D}_2$): 8PSK-2FSK held out. The 8PSK is moderately complex, testing whether the disentangled semantic space can generalize across increasing constellation density.
    \item \textbf{Dataset 3} ($\mathcal{D}_3$): 128QAM-8FSK held out. Both the ultra-high-order 128QAM and 8FSK pose significant discrimination challenges, representing the hardest zero-shot scenario.
\end{itemize}

Fig.~\ref{fig:high-order} reports the unknown-class recognition accuracy as a function of the number of convolutional layers in the backbone, evaluated on each dataset at SNR = 8 dB.
A consistent trend is observed across all three datasets: a single convolutional layer is insufficient when the dataset contains high-order or spectrally dense modulation types, yielding a pronounced accuracy drop particularly for $\mathcal{D}_3$.
Increasing the network depth progressively improves performance, with the accuracy curve saturating at three to four layers.
This suggests a hierarchical feature learning process: shallow layers capture low-level spectral patterns shared across modulation types, while deeper layers extract the fine-grained distinctions necessary to discriminate high-order constellations and closely spaced frequency schemes.
The full model with an appropriately deep backbone successfully maintains high zero-shot accuracy even on $\mathcal{D}_3$, confirming the scalability of the proposed framework to challenging modulation sets.

\subsubsection{Performance Analysis across Held-out Modulation Types}

\begin{figure}[!t]
    \centering
    \includegraphics[width=0.9\linewidth]{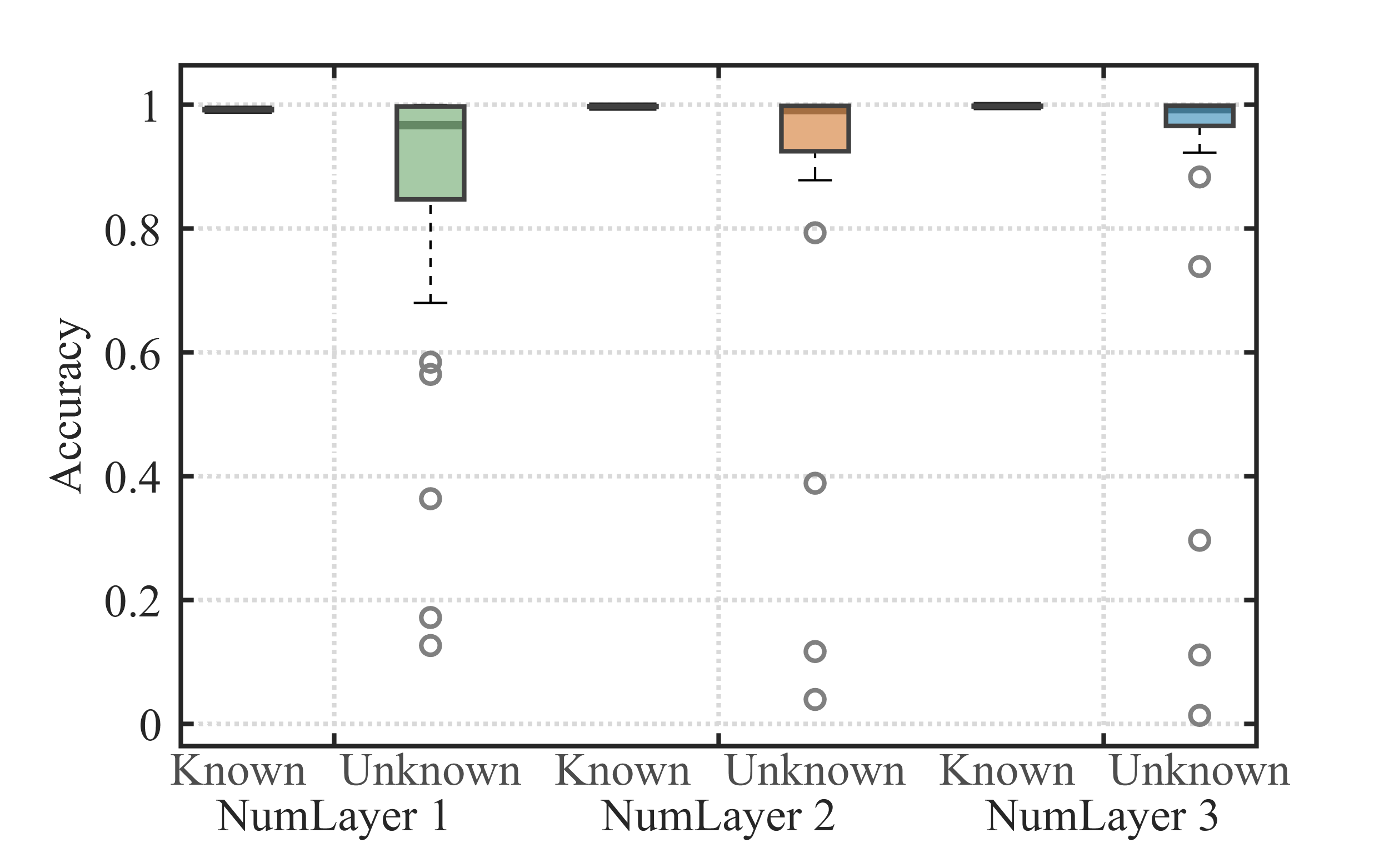}
    \caption{Distribution of known-class and unknown-class accuracy across all held-out CM types, evaluated at different backbone depths (1-5 convolutional layers). Each box aggregates results over all possible unseen CM combinations.}
    \label{fig:imbalance_1}
\end{figure}

\begin{figure*}[!t]
    \centering
    \subfloat[Branch 1]{\includegraphics[height=2.2in]{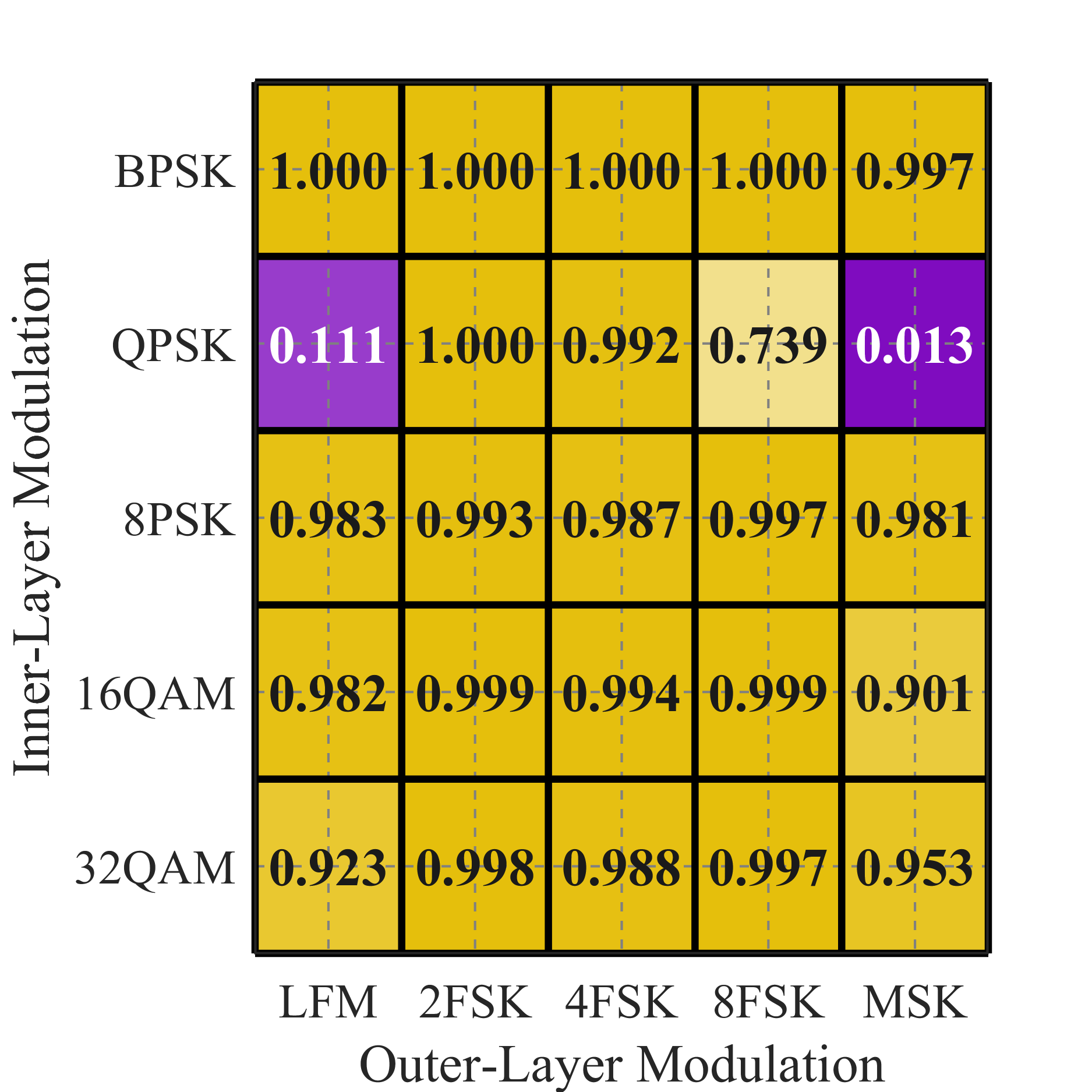}}
    \subfloat[Branch 2]{\includegraphics[height=2.2in]{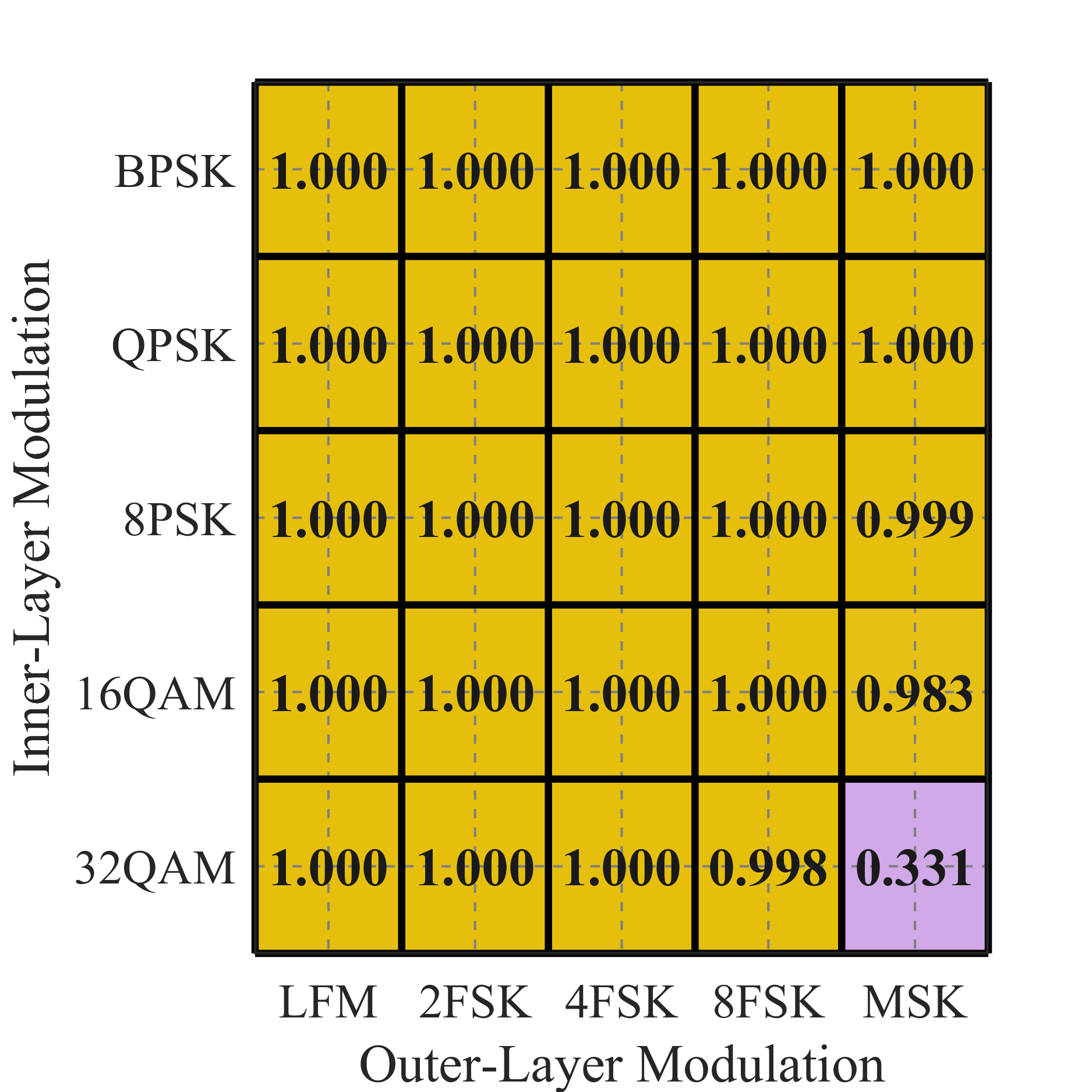}}
    \subfloat[Composite]{\includegraphics[height=2.2in]{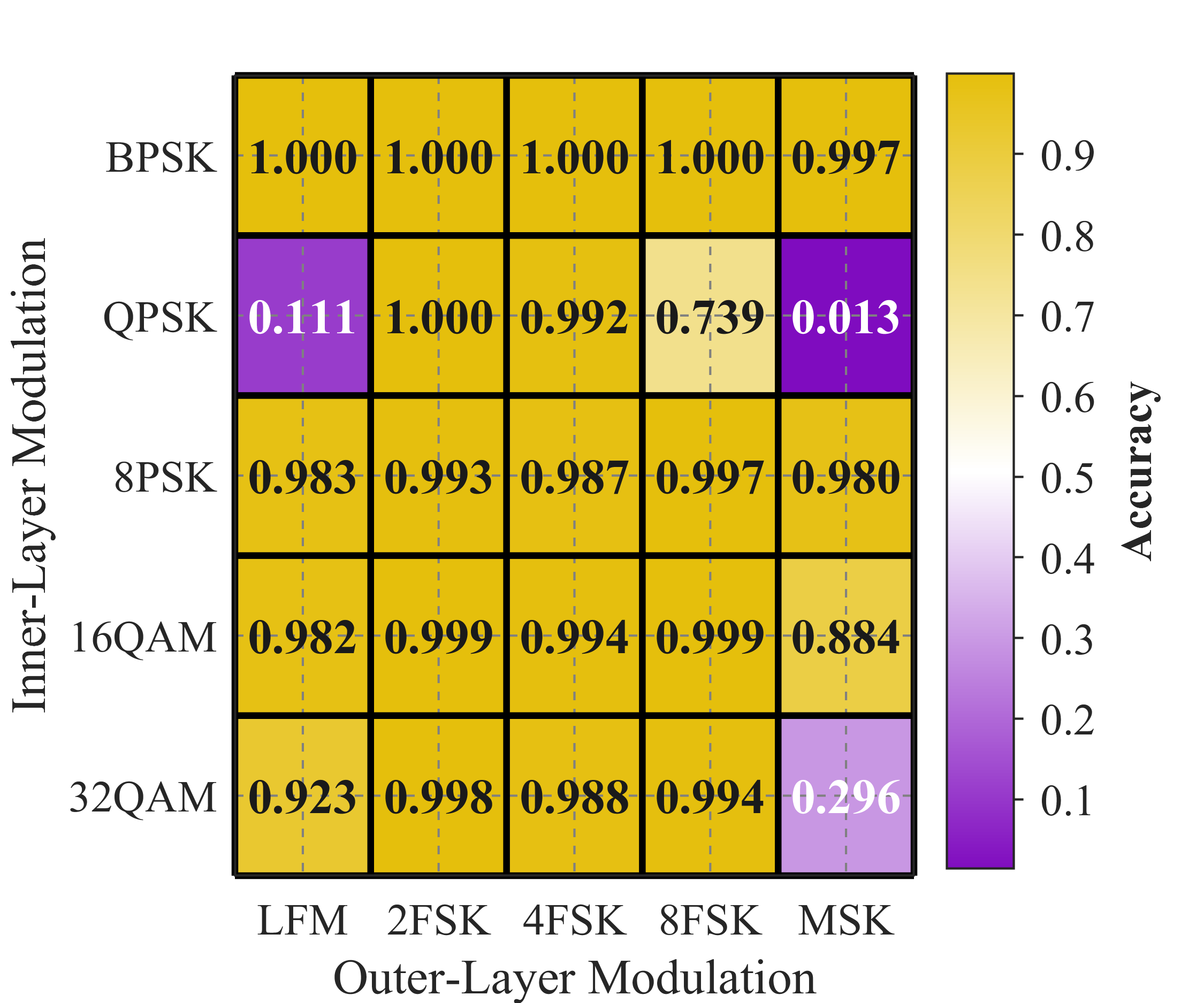}}
    \caption{Per-branch accuracy heatmap at backbone depth = 3. Each row corresponds to one CM type held out as unseen. Columns show the accuracy of Branch~1 (inner-layer), Branch~2 (outer-layer), and Composite (CM) prediction. 
    Cells with notably low accuracy indicate challenging disentanglement scenarios.}
    \label{fig:imbalance_2}
\end{figure*}

The zero-shot accuracy is expected to vary depending on which CM type is held out as unseen, since different combinations impose different levels of difficulty on the disentanglement and classification modules.
This sub-experiment conducts a systematic per-CM-type analysis to identify the failure modes and their underlying causes.

Fig.~\ref{fig:imbalance_1} presents box plots of the known-class and unknown-class accuracy distributions, where each box aggregates results over all possible held-out CM types for a given backbone depth from 1 to 5.
Two observations stand out.
First, the known-class accuracy remains consistently high across all depths, confirming that the framework does not sacrifice in-distribution performance for zero-shot capability.
Second, the unknown-class accuracy exhibits considerable spread at shallow depths, with the interquartile range narrowing substantially as the backbone deepens.
This indicates that shallow networks produce feature representations that are sensitive to which specific combination is held out, some combinations are easy to generalize to while others are not, whereas deeper networks learn more robust and transferable semantic features.
At four or more layers, the distribution stabilizes, suggesting that the network has extracted sufficient hierarchical features for reliable compositional generalization.

To pinpoint the source of recognition failures, Fig.~\ref{fig:imbalance_2} presents a heatmap of the average accuracy for each layer-wise branch and their fusion, evaluated at a backbone depth of three layers.
Each row corresponds to one CM type being held out as unseen.
The heatmap reveals three distinct failure modes, where the unknown accuracy drops markedly below the dataset average:
For {QPSK-LFM} and {QPSK-MSK}, the accuracy drop is concentrated in the inner-layer branch (Branch~1), indicating that the inner-layer disentanglement fails specifically when QPSK is paired with these outer-layer schemes.
For {32QAM-MSK}, the accuracy drop is concentrated in the outer-layer branch (Branch~2), indicating that the outer-layer disentanglement fails for the MSK component in this combination.
These observations suggest that the difficulty of disentanglement depends on the interaction between the two layers: a given inner-layer type may be easily separated from the outer layer in most combinations but becomes problematic when paired with specific outer-layer schemes that produce similar spectral features to other known outer types.

Fig.~\ref{fig:imbalance_3} presents the confusion matrices for the two branches in the three problematic cases, providing further insight into the error patterns.
For {QPSK-LFM} held out as unseen, the inner-layer confusion matrix shows that a fraction of known-class 8PSK samples are misclassified as QPSK during training.
This feature overlap propagates to the zero-shot setting: 88.8\% of QPSK samples are misclassified as 8PSK by Branch~1.
The same pattern is even more severe when {QPSK-MSK} is held out, where 98.6\% of QPSK is misclassified as 8PSK.
The key difference between these two cases lies in the outer layer: LFM and MSK exhibit lower spectral similarity to the other known MFSK compared to the inner-layer distinction between QPSK and 8PSK.
When the outer layer using MSK is harder to disentangle, the residual coupling between layers exacerbates the QPSK-8PSK confusion in Branch~1.
This confirms that the quality of outer-layer disentanglement directly affects inner-layer recognition accuracy, consistent with the multiplicative coupling model in Section~\ref{sec:system_model}.
For {32QAM-MSK} held out as unseen, the outer-layer confusion matrix reveals a different failure mode: during known-class recognition, a small fraction of MSK samples are confused with NONE.
This minor known-class overlap amplifies dramatically in the zero-shot setting, where 66.9\% of MSK samples are misclassified as NONE by Branch~2.
This finding highlights a critical insight: even a negligible feature ambiguity in the known classes can cascade into severe zero-shot misclassification, underscoring the importance of learning well-separated semantic prototypes for every known modulation type.

\begin{figure*}[!t]
    \centering
    \subfloat[QPSK-LFM: Branch 1-Known]{\includegraphics[width=0.3\linewidth]{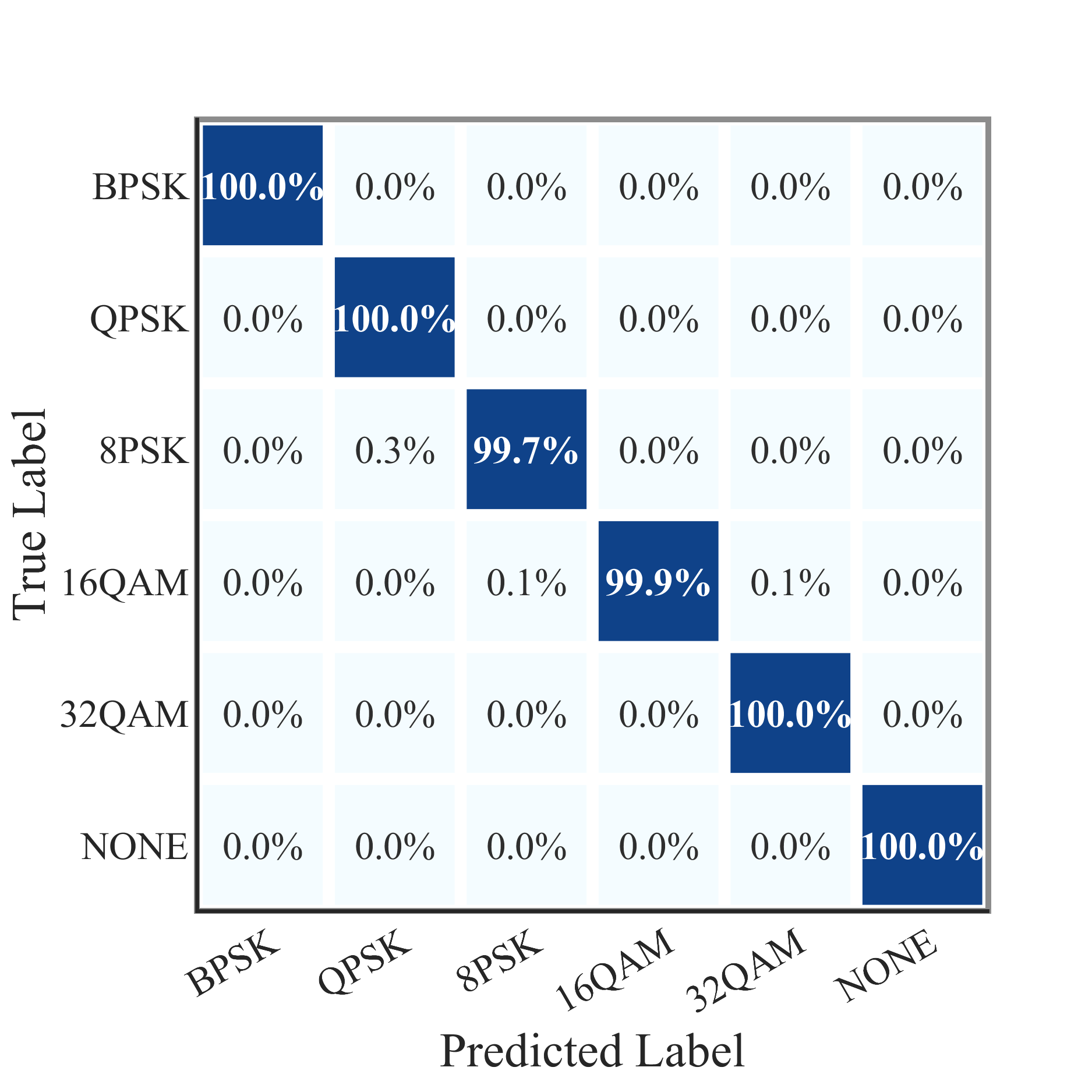}}
    \subfloat[QPSK-MSK: Branch 1-Known]{\includegraphics[width=0.3\linewidth]{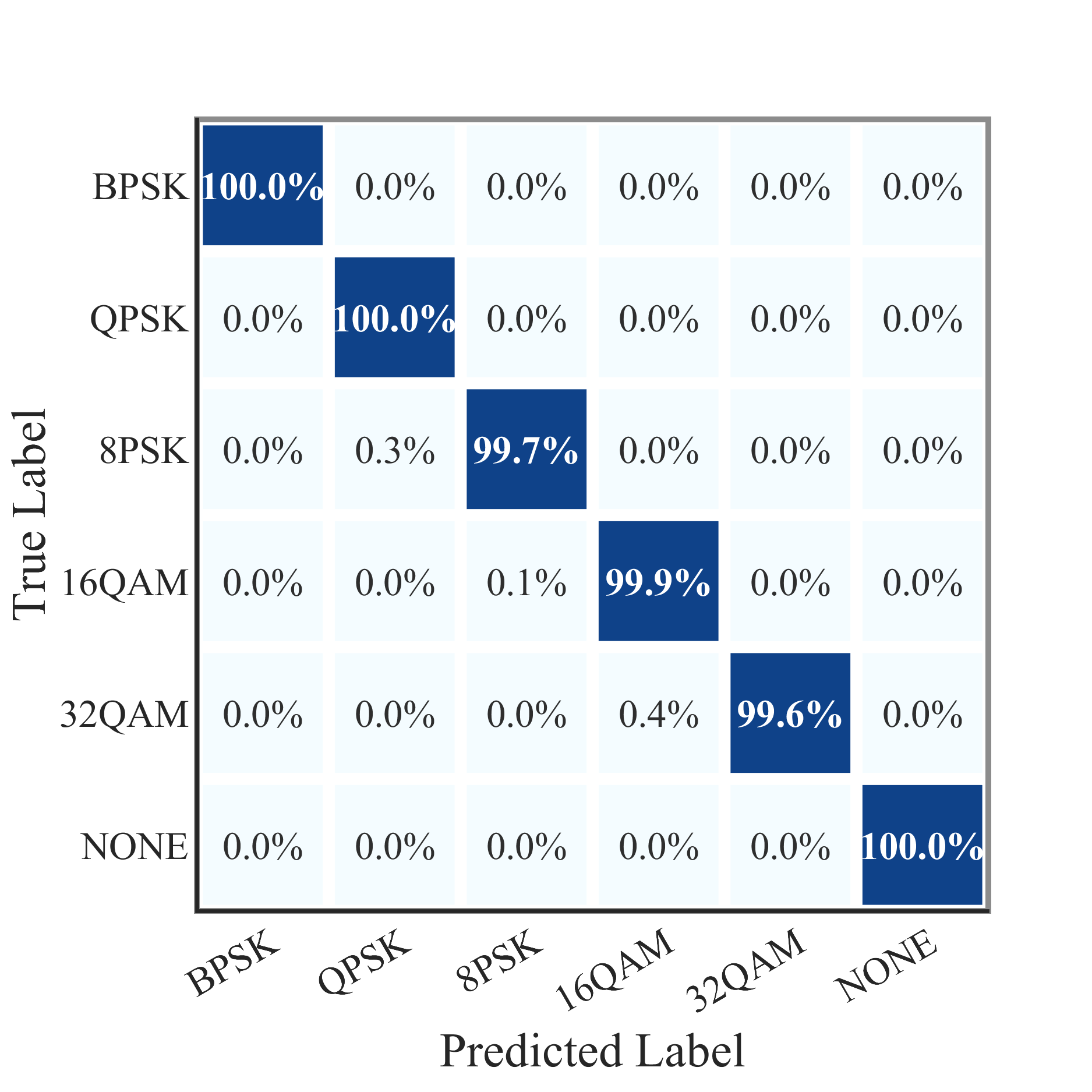}}
    \subfloat[32QAM-MSK: Branch 2-Known]{\includegraphics[width=0.35\linewidth]{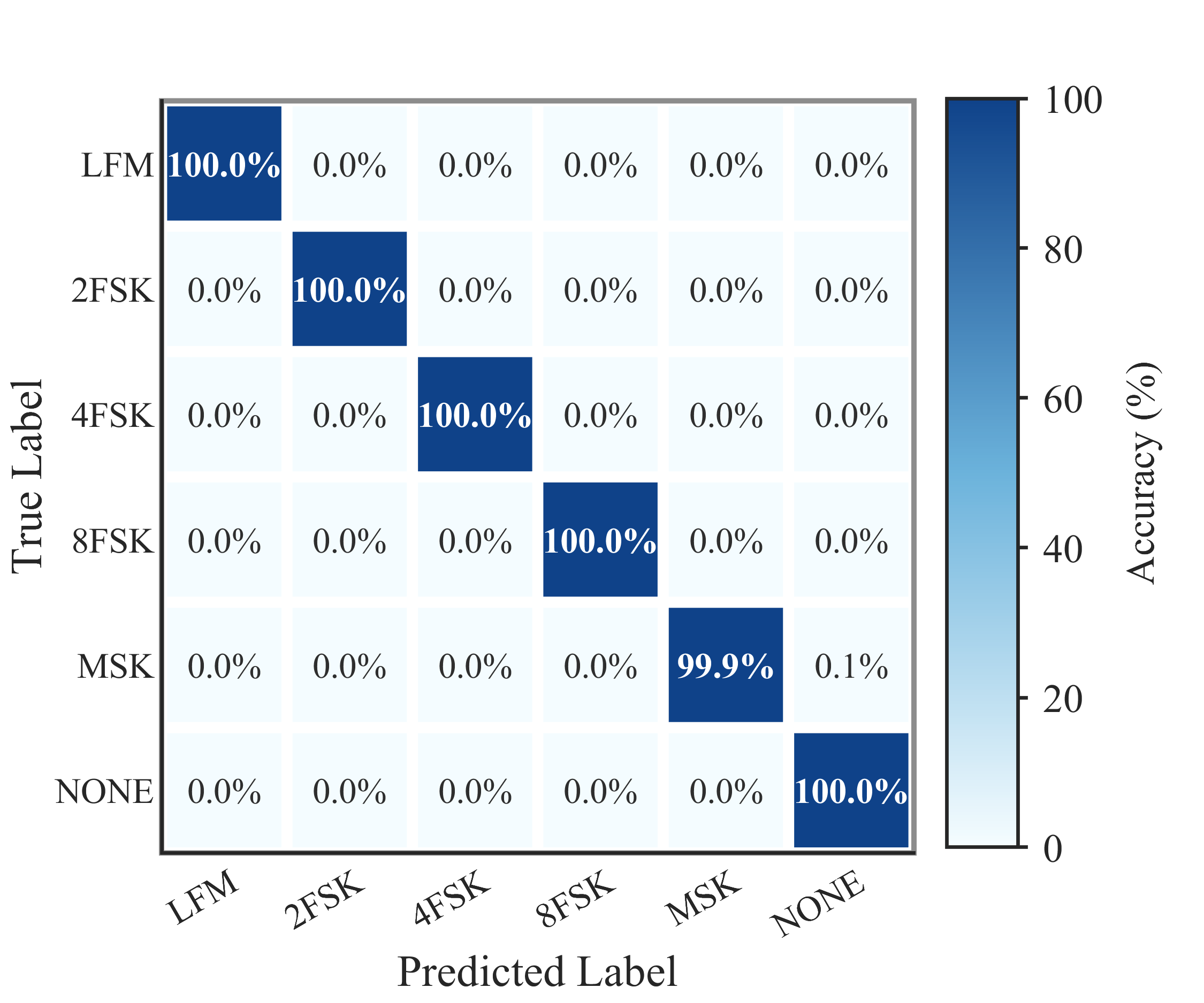}}
    \\
    \subfloat[QPSK-LFM: Branch 1-Unknown]{\includegraphics[width=0.3\linewidth]{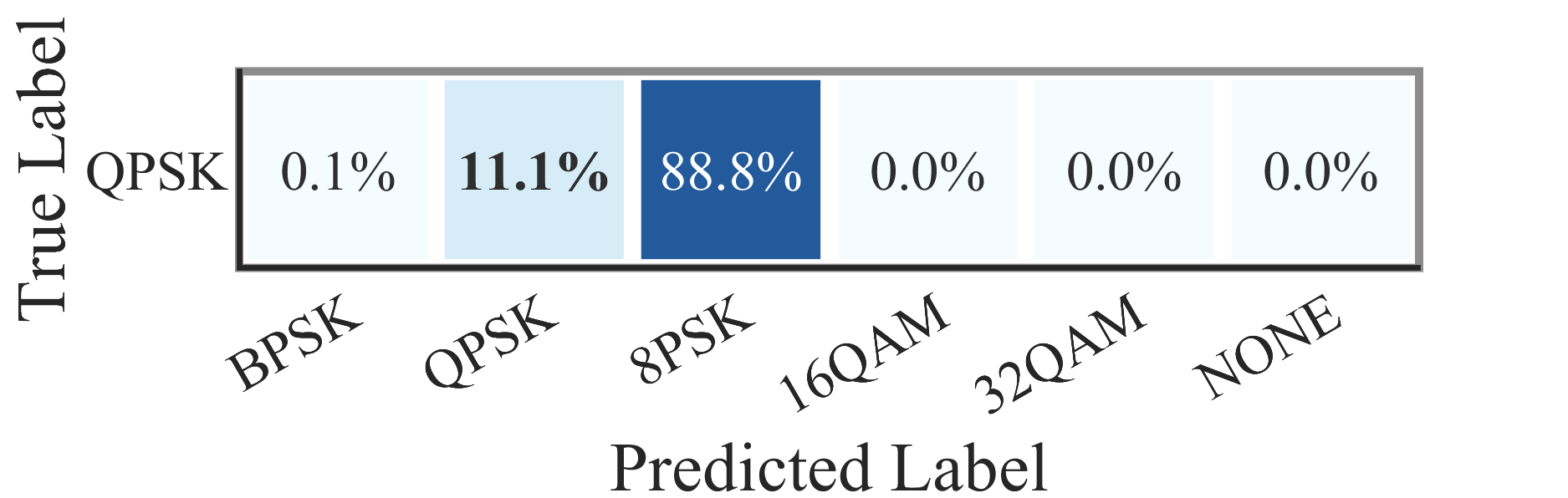}}
    \subfloat[QPSK-MSK: Branch 1-Unknown]{\includegraphics[width=0.3\linewidth]{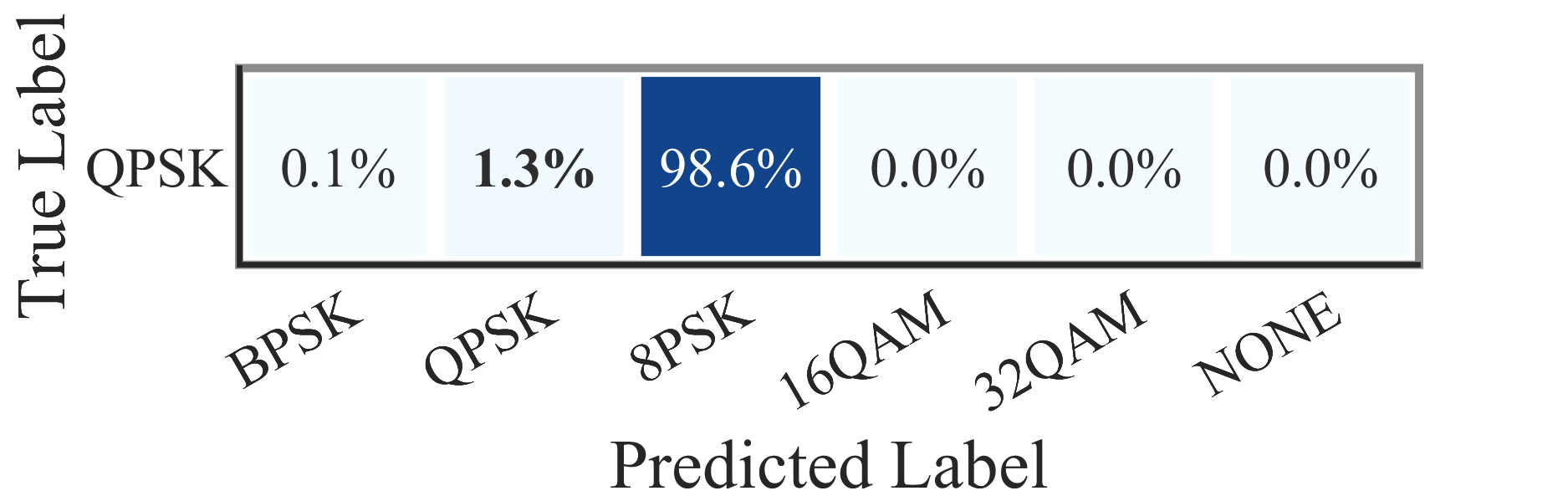}}
    \subfloat[32QAM-MSK: Branch 2-Unknown]{\includegraphics[width=0.35\linewidth]{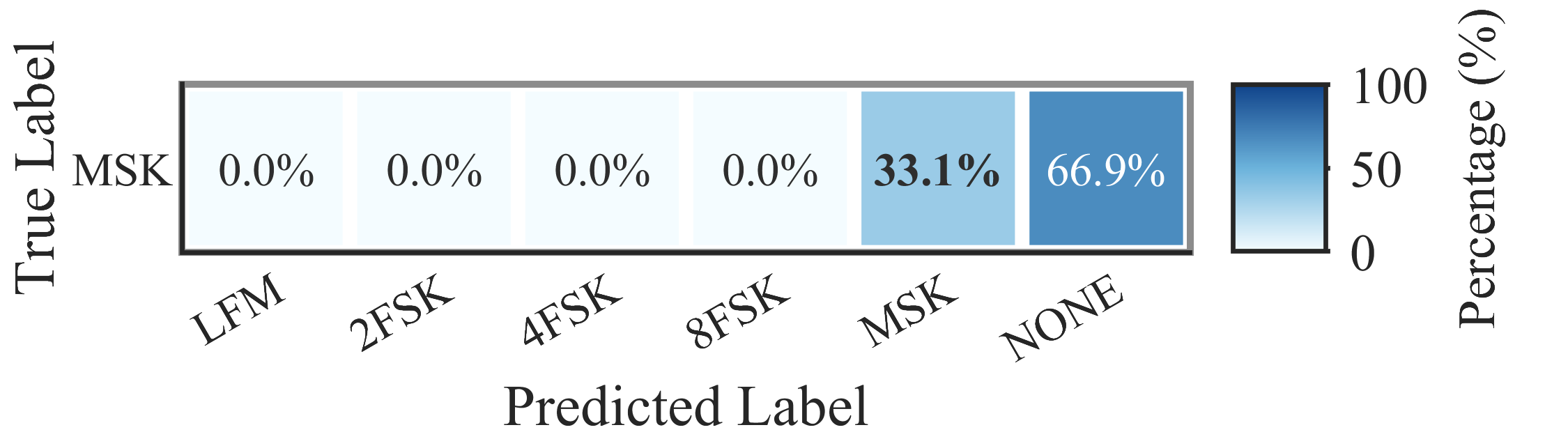}}
    \caption{Confusion matrices for the three most challenging held-out cases. (a)(d) Branch~1 when QPSK-LFM is unseen: QPSK-8PSK confusion. (b)(e) Branch~1 when QPSK-MSK is unseen: severe QPSK-8PSK confusion. (c)(f) Branch~2 when 32QAM-MSK
  is unseen: MSK-NONE confusion.}
    \label{fig:imbalance_3}
\end{figure*}

\subsection{Experiment 5: Robustness under Imperfect Factors}

\subsubsection{Robustness to AWGN}
\label{subsubsec:exp_awgn}

We first evaluate performance under additive white Gaussian noise to establish a noise-sensitivity baseline.
The SNR is varied from 0 dB to 16 dB in steps of 2 dB, using the one-CM-held-out protocol with BPSK-LFM as the unseen class.

Fig.~\ref{fig:awgn_robustness} presents the known-class and unknown-class accuracy as a function of SNR for backbone depths of one and two convolutional layers.
Two trends are evident.
First, the known-class accuracy is largely insensitive to the backbone depth, remaining above 90\% even at 0 dB for both configurations.
Second, the unknown-class accuracy exhibits a strong dependence on depth: the two-layer backbone outperforms the single-layer counterpart by approximately 23 percentage points across the entire SNR range.
This gap confirms that deeper networks extract semantic features that are more invariant to additive perturbations, which is particularly beneficial for generalizing to unseen combinations.
At high SNR $\geq 8$ dB, both configurations converge, indicating that the primary benefit of additional depth lies in noise resilience rather than inherent representational capacity.

\begin{figure}[!t]
    \centering
    \includegraphics[width=0.9\linewidth]{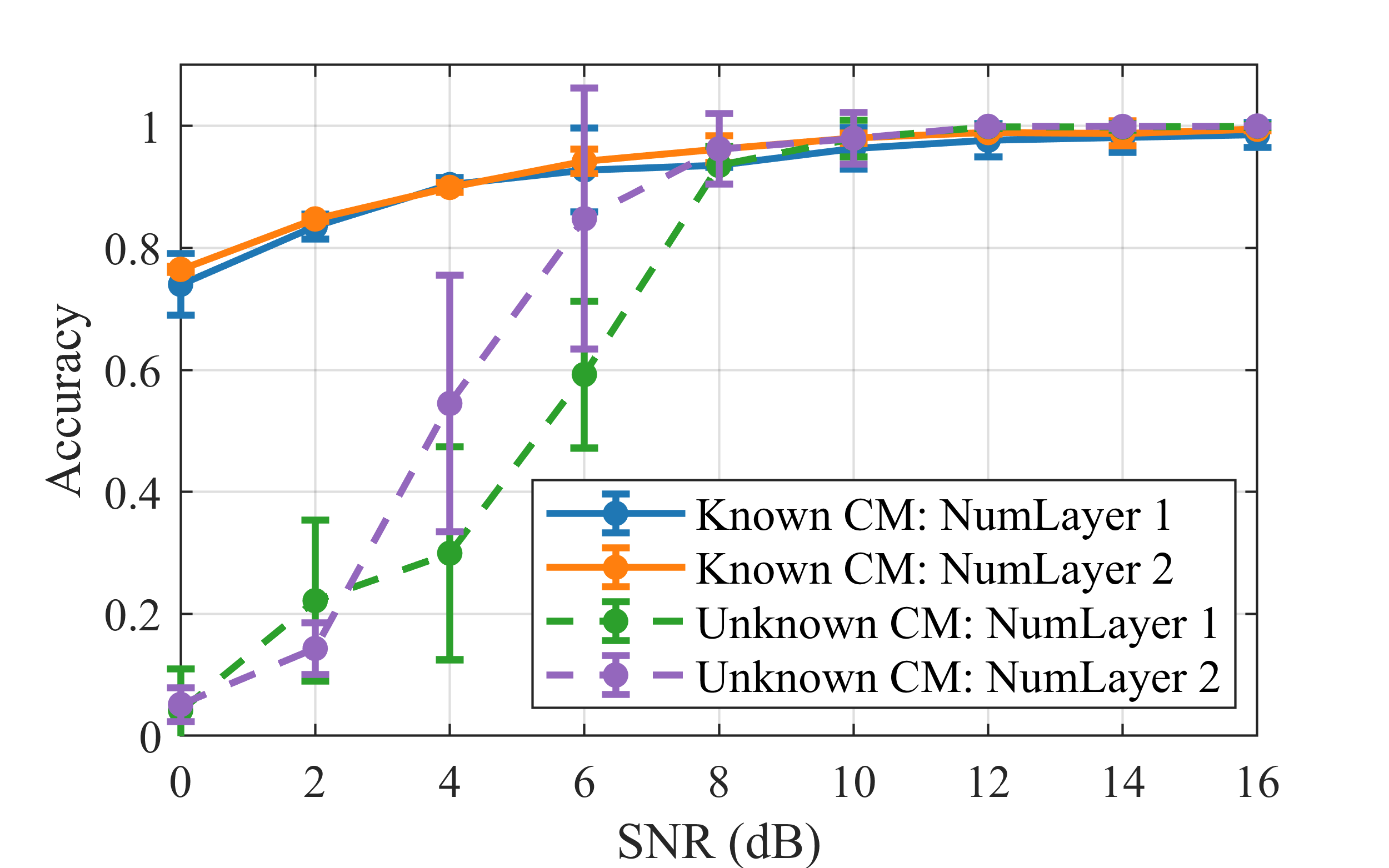}
    \caption{Known-class and unknown-class accuracy versus SNR under pure AWGN, comparing single-layer and two-layer convolutional backbones. BPSK-LFM is held out as the unseen class.}
    \label{fig:awgn_robustness}
\end{figure}

\subsubsection{Robustness under Combined Impairments}
\label{subsubsec:exp_combined}

We now evaluate performance under the full impairment model of \eqref{eq:full_model}, which integrates additive white Gaussian noise $w(n)$, multipath fading with one multipath component characterized by coefficients $a_2$ and $\tau_2$ with $a_1=1$ and $\tau_1=0$, 
PA nonlinearity $s_{\rm PA}(n) = g_1 x(n) + g_2 [x(n)]^2$ with $g_1 = 1$ and $g_i = 0$ for $i \ge 3$, oscillator phase noise with parameters $\lambda$ and $f_{\rm os}$, and IQ imbalance with parameters $\alpha$ and $\phi$. Three severity levels, namely \emph{clear}, \emph{mild}, and \emph{severe}, are defined, with all parameter values listed in Table~\ref{tab:impairment_params}. The SNR is fixed at 12 dB.

Fig.~\ref{fig:combined_robustness} shows the known-class and unknown-class accuracy at each severity level. Under mild impairment, the known-class accuracy is slightly lower than under severe impairment, but the unknown-class accuracy is higher under mild impairment than under severe impairment. This opposite trend arises from the interaction between impairment strength and feature learning. Under mild impairment, the signal quality is high enough for the network to extract fine-grained features that discriminate known classes. However, these features are affected by subtle distortion patterns in the training data and do not transfer well to unseen combinations. Under severe impairment, the dominant distortion masks these subtle patterns, forcing the network to rely on coarser but more robust features. These features generalize better across known classes but have reduced discriminative power for unseen ones. This observation reveals a fundamental trade-off in robust ACMR. Features optimized for known-class accuracy under mild degradation may not support zero-shot generalization, highlighting the need for explicit regularization toward compositional semantic structure.

\begin{table}[!t]
\centering
\caption{Parameter settings for the combined impairment evaluation.}
\label{tab:impairment_params}
\begin{tabular}{lcccccccc}
\toprule
\multirow{3}{*}{{Level}} & \multicolumn{2}{c}{{Multipath}} & \multicolumn{2}{c}{{PA}} & \multicolumn{2}{c}{{HFO}} & \multicolumn{2}{c}{{IQ}} \\
\cmidrule(lr){2-3} \cmidrule(lr){4-5}  \cmidrule(lr){6-7} \cmidrule(lr){8-9}
& $a_2$ & $\tau_2$ & $g_1$ & $g_2$ & $\lambda$ & $f_{\mathrm{os}}$ & $\alpha$ & $\phi$ \\
\midrule
Clear   & 0    & 0 & 1 & 0     & 0     & 0      & 0     & 0   \\
Mild    & 0.03 & 4 & 1 & 0.02  & 0.05  & 0.02   & 0.03  & 5   \\
Severe  & 0.05 & 8 & 1 & 0.03  & 0.10  & 0.05   & 0.05  & 10  \\
\bottomrule
\end{tabular}
\vspace{2pt}

\end{table}

\begin{figure}[!t]
    \centering
    \includegraphics[width=0.9\linewidth]{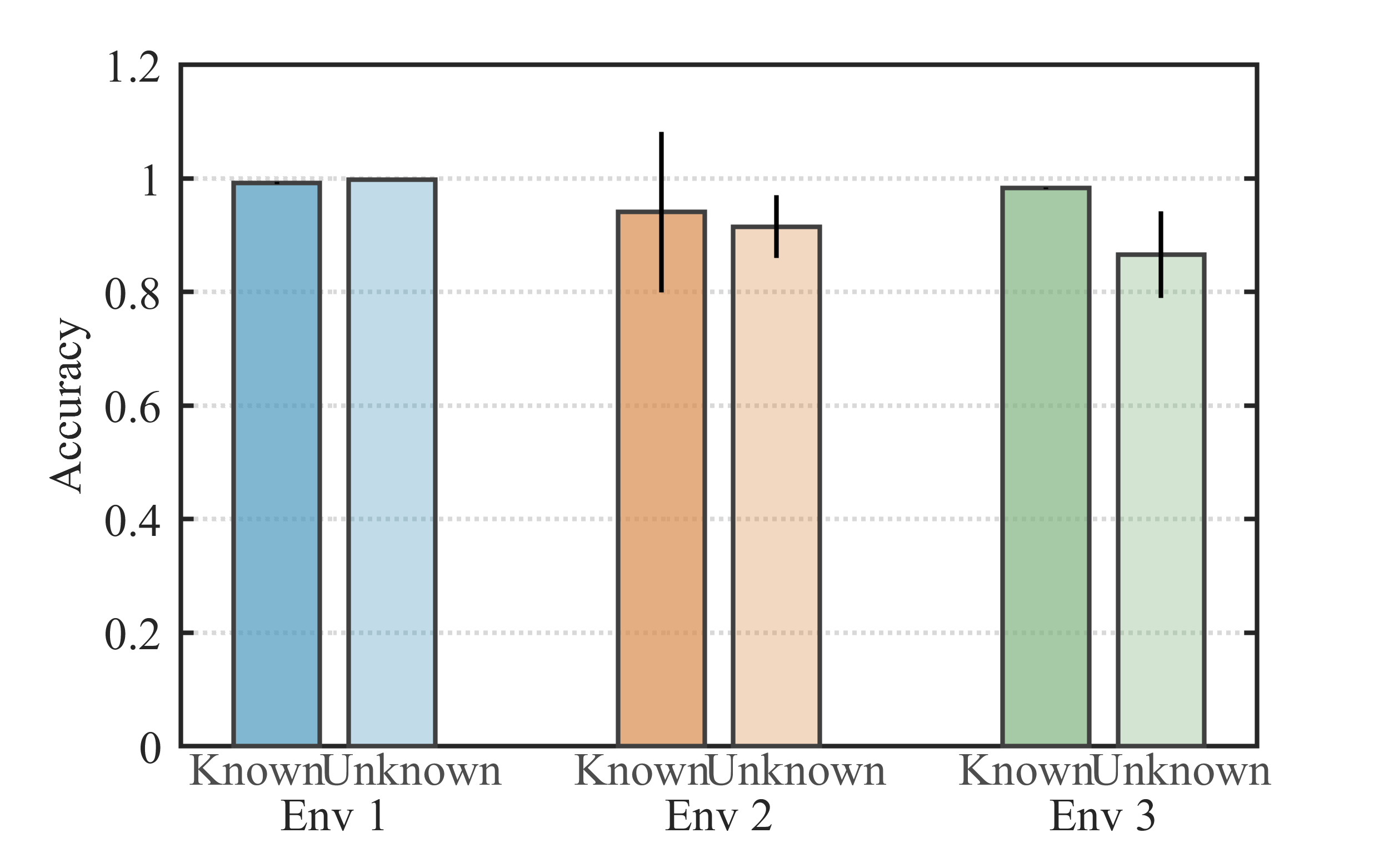}
    \caption{Known-class and unknown-class accuracy under combined channel and hardware impairments at three severity levels (clear, mild, severe). SNR = 12 dB.}
    \label{fig:combined_robustness}
\end{figure}

\color{black}

\section{Conclusion and Discussion}
This paper tackled the zero-shot ACMR problem in ISAC systems, where the combinatorial growth of inner-outer layer modulation pairs renders exhaustive supervised training impractical and existing methods fail to recognize unseen combinations.
To overcome this bottleneck, we first identified that the multiplicative waveform-level coupling between the two modulation layers can be globally linearized via a logarithmic projection onto the tangent space of the signal manifold. 
This inspired the design of a disentangled semantic space that factorizes the joint modulation label space into two independent layer-wise prototype sets.
Building on this geometric insight, we established a principled zero-shot framework, which comprises tangent-space mapping, learnable geometric disentanglement, and compositional semantic matching.
Then, we proved that input-dependent affine transformations are structurally sufficient to recover the individual modulation components from the additive tangent-space representation.
This theoretical framework was then instantiated as TSDN, which realizes the disentangling transform through a spatial transformer network and balances discrimination with cross-domain generalization via a multi-objective loss.

Extensive experiments confirmed that the proposed approach successfully generalizes to entirely unseen modulation combinations by leveraging the learned layer-wise semantic representations, substantially outperforming both unified-semantic and multi-task baselines while maintaining robust performance under realistic channel fading and hardware impairments.
More broadly, this work demonstrates that exploiting the intrinsic algebraic structure of composite waveforms opens a viable path toward recognition systems that scale gracefully with the expanding modulation landscape of next-generation wireless networks, without requiring retraining for every new combination.

Despite the promising results demonstrated by our proposed methodology, four critical challenges remain. These challenges merit careful consideration and suggest avenues for future research.
\begin{enumerate}
\color{black}
\item \textbf{Robustness Analysis and Future Research Avenues:}
In imperfect environments, especially when the SNR drops below 4 dB, the framework exhibits noticeable performance degradation. 
Increasing network depth can enhance the extraction of semantic features and slow this degradation. 
However, it does not provide a principled mechanism for extracting semantic invariants from heavily distorted signals. 
A fundamental solution calls for a shift toward physics-empowered machine learning \cite{10599118, mousa2026physics}. 
Embedding deterministic physical laws directly into the neural pipeline imposes strong structural constraints on the learning process, such as wave propagation invariants or explicit hardware distortion models. 
These physics-driven constraints act as a powerful regularizer, shrinking the feasible parameter space and guiding the feature extractor to preserve semantic integrity, thereby maintaining robustness even under extreme impairments.

\color{black}

\item \textbf{Limited Scope of Unknown Modulation Recognition:}
Our current framework assumes unknown CM signals are novel combinations of known individual modulation schemes. While this enables effective recognition within the defined parameter space, it limits applicability when either layer uses entirely unknown modulation schemes. A comprehensive knowledge database with an intelligent reasoning framework \cite{kojima2022large, xie2020region} can address this limitation, incorporating a hierarchical modulation taxonomy and using machine learning based inference to identify and classify unseen modulation schemes.

\item \textbf{Imbalanced Recognition Performance Across Modulation Combinations:}
Our analysis reveals non-uniform recognition performance across different CM signal combinations. This imbalance is due to training dataset bias, where some modulation schemes are more prevalent, causing feature extractors to perform better on well-represented combinations and worse on underrepresented pairings. To address this, we recommend cost-sensitive learning frameworks that account for class distribution disparities, making misclassification penalties inversely proportional to class frequency \cite{thabtah2020data}.

\color{black}
\item \textbf{Computational Complexity and Real-Time Implementation Considerations:}
The sequential nature of inner-layer and outer-layer modulation recognition, along with feature extraction complexity, may cause processing latencies that exceed acceptable thresholds for time-critical applications. Since the proposed method already uses a lightweight structure, future work should optimize computational efficiency via model pruning \cite{lin2020dynamic}, quantization \cite{kuzmin2023pruning}, and parallel processing \cite{bernstein1966analysis}.

\color{black}
    
\end{enumerate}

\bibliographystyle{IEEEtran}
\bibliography{reference.bib}

@article{yan2024automatic,
  title={Automatic composite-modulation classification using cyclic-paw-print features for cognitive aerospace communications},
  author={Yan, Xiao and Zhong, Xunuo and Wu, Hsiao-Chun and Yang, PengFei and Wang, Qian and Chen, Yiyun},
  journal={IEEE Transactions on Communications},
  volume={72},
  number={9},
  pages={5486--5502},
  year={2024},
  publisher={IEEE}
}

@inproceedings{kendall2018multi,
  title={Multi-Task Learning Using Uncertainty to Weigh Losses for Scene Geometry and Semantics},
  author={Alex Kendall and Yarin Gal and Roberto Cipolla},
  booktitle={Proceedings of the IEEE Conference on Computer Vision and Pattern Recognition (CVPR)},
  pages={7482--7491},
  year={2018}
}

@article{chang2022hierarchical,
  title={A hierarchical classification head based convolutional gated deep neural network for automatic modulation classification},
  author={Chang, Shuo and Zhang, Ruiyun and Ji, Kejia and Huang, Sai and Feng, Zhiyong},
  journal={IEEE Transactions on Wireless Communications},
  volume={21},
  number={10},
  pages={8713--8728},
  year={2022},
  publisher={IEEE}
}

@article{Zhao2025ZeroShotAM,
  title={Zero-Shot Automatic Modulation Recognition Using a Large Vision-Language Model},
  author={Yurui Zhao and Xiang Wang and Shuya Cao and Zhitao Huang},
  journal={IEEE Transactions on Communications},
  year={2025},
  volume={73},
  pages={15765-15782},
  url={https://api.semanticscholar.org/CorpusID:281447064}
}

@inproceedings{yang2017deep,
  title={Deep Multi-task Representation Learning: A Tensor Factorisation Approach},
  author={Yang, Yongxin and Hospedales, Timothy M.},
  booktitle={ICLR},
  year={2017}
}

@article{mousa2026physics,
  title={Physics-Integrated Inference for Signal Recovery in Non-Gaussian Regimes},
  author={Mousa, Mohamed A and Bauer, Leif and Yang, Ziyi and Singh, Utkarsh and Deka, Angshuman and Jacob, Zubin},
  journal={arXiv preprint arXiv:2601.18074},
  year={2026}
}

@inproceedings{wen2016discriminative,
  title={A Discriminative Feature Learning Approach for Deep Face Recognition},
  author={Wen, Yandong and Zhang, Kaipeng and Li, Zhifeng and Qiao, Yu},
  booktitle={ECCV},
  year={2016}
}

@ARTICLE{10599118,
  author={Zhu, Ethan and Sun, Haijian and Ji, Mingyue},
  journal={IEEE Wireless Communications}, 
  title={Physics-Informed Generalizable Wireless Channel Modeling with Segmentation and Deep Learning: Fundamentals, Methodologies, and Challenges}, 
  year={2024},
  volume={31},
  number={6},
  pages={170-177},
  doi={10.1109/MWC.015.2300603}}

@ARTICLE{sr2cnn,
  author={Dong, Yihong and Jiang, Xiaohan and Zhou, Huaji and Lin, Yun and Shi, Qingjiang},
  journal={IEEE Transactions on Signal Processing}, 
  title={SR2CNN: Zero-Shot Learning for Signal Recognition}, 
  year={2021},
  volume={69},
  number={},
  pages={2316-2329},
  doi={10.1109/TSP.2021.3070186}}

@article{van2008visualizing,
  title={Visualizing data using t-SNE.},
  author={Van der Maaten, Laurens and Hinton, Geoffrey},
  journal={Journal of machine learning research},
  volume={9},
  number={11},
  year={2008}
}

@inproceedings{liu2008specific,
  title={Specific emitter identification using nonlinear device estimation},
  author={Liu, Ming-Wei and Doherty, John F},
  booktitle={2008 IEEE Sarnoff Symposium},
  pages={1--5},
  year={2008},
  organization={IEEE}
}

@book{absil2008optimization,
  author    = "P.-A. Absil and R. Mahony and R. Sepulchre",
  title     = "Optimization Algorithms on Matrix Manifolds",
  publisher = "Princeton University Press",
  year      = 2008,
  isbn      = "978-0-691-13298-3"
}

@inproceedings{kim2002prediction,
  title={Prediction of Performance Loss Due to Phase Noise in Digital Satellite Communication System},
  author={Kim, Young-wan and Park, Dong-chul},
  booktitle={IST Mobile \& Wireless Telecommunications conf.},
  pages={575--578},
  year={2002},
  organization={Citeseer}
}

@inproceedings{NIPS2015_33ceb07b,
 author = {Jaderberg, Max and Simonyan, Karen and Zisserman, Andrew and kavukcuoglu, koray},
 booktitle = {Advances in Neural Information Processing Systems},
 editor = {C. Cortes and N. Lawrence and D. Lee and M. Sugiyama and R. Garnett},
 pages = {},
 publisher = {Curran Associates, Inc.},
 title = {Spatial Transformer Networks},
 volume = {28},
 year = {2015}
}

@article{zhao2022concentrate,
  title={Concentrate on hardware imperfection via aligning reconstructed states},
  author={Zhao, Yurui and Wang, Xiang and Huang, Zhitao},
  journal={IEEE Communications Letters},
  volume={26},
  number={12},
  pages={2934--2938},
  year={2022},
  publisher={IEEE}
}

@article{lei2024understanding,
  title={Understanding complex-valued transformer for modulation recognition},
  author={Lei, Jingreng and Li, Yang and Yung, Long-Yin and Leng, Yang and Lin, Qingfeng and Wu, Yik-Chung},
  journal={IEEE Wireless Communications Letters},
  volume={13},
  number={12},
  pages={3523--3527},
  year={2024},
  publisher={IEEE}
}

@article{lei2025unified,
  title={A Unified Distributed Algorithm for Hybrid Near-Far Field Activity Detection in Cell-Free Massive MIMO},
  author={Lei, Jingreng and Li, Yang and Wang, Ziyue and Lin, Qingfeng and Liu, Ya-Feng and Wu, Yik-Chung},
  journal={arXiv preprint arXiv:2509.15162},
  year={2025}
}

@inproceedings{huang2024multi,
  title={Multi-stage Learning for Radar Pulse Activity Segmentation},
  author={Huang, Zi and Pemasiri, Akila and Denman, Simon and Fookes, Clinton and Martin, Terrence},
  booktitle={ICASSP 2024-2024 IEEE International Conference on Acoustics, Speech and Signal Processing (ICASSP)},
  pages={7340--7344},
  year={2024},
  organization={IEEE}
}

@article{si2021intra,
  title={Intra-pulse modulation recognition of dual-component radar signals based on deep convolutional neural network},
  author={Si, Weijian and Wan, Chenxia and Deng, Zhian},
  journal={IEEE Communications Letters},
  volume={25},
  number={10},
  pages={3305--3309},
  year={2021},
  publisher={IEEE}
}

@article{o2018over,
  title={Over-the-air deep learning based radio signal classification},
  author={O’Shea, Timothy James and Roy, Tamoghna and Clancy, T Charles},
  journal={IEEE Journal of Selected Topics in Signal Processing},
  volume={12},
  number={1},
  pages={168--179},
  year={2018},
  publisher={IEEE}
}

@article{tekbiyik2019hisarmod,
  title={Hisarmod: A new challenging modulated signals dataset},
  author={Tekb{\i}y{\i}k, K and Ke{\c{c}}eci, C and Ekti, A and G{\"o}r{\c{c}}in, A and Kurt, G},
  journal={IEEE Dataport},
  year={2019}
}

@article{zhang2026cognitive,
  title={Cognitive Radio for Satellite TT \& C System: A General Dataset Using Software-defined Radio},
  author={Zhang, Yi and Zang, Bo and Ji, Hongbing and Li, Lin and Li, Shiyao and Chen, Leyan},
  journal={Scientific Data},
  year={2026},
  publisher={Nature Publishing Group UK London}
}

@article{pan2020automatic,
  title={Automatic waveform recognition of overlapping LPI radar signals based on multi-instance multi-label learning},
  author={Pan, Zesi and Wang, Shafei and Zhu, Mengtao and Li, Yunjie},
  journal={IEEE Signal Processing Letters},
  volume={27},
  pages={1275--1279},
  year={2020},
  publisher={IEEE}
}

@inproceedings{NIPS2016_45fbc6d3,
 author = {Bousmalis, Konstantinos and Trigeorgis, George and Silberman, Nathan and Krishnan, Dilip and Erhan, Dumitru},
 booktitle = {Advances in Neural Information Processing Systems},
 editor = {D. Lee and M. Sugiyama and U. Luxburg and I. Guyon and R. Garnett},
 pages = {},
 publisher = {Curran Associates, Inc.},
 title = {Domain Separation Networks},
 volume = {29},
 year = {2016}
}

@article{huan1995likelihood,
  title={Likelihood methods for MPSK modulation classification},
  author={Huan, Chung-Yu and Polydoros, Andreas},
  journal={IEEE Transactions on Communications},
  volume={43},
  number={2/3/4},
  pages={1493--1504},
  year={1995},
  publisher={IEEE}
}

@article{bkassiny2012survey,
  title={A survey on machine-learning techniques in cognitive radios},
  author={Bkassiny, Mario and Li, Yang and Jayaweera, Sudharman K},
  journal={IEEE Communications Surveys \& Tutorials},
  volume={15},
  number={3},
  pages={1136--1159},
  year={2012},
  publisher={IEEE}
}

@article{liu2020modulation,
  title={Modulation recognition with pre-denoising convolutional neural network},
  author={Liu, Yabo and Liu, Yi},
  journal={Electronics Letters},
  volume={56},
  number={5},
  pages={255--257},
  year={2020},
  publisher={Wiley Online Library}
}

@article{zhang2020automatic,
  title={Automatic modulation classification using CNN-LSTM based dual-stream structure},
  author={Zhang, Zufan and Luo, Hao and Wang, Chun and Gan, Chenquan and Xiang, Yong},
  journal={IEEE Transactions on Vehicular Technology},
  volume={69},
  number={11},
  pages={13521--13531},
  year={2020},
  publisher={IEEE}
}

@article{wang2021multidimensional,
  title={Multidimensional CNN-LSTM network for automatic modulation classification},
  author={Wang, Na and Liu, Yunxia and Ma, Liang and Yang, Yang and Wang, Hongjun},
  journal={Electronics},
  volume={10},
  number={14},
  pages={1649},
  year={2021},
  publisher={MDPI}
}

@article{lin2020dynamic,
  title={Dynamic model pruning with feedback},
  author={Lin, Tao and Stich, Sebastian U and Barba, Luis and Dmitriev, Daniil and Jaggi, Martin},
  journal={arXiv preprint arXiv:2006.07253},  
  year={2020}
}

@article{bernstein1966analysis,
  title={Analysis of programs for parallel processing},
  author={Bernstein, Arthur J},
  journal={IEEE transactions on electronic computers},
  number={5},
  pages={757--763},
  year={1966},
  publisher={IEEE}
}

@article{kuzmin2023pruning,
  title={Pruning vs quantization: Which is better?},
  author={Kuzmin, Andrey and Nagel, Markus and Van Baalen, Mart and Behboodi, Arash and Blankevoort, Tijmen},
  journal={Advances in neural information processing systems},
  volume={36},
  pages={62414--62427},
  year={2023}
}

@inproceedings{xie2020region,
  title={Region graph embedding network for zero-shot learning},
  author={Xie, Guo-Sen and Liu, Li and Zhu, Fan and Zhao, Fang and Zhang, Zheng and Yao, Yazhou and Qin, Jie and Shao, Ling},
  booktitle={European conference on computer vision},
  pages={562--580},
  year={2020},
  organization={Springer}
}

@article{kojima2022large,
  title={Large language models are zero-shot reasoners},
  author={Kojima, Takeshi and Gu, Shixiang Shane and Reid, Machel and Matsuo, Yutaka and Iwasawa, Yusuke},
  journal={Advances in neural information processing systems},
  volume={35},
  pages={22199--22213},
  year={2022}
}

@article{thabtah2020data,
  title={Data imbalance in classification: Experimental evaluation},
  author={Thabtah, Fadi and Hammoud, Suhel and Kamalov, Firuz and Gonsalves, Amanda},
  journal={Information Sciences},
  volume={513},
  pages={429--441},
  year={2020},
  publisher={Elsevier}
}

@inproceedings{mao2023cross,
  title={Cross-entropy loss functions: Theoretical analysis and applications},
  author={Mao, Anqi and Mohri, Mehryar and Zhong, Yutao},
  booktitle={International conference on Machine learning},
  pages={23803--23828},
  year={2023},
  organization={pmlr}
}

@inproceedings{sainath2015convolutional,
  title={Convolutional, long short-term memory, fully connected deep neural networks},
  author={Sainath, Tara N and Vinyals, Oriol and Senior, Andrew and Sak, Ha{\c{s}}im},
  booktitle={2015 IEEE international conference on acoustics, speech and signal processing (ICASSP)},
  pages={4580--4584},
  year={2015},
  organization={Ieee}
}

@article{rajendran2018deep,
  title={Deep learning models for wireless signal classification with distributed low-cost spectrum sensors},
  author={Rajendran, Sreeraj and Meert, Wannes and Giustiniano, Domenico and Lenders, Vincent and Pollin, Sofie},
  journal={IEEE Transactions on Cognitive Communications and Networking},
  volume={4},
  number={3},
  pages={433--445},
  year={2018},
  publisher={IEEE}
}

@inproceedings{west2017deep,
  title={Deep architectures for modulation recognition},
  author={West, Nathan E and O'shea, Tim},
  booktitle={2017 IEEE international symposium on dynamic spectrum access networks (DySPAN)},
  pages={1--6},
  year={2017},
  organization={IEEE}
}

@inproceedings{o2016convolutional,
  title={Convolutional radio modulation recognition networks},
  author={O’Shea, Timothy J and Corgan, Johnathan and Clancy, T Charles},
  booktitle={Engineering Applications of Neural Networks: 17th International Conference, EANN 2016, Aberdeen, UK, September 2-5, 2016, Proceedings 17},
  pages={213--226},
  year={2016},
  organization={Springer}
}

@article{yan2024automatic2,
  title={Automatic composite-modulation classification using ultra lightweight deep-learning network based on cyclic-paw-print},
  author={Yan, Xiao and Yang, Pengfei and Zhong, Xunuo and Wang, Qian and Wu, Hsiao-Chun and He, Ling},
  journal={IEEE Transactions on Cognitive Communications and Networking},
  volume={10},
  number={3},
  pages={866--879},
  year={2024},
  publisher={IEEE}
}

@inproceedings{liu2017deep,
  title={Deep neural network architectures for modulation classification},
  author={Liu, Xiaoyu and Yang, Diyu and El Gamal, Aly},
  booktitle={2017 51st Asilomar Conference on Signals, Systems, and Computers},
  pages={915--919},
  year={2017},
  organization={IEEE}
}

@article{hermawan2020cnn,
  title={CNN-based automatic modulation classification for beyond 5G communications},
  author={Hermawan, Ade Pitra and Ginanjar, Rizki Rivai and Kim, Dong-Seong and Lee, Jae-Min},
  journal={IEEE Communications Letters},
  volume={24},
  number={5},
  pages={1038--1041},
  year={2020},
  publisher={IEEE}
}

@inproceedings{hong2017automatic,
  title={Automatic modulation classification using recurrent neural networks},
  author={Hong, Dehua and Zhang, Zilong and Xu, Xiaodong},
  booktitle={2017 3rd IEEE international conference on computer and communications (ICCC)},
  pages={695--700},
  year={2017},
  organization={IEEE}
}

@article{ke2021real,
  title={Real-time radio technology and modulation classification via an LSTM auto-encoder},
  author={Ke, Ziqi and Vikalo, Haris},
  journal={IEEE Transactions on Wireless Communications},
  volume={21},
  number={1},
  pages={370--382},
  year={2021},
  publisher={IEEE}
}

@article{huynh2020mcnet,
  title={MCNet: An efficient CNN architecture for robust automatic modulation classification},
  author={Huynh-The, Thien and Hua, Cam-Hao and Pham, Quoc-Viet and Kim, Dong-Seong},
  journal={IEEE Communications Letters},
  volume={24},
  number={4},
  pages={811--815},
  year={2020},
  publisher={IEEE}
}

@inproceedings{tekbiyik2020robust,
  title={Robust and fast automatic modulation classification with CNN under multipath fading channels},
  author={Tekb{\i}y{\i}k, K{\"u}r{\c{s}}at and Ekti, Ali R{\i}za and G{\"o}r{\c{c}}in, Ali and Kurt, G{\"u}ne{\c{s}} Karabulut and Ke{\c{c}}eci, Cihat},
  booktitle={2020 IEEE 91st Vehicular Technology Conference (VTC2020-Spring)},
  pages={1--6},
  year={2020},
  organization={IEEE}
}

@article{yan2024efficient,
  title={Efficient automatic composite-modulation classifier using cyclic-paw-print features},
  author={Yan, Xiao and Chen, Yiyun and Zhong, Xunuo and Wu, Hsiao-Chun and Wang, Qian},
  journal={IEEE Communications Letters},
  volume={28},
  number={3},
  pages={652--656},
  year={2024},
  publisher={IEEE}
}

@inproceedings{li2021composite,
  title={Composite radar modulation identification by transfer learning},
  author={Li, Fang and Yang, Zixian and Huang, Bo and Chen, Yang},
  booktitle={2021 International Conference on Microwave and Millimeter Wave Technology (ICMMT)},
  pages={1--3},
  year={2021},
  organization={IEEE}
}

@article{consultative2009radio,
  title={Radio Frequency and Modulation Systems-part 1 Earth Stations and space craft},
  author={Consultative Committee for Space Data Systems (CCSDS and others},
  journal={Recommended standard CCSDS 401.0-B},
  year={2009}
}

@article{huang2025design,
  title={Design of frequency index modulated waveforms for integrated sar and communication on high-altitude platforms (HAPs)},
  author={Huang, Bang and Ahmed, Sajid and Alouini, Mohamed-Slim},
  journal={IEEE Transactions on Communications},
  year={2025},
  publisher={IEEE}
}

@article{zhou2022integrated,
  title={Integrated sensing and communication waveform design: A survey},
  author={Zhou, Wenxing and Zhang, Ruoyu and Chen, Guangyi and Wu, Wen},
  journal={IEEE Open Journal of the Communications Society},
  volume={3},
  pages={1930--1949},
  year={2022},
  publisher={IEEE}
}

@article{zheng2019radar,
  title={Radar and communication coexistence: An overview: A review of recent methods},
  author={Zheng, Le and Lops, Marco and Eldar, Yonina C and Wang, Xiaodong},
  journal={IEEE Signal Processing Magazine},
  volume={36},
  number={5},
  pages={85--99},
  year={2019},
  publisher={IEEE}
}

@article{zhang2021overview,
  title={An overview of signal processing techniques for joint communication and radar sensing},
  author={Zhang, J Andrew and Liu, Fan and Masouros, Christos and Heath, Robert W and Feng, Zhiyong and Zheng, Le and Petropulu, Athina},
  journal={IEEE Journal of Selected Topics in Signal Processing},
  volume={15},
  number={6},
  pages={1295--1315},
  year={2021},
  publisher={IEEE}
}

@article{liu2022survey,
  title={A survey on fundamental limits of integrated sensing and communication},
  author={Liu, An and Huang, Zhe and Li, Min and Wan, Yubo and Li, Wenrui and Han, Tony Xiao and Liu, Chenchen and Du, Rui and Tan, Danny Kai Pin and Lu, Jianmin and others},
  journal={IEEE Communications Surveys \& Tutorials},
  volume={24},
  number={2},
  pages={994--1034},
  year={2022},
  publisher={IEEE}
}

@article{luong2021radio,
  title={Radio resource management in joint radar and communication: A comprehensive survey},
  author={Luong, Nguyen Cong and Lu, Xiao and Hoang, Dinh Thai and Niyato, Dusit and Kim, Dong In},
  journal={IEEE Communications Surveys \& Tutorials},
  volume={23},
  number={2},
  pages={780--814},
  year={2021},
  publisher={IEEE}
}

@article{chiriyath2019novel,
  title={Novel radar waveform optimization for a cooperative radar-communications system},
  author={Chiriyath, Alex Rajan and Ragi, Shankarachary and Mittelmann, Hans D and Bliss, Daniel W},
  journal={IEEE Transactions on Aerospace and Electronic Systems},
  volume={55},
  number={3},
  pages={1160--1173},
  year={2019},
  publisher={IEEE}
}

@article{wang2025blind,
  title={Blind Recognition Algorithm of Multi-Carrier Composite Modulation Signal Based on Multi-Dimensional Time-Frequency Superimposed Spectrum},
  author={Wang, Shoubin and Li, Huan and Zhang, Xiaolong and Jiang, Hao and Shen, Lei},
  journal={Sensors},
  volume={25},
  number={13},
  pages={4007},
  year={2025},
  publisher={MDPI}
}

@inproceedings{wang2021modulation,
  title={Modulation recognition of composite signal based on ResNet and frequency domain graph},
  author={Wang, Lijun and Han, Yu and Ge, Hongfang and Zhao, Yongkuan and Shen, Lei},
  booktitle={Journal of Physics: Conference Series},
  volume={2025},
  number={1},
  pages={012029},
  year={2021},
  organization={IOP Publishing}
}

@inproceedings{lijun2021modulation,
  title={Modulation recognition of composite modulation signal based on two-fold digital receiver and goodness of fit test},
  author={Lijun, Wang and Yu, Han and Pan, Zhou and Hongfang, Ge and Lei, Shen},
  booktitle={2021 IEEE International Conference on Electronic Technology, Communication and Information (ICETCI)},
  pages={433--437},
  year={2021},
  organization={IEEE}
}





\end{document}